\documentclass[11pt]{article} \usepackage[active]{srcltx}
\usepackage[margin=1in]{geometry} \usepackage{setspace}
\usepackage{hyperref}\usepackage{float}
\usepackage{microtype,xspace,amssymb,amsmath,
  amsthm}\usepackage{graphicx, booktabs}
\usepackage{tikz} \usepackage{tikz-cd} \usepackage[all]{xy}
\usetikzlibrary{arrows,automata} \tikzset{>=stealth', shorten > =
  1pt,auto, ArrowNode/.style = {font=\small}}
 \usepackage{multicol} \usepackage{enumitem}
\usepackage{stmaryrd} \setlist[itemize]{label=$\circ$,leftmargin = *}
\setlist[enumerate]{label=$(\alph*)$, leftmargin = *}
\usepackage{xcolor}
\newtheorem{theorem}{Theorem}[section]
\newtheorem{corollary}[theorem]{Corollary}
\newtheorem{lemma}[theorem]{Lemma}
\newtheorem{proposition}[theorem]{Proposition}
\theoremstyle{definition}
\newtheorem{definition}[theorem]{Definition}
\newtheorem{example}[theorem]{Example}
\newtheorem{remark}[theorem]{Remark}

\newcommand{\bim}{{\sf{BiM}}}

\newcommand{\synt}[1]{{\sf Synt}(#1)}

\newcommand{\im}{{\sf Im }}
\newcommand{\cl}{{\sf Clopen}}
\newcommand{\viet}{{\sf Viet}}

\newcommand\x{{\bf x}}
\newcommand\y{{\bf y}}

\newcommand{\Z}{\mathbb Z}
\newcommand{\N}{\mathbb N}

\newcommand{\card}[1]{\left\lvert#1\right\rvert}
\newcommand{\vcirc}{\odot}

\newcommand{\cA}{\mathcal A}
\newcommand{\cB}{\mathcal B}
\newcommand{\cC}{\mathcal C}
\newcommand{\cD}{\mathcal D}

\newcommand{\cP}{\mathcal P}
\newcommand{\cV}{\mathcal V}
\newcommand{\cW}{\mathcal W}

\newcommand{\cQ}{\mathcal Q}
\newcommand{\cN}{\mathcal N}
\newcommand{\cS}{\mathcal S}

\newcommand{\cF}{\mathcal F}

\newcommand{\V}{{\bf V}}

\newcommand{\At}{{\rm At}}

\newcommand{\tR}{{\sf R}}
\newcommand{\tS}{{\sf S}}

\newcommand{\tP}{{\sf P}}
\newcommand{\tQ}{Q}

\renewcommand{\log}[4]{{#1}_{#2, #3}[#4]}
\newcommand{\just}[2]{\stackrel{#2}{#1}}
\renewcommand{\Upsilon}{\mathcal X}

\renewcommand{\tilde}{\widetilde}
\newcommand{\rvd}{{\bf R}_{\cV \circ \cD}}
\title{Substitution Principle and semidirect products\footnote{This
    project has received funding from the European Research Council
    (ERC) under the European Union's Horizon 2020 research and
    innovation program (grant agreement No.670624). The first-named
    author was also partially supported by the Centre for Mathematics
    of the University of Coimbra - UIDB/00324/2020, funded by the
    Portuguese Government through FCT/MCTES.}} \author{C\'elia Borlido\footnote{Centre
    for Mathematics of the University of Coimbra (CMUC),
    \mbox{3001-501}~Coimbra, Portugal, cborlido@mat.uc.pt}
  \and Mai Gehrke\footnote{Laboratoire J. A. Dieudonn\'e, CNRS, Universit\'e
  C\^ote d'Azur, France, mai.gehrke@unice.fr}}
\begin{document}\date{}
\maketitle
\begin{abstract}
  In the classical theory of regular languages the concept of
  recognition by profinite monoids is an important tool. Beyond
  regularity, Boolean spaces with internal monoids (BiMs) were
  recently proposed as a generalization.  On the other hand, fragments
  of logic defining regular languages can be studied inductively via
  the so-called ``Substitution Principle''. In this paper we make the
  logical underpinnings of this principle explicit and extend it to
  arbitrary languages using Stone duality. Subsequently we show how it
  can be used to obtain topo-algebraic recognizers for classes of
  languages defined by a wide class of first-order logic
  fragments. This naturally leads to a notion of semidirect product of
  BiMs extending the classical such construction for profinite
  monoids. Our main result is a generalization of Almeida and Weil's
  Decomposition Theorem for semidirect products from the profinite
  setting to that of BiMs. This is a crucial step in a program to
  extend the profinite methods of regular language theory to the
  setting of complexity theory.
\end{abstract} 
\section{Introduction}
Profinite monoids have proved to be a powerful tool in the theory of
regular languages. 
Eilenberg and Reiterman Theorems allow for the study of the so-called
\emph{varieties of regular languages} through the study of
topo-algebraic properties of suitable profinite monoids.
In~2008, Gehrke, Grigorieff and Pin~\cite{GehrkeGrigorieffPin08}
proposed a unified approach for the study of regular languages, using
Stone duality. In particular, they realized that profinite monoids
could be seen as the extended dual spaces of certain Boolean algebras
equipped with a residuation structure.
More generally, to a Boolean algebra of (possibly non-regular)
languages closed under quotients, one can assign a monoid equipped
with a uniformity for which the natural biaction of the monoid on
itself has uniformly continuous
components~\cite{GehrkeGrigorieffPin10}. A slight variant (a so-called
\emph{\bim{} - Boolean space with an internal monoid}) was identified
in~\cite{GehrkePetrisanReggio16}.

Complexity theory and the theory of regular languages are intimately
connected through logic. As with classes of regular languages, many
computational complexity classes have been given characterizations as
model classes of appropriate logic fragments on finite
words~\cite{Immerman99}.  For example, ${\rm AC}^0={\bf
  FO}[\text{arb}]$, ${\rm ACC}^0= ({\bf FO}+{\bf MOD})[{\text{arb}}]$,
and ${\rm TC}^0 = {\bf MAJ}[\text{arb}]$ where $\text{arb}$ is the set
of all possible (numerical) predicates of all arities on the positions
of a word, {\bf FO} is usual first-order logic, and {\bf MOD} and {\bf
  MAJ} stand for the \emph{modular} and \emph{majority} quantifiers,
respectively.  On the one hand, the presence of \emph{arbitrary
  predicates}, and on the other hand, the presence of the
\emph{majority quantifier} is what brings one far beyond the scope of
the profinite algebraic theory of regular languages.

Most results in the field of complexity theory are proved using
combinatorial and probabilistic, as well as algorithmic
methods~\cite{Williams14}. However, there are a few connections with
the topo-algebraic tools of the theory of regular languages.  A famous
result of Barrington, Compton, Straubing, and
Th\'erien~\cite{BarringtonComptonStraubingTherien92} states that a
regular language belongs to AC$^0$ if and only if its syntactic
homomorphism is quasi-aperiodic.  Although this result relies on the
lower-bound results in circuit complexity of~\cite{FurstSaxeSipser84}
and no purely algebraic proof is known, being able to characterize the
class of regular languages that are in AC$^0$ gives some hope that the
non-regular classes might be amenable to treatment by the generalized
topo-algebraic methods.

Indeed, the hope is that one can generalize the tools of algebraic
automata theory and this paper is a key step in this program.  In the
paper {\it Logic Meets Algebra: the case of regular
  languages}~\cite{TessonTherien07}, Tesson and Th\'erien lay out the
theory used to characterize logic classes in the setting of regular
languages in terms of their recognizers. Here we add topology to the
picture, using Stone duality, to obtain corresponding tools that apply
beyond the setting of regular languages. In particular, we generalize
the \emph{Substitution Principle} and the associated semidirect
product construction used in the study of logic on words for regular
languages to the general setting. This allows for an inductive
topo-algebraic study of fragments of logic, which differs sharply from
Lawvere's point of view of hyperdoctrines in categorical logic
(see~\cite{Lawvere06, Marques21}).  Apart from being a contribution in extending
profinite methods to complexity classes, this is also a key result in
unifying semantics and complexity and the results here and in the
original proceedings paper~\cite{BorlidoCzarnetzkiGehrkeKrebs17} were
key inspiration for this trend (see also~\cite{GehrkeJaklReggio2020}).

In the sequence of papers
\cite{GehrkePetrisanReggio16,GehrkePetrisanReggio17,GehrkePetrisanReggio21},
a semidirect product construction for \bim{}s corresponding to the
addition of a layer of quantifiers, which is closely related in spirit
to our results here, has been studied. While related, our results here
differ in three main points. For one, those results require the
lp-variety (that is, the input corresponding to the quantifier) to be
regular, and only the predicates or formula algebra need not be
regular. Secondly, the lp-variety needs to be generated by a finite
\emph{semiring}, while our results apply to \emph{any} lp-variety and
only the corresponding \emph{monoid} structure is used. Finally, their
approach is category theoretic and does not make the link with the
substitution principle of~\cite{TessonTherien07}.
When compared to~\cite{BorlidoCzarnetzkiGehrkeKrebs17}, where the link
with the substitution principle of~\cite{TessonTherien07} is
highlighted, a major difference is that
in~\cite{BorlidoCzarnetzkiGehrkeKrebs17} we mostly consider varieties
of finite algebraic structures (the so-called \emph{typed monoids}),
while here we put our emphasis on their projective limits.  There is
also a couple of differences in the terminology used, that we point
out along the text where relevant.

The paper is organized as follows. In Section~\ref{sec:prelim} we
introduce the background needed and set up the main notation used
throughout the paper. Important concepts here are those of an
\emph{lp-strain} of languages and of a \emph{class of sentences}.
Essentially, an \emph{lp-strain} of languages is an assignment to any
finite \emph{alphabet}~$A$ (i.e., a finite set) of a Boolean algebra
of \emph{languages} (i.e., a Boolean subalgebra of the powerset
$\cP(A^*)$, where $A^*$ denotes the $A$-generated free monoid), which
is closed under taking preimages under \emph{length-preserving
  homomorphisms}.
A \emph{class of sentences} is the corresponding notion in the setting
of logic on words.
Section~\ref{sec:profinite} is devoted to languages over profinite
alphabets.
These are essential to the main results of this paper.
Section~\ref{sec:bims} is an introduction
to recognition of Boolean algebras of (not necessarily regular)
languages that are \emph{closed under quotients}. We define what is a
\emph{Boolean space with an internal monoid}
({\bim})~\cite{GehrkeGrigorieffPin10, GehrkePetrisanReggio16}, and 
consider a slight generalization of this notion: that of a
\emph{\bim{}-stamp}. Informally, a \emph{\bim{}} is to a \emph{variety
  of languages} what a \emph{\bim{}-stamp} is to an \emph{lp-variety
  of languages} (i.e., an lp-strain in which we require \emph{closure
  under quotients}).
In Section~\ref{sec:substitution} we cast the \emph{Substitution
  Principle}~\cite{TessonTherien07} in duality theoretic terms, which
allows us to formulate a very general version that applies beyond
regularity.  We start by considering, as in Tesson and
Th\'erien~\cite{TessonTherien07}, the case of finite Boolean algebras
in Section~\ref{sec:sub-BA-fin}.  The key idea in this section is to
use letters of a finite alphabet to play the role of atoms of a finite
Boolean algebra of formulas. Since any Boolean algebra is the direct
limit of its finite subalgebras, its Stone dual space is the
projective limit of the sets of atoms of its finite subalgebras. This
is why \emph{profinite alphabets} (i.e. Boolean spaces) come into play
in a natural way in Section~\ref{sec:inf-BA}, where we study
substitution with respect to arbitrary Boolean algebras, and why the
concept of \emph{lp-strain} (and \emph{class of sentences}) defined
over finite but also profinite alphabets is crucial.
In~\cite{TessonTherien07} only formulas with at most one free variable
are considered since the extension to finite sets of free variables,
needed to study predicate logic, is straightforward.  For
completeness, we prove this in Section~\ref{sec:encode}.  Finally, in
Section~\ref{sec:app}, we illustrate the \emph{Substitution Principle}
in the setting of logic on words, by showing how it can be used to
understand the effect of applying a layer of quantifiers to a given
Boolean algebra of formulas.
The main result of the paper, Theorem~\ref{t:2}, is presented in
Section~\ref{sec:closing-quot}. In this section we get out of the
context of logic, and study the closure under quotients of the
operation on languages derived in Section~\ref{sec:substitution} via
duality. Theorem~\ref{t:2} is a generalization of the classical result
by Almeida and Weil~\cite[Theorem~5.1]{AlmeidaWeil1995}. Finally, in
Section~\ref{sec:examples}, we compute topo-algebraic recognizers for
certain Boolean algebras of quantified languages, which are built on
the recognizers for the corresponding non-quantified languages. More
precisely, in Section~\ref{sec:2} we study the effect of applying a
layer of existential quantifiers, in Section~\ref{sec:4} of modular
quantifiers, and in Section~\ref{sec:5} of majority quantifiers.

\section{Preliminaries}\label{sec:prelim}
In this section we briefly present the background needed in the rest
of the paper. For further reading on duality we refer
to~\cite{Johnstone86}, and for formal language theory and logic on
words to~\cite{Straubing94}.
\subsection{Discrete duality}\label{ss:discrdual}
This is the most basic duality we use and it provides a correspondence
between powerset Boolean algebras (these are the complete and atomic
Boolean algebras) and sets. Given such  Boolean algebra $\cB$, its
dual is its set of atoms (that is, the set of non-zero minimal
elements), denoted $\At(\cB)$ and, given a set $X$, its dual is the
Boolean algebra $\cP(X)$. Clearly going back and forth yields
isomorphic objects. If $h\colon\cA\to\cB$ preserves arbitrary meets
and joins, the dual of $h$, denoted $\At(h)\colon\At(\cB)\,{\to}\,
\At(\cA)$, is given by the adjunction:
\begin{equation*}
  \forall \ a\,{\in}\, \cA,\  x\,{\in}\,
  \At(\cB)   \quad\left(\At(h)(x)\leq a \iff x\leq h(a)\right).
\end{equation*}
For example, if $\iota\colon\cA\rightarrowtail\cP(X)$ is the inclusion
of a finite Boolean subalgebra of a powerset, then $\At(\iota)\colon
X\twoheadrightarrow\At(\cA)$ is the quotient map corresponding to the
finite partition of~$X$ given by the atoms of~$\cA$. Conversely, given
a function $f{\colon} X{\to} Y$, the dual is just
$f^{-1}{\colon}\cP(Y){\to}\cP(X)$.
\subsection{Stone duality}\label{ss:Stonedual}

Since it will be useful in the next section, we will present Stone
duality as a restriction of a dual adjunction between Boolean algebras
and topological spaces.

Generally Boolean algebras do not have enough atoms, and we have to
consider {\em ultrafilters} instead (which may be seen as `searches
downwards' for atoms). Given an arbitrary Boolean algebra~$\cB$, an
ultrafilter of~$\cB$ is a subset $\gamma$ of $\cB$ satisfying:
\begin{itemize}
\item $\gamma$ is an upset, i.e., $a\,{\in}\, \gamma$ and $a\leq b$
  implies $b\,{\in}\, \gamma$;
\item $\gamma$ is closed under finite meets, i.e.,
  $a,b\,{\in}\,\gamma$ implies $a\wedge b\,{\in}\, \gamma$;
\item for all $a\in\cB$ exactly one of $a$ and $\neg a$ is in
  $\gamma$.
\end{itemize}
We will denote the set of ultrafilters of $\cB$ by $X_\cB$, and we
will consider it as a topological space equipped with the topology
generated by the sets $\widehat{a}=\{\gamma\in X_\cB\mid a\in\gamma\}$
for $a\in\cB$. The last property in the definition of ultrafilters
implies that these basic open sets are also closed (and thus {\em
  clopen}). The resulting spaces are compact, Hausdorff, and have a
basis of clopen sets. Such spaces are called {\em Boolean
  spaces}. Given a homomorphism $h\colon\cA\to\cB$ between Boolean
algebras, the inverse image of an ultrafilter is an ultrafilter and
thus $h^{-1}$ induces a map $f:X_\cB\to X_\cA$. Since $f^{-1}(\widehat
a) = \widehat {h(a)}$ for every $a \in \cA$, this map is continuous.

Conversely, given \emph{any} topological space~$X$, its clopen subsets
form a Boolean algebra~$\cl(X)$, and given a continuous function
$f\colon X\to Y$, the inverse map $f^{-1}$ restricts and corestricts
to a homomorphism $\cl(Y) \to \cl(X)$ of Boolean algebras.

These assignments define a dual adjunction between Boolean algebras
and topological spaces. When seen as an adjunction between the
opposite category of Boolean algebras and that of topological spaces,
the counit at a Boolean algebra $\cB$, and unit at a topological space
$X$ are, respectively, given by
\[\varepsilon_\cB:\cB \to \cl(X_\cB), \qquad a \mapsto \widehat a\]
and
\begin{equation}
\eta_X: X \to X_{\cl(X)}, \qquad x \mapsto \{K \in \cl(X) \mid x \in
  K\}.\label{eq:29}
\end{equation}
One can show that $\varepsilon_\cB$ is allways an isomorphism, while
$\eta_X$ is an isomorphism if and only if $X$ is a Boolean
space. Thus, this adjunction restricts to a duality between Boolean
algebras and Boolean spaces: the so-called \emph{Stone duality}.

\subsection{Compactification of topological spaces}\label{ss:comp}
We consider two compactifications of topological spaces: the \emph{\v
  Cech-Stone} and the \emph{Banaschewski} compactifications, the
latter being defined only for zero-dimensional spaces. We also relate
them with the dual spaces of suitable Boolean algebras.

Let $X$ be a topological space. Then, the \emph{\v Cech-Stone
  compactification} of $X$ is the unique (up to homeomorphism) compact
Hausdorff space $\beta (X)$ together with a continuous function $e: X
\to \beta (X)$ satisfying the following universal property: every
continuous function $f: X \to Z$ into a compact Hausdorff space~$Z$
uniquely determines a continuous function $\beta f: \beta (X) \to Z$
making the following diagram commute:
\begin{center}
  \begin{tikzpicture}[node distance = 20mm]
    \node (X) at (0,0) {$X$}; \node[right of = X] (bX) {$\beta (X)$};
    \node[below of = bX, yshift = 5mm] (Z) {$Z$};
    \draw[->] (X) to node[above, ArrowNode] {$e$}(bX); \draw[->] (X) to
    node[left, yshift = -2mm, ArrowNode] {$f$} (Z); \draw[->, dashed] (bX) to
    node[right, ArrowNode] {$\beta f$} (Z);
  \end{tikzpicture}
\end{center}
\noindent In the case where $X$ is a completely regular $T_1$ space,
the map~$e$ is in fact an embedding~\cite[Section~19]{Willard70}. This
is for instance the case of discrete and of Boolean spaces.
It may be proved that, for a set $S$, the dual space of $\cP(S)$ is
the \emph{\v Cech-Stone compactification} of the discrete space~$S$.
This fact has been extensively used in the topo-algebraic approach to
non-regular languages over finite alphabets (see
e.g.~\cite{GehrkeGrigorieffPin10,GehrkePetrisanReggio16}).
However, as already mentioned, in the present work we also consider
languages over profinite alphabets. For that reason, we will be
handling dual spaces of Boolean algebras of the form $\cl(X)$, for
spaces~$X$ that are not discrete.  In the case where $X$ is
zero-dimensional, the dual space of $\cl(X)$ is known as the
\emph{Banaschewski compactification}~\cite{Banaschewski55,
  Banaschewski64} of~$X$ and it is denoted by~$\beta_0(X)$. At the
level of morphisms, $\beta_0$ assigns to each continuous map the dual
of the Boolean algebra homomorphism taking inverse image on
clopens. Moreover, for a zero-dimensional space~$X$, the
unit~\eqref{eq:29} of the dual adjunction between Boolean algebras and
topological spaces is an embedding $\eta_X: X \rightarrowtail
\beta_0(X)$ with dense image (see e.g. \cite[Section~4.7,
Proposition~(b)]{PorterWoods88}) and it is easy to see that for every
continuous function $f: X \to Y$ between zero-dimensional spaces the
following diagram commutes:
\begin{center}
  \begin{tikzpicture}[node distance = 20mm]
    \node (X) at (0,0) {$X$}; \node[right of = X] (bX) {$\beta_0
      (X)$}; \node[below of = bX, yshift = 5mm] (Z) {$\beta_0(Y)$};
    \node[below of = X, yshift = 5mm](Y) {$Y$};
    \draw[>->] (X) to node[above, ArrowNode] {$\eta_X$}(bX); \draw[->]
    (X) to node[left, ArrowNode] {$f$} (Y); \draw[->, dashed] (bX) to
    node[right, ArrowNode] {$\beta_0 f$} (Z); \draw[>->] (Y) to
    node[below, ArrowNode] {$\eta_Y$}(Z);
  \end{tikzpicture}
\end{center}

There are some cases where the \v Cech-Stone and the Banaschewski
compactifications coincide, e.g. for discrete spaces as mentioned
above. More generally, we have the following:

A space is said to have \emph{\v Cech-dimension zero} provided that
each finite open cover admits a finite clopen refinement. In
particular, every compact zero-dimensional space, and so, every
Boolean space, has \v Cech-dimension zero.

\begin{theorem}[{\cite[Theorem 4]{Banaschewski55}}] \label{t:1} A
  zero-dimensional space $X$ is of \v Cech-dimension zero if and only
  if it is normal and the \v Cech-Stone and Banaschewski
  compactifications of~$X$ coincide.
\end{theorem}

This is of interest to us as it shows that $\beta$ is given by duality
not only for discrete spaces but more generally for spaces of \v
Cech-dimension zero such as $Y^*$ for $Y$ a profinite alphabet
(cf. Lemma~\ref{l:10} and Theorem~\ref{t:4}).

\subsection{Projective and direct limits}\label{ss:lims}
A \emph{projective limit system} (also known as an \emph{inverse limit
  system}, or a \emph{cofiltered diagram}) $\cF$ of sets assigns to
each element $i$ of a \emph{directed} partially ordered set $I$, a set
$S_i$, and to each ordered pair $i\geq j$ in $I$, a map $f_{i,j}\colon
S_i \to S_j$ so that for all $i, j, k \in I$ with $i \ge j \ge k$ we
have $f_{i,i} = id_{S_i}$ and $f_{j,k}\circ f_{i,j} = f_{i,k}$.  The
\emph{projective limit} (or \emph{inverse limit} or \emph{cofiltered
  limit}) of $\cF$, denoted $\lim\limits_{\longleftarrow} \cF$, comes
equipped with projection maps $\pi_i: \lim\limits_{\longleftarrow} \cF
\to S_i$ compatible with the system. That is, for $i,j\in I$ with
$i\geq j$, $f_{i,j} \circ \pi_i= \pi_j$. Further, it satisfies the
following \emph{universal property}: whenever $\{\pi_i':S' \to
S_i\}_{i \in I}$ is a family of maps satisfying $f_{i,j} \circ \pi'_i
= \pi'_j$ for all $i\geq j$, there exists a unique map $g: S'\to
\lim\limits_{\longleftarrow} \cF$ satisfying $\pi_i' = \pi_i \circ g$,
for all $i \in I$.

The notion dual to projective limit, obtained by reversing the
directions of the maps, is that of \emph{direct limit} (also known as
an \emph{injective limit}, \emph{inductive limit} or \emph{filtered
  colimit}).

There are several things worth noting about these notions. First, the
projective limit of $\cF$ may be constructed as follows:
\begin{equation}
  \lim\limits_{\longleftarrow} \cF = \left\{(s_i)_{i \in I} \in
    \prod_{i \in I}S_i \mid f_{i,j} (s_i) = s_j \text{ whenever $i \ge
      j$}\right\}.\label{eq:4}
\end{equation}

Second, projective limits of finite sets, called \emph{profinite
  sets}, are equivalent to Boolean spaces. If each $S_i$ is finite,
then it is a Boolean space in the discrete topology, and the
projective limit is a closed subspace of the product and thus again a
Boolean space. Conversely, a Boolean space is the projective limit of
the projective system of its finite continuous quotients,
corresponding via duality to the fact that Boolean algebras are
locally finite and thus, direct limits of their finite subalgebras.

Third, one also has projective systems and projective limits of richer
structures than sets, being the connecting maps required to be
morphisms of the appropriate kind. In this work, we will use the
description of projective limits in the categories of topological
spaces and of monoids.  A very useful fact is that, in these settings,
projective limits are given as for sets with the obvious enriched
structure.
\subsection{Biactions and semidirect products of monoids}\label{ss:semidirect}

Let $S$ be a set and $M$ a monoid. We denote the identity of~$M$
by~$1$. A \emph{biaction} of~$M$ on~$S$ is a family of functions
\[\lambda_m: S \to S, \qquad \text{and}\qquad \rho_m: S \to S,\]
for each $m \in M$, that satisfy the following conditions:
\begin{itemize}
\item $\{\lambda_m\}_{m \in M}$ induces a left action of~$M$ on~$S$,
  that is, $\lambda_1$ the identity function on~$S$, and for every $m,
  m' \in M$, the equality $\lambda_m \circ \lambda_{m'} = \lambda_{m
    m'}$ holds;
\item $\{\rho_m\}_{m \in M}$ induces a right action of~$M$ on~$S$,
  that is, $\rho_1$ the identity function on~$S$, and for every $m, m'
  \in M$, the equality $\rho_{m'} \circ \rho_{m} = \rho_{mm'}$ holds;
\item $\{\lambda_m\}_{m \in M}$ and $\{\rho_m\}_{m \in M}$ are
  compatible, that is, for every $m, m' \in M$, we have $\lambda_m
  \circ \rho_{m'} = \rho_{m'} \circ \lambda_m$.
\end{itemize}

In the case where $S$ is also a monoid, we have the notion of a
\emph{monoid biaction}.  In order to improve readability, we shall
denote the operation on $S$ additively, although $S$ is not assumed to
be commutative. A \emph{monoid biaction of $M$ on $S$} is a biaction
of $M$ on the underlying set of $S$, satisfying the following
additional properties:
\begin{itemize}
\item $\lambda_{m_1} \circ \rho_{m_2} (s_1 + s_2) = \lambda_{m_1}
  \circ \rho_{m_2} (s_1 ) + \lambda_{m_1} \circ \rho_{m_2} (s_2)$, for
  all $m_1,m_2\in M$, and $s_1, s_2 \in S$;
\item $\lambda_{m_1} \circ \rho_{m_2} (0) = 0$, for all $m_1,m_2 \in M$.
\end{itemize}
Finally, given such a biaction, we may define a new monoid, called the
\emph{(two-sided) semidirect product} of~$S$ and~$M$, usually denoted
by $S{**}M$. The monoid $S{**}M$ has underlying set $S \times M$, and
its binary operation is defined by
\[(s_1, m_1)(s_2, m_2) = (\rho_{m_2}(s_1) + \lambda_{m_1}(s_2),\,
  m_1m_2).\]
\subsection{Formal languages}\label{ss:languages}
Let $A$ be a set, which we call an \emph{alphabet}. A \emph{word over
  $A$} is an element of the $A$-generated free monoid~$A^*$, and a
\emph{language} over $A$ is a set $L \subseteq A^*$ of words, that is,
an element of the powerset~$\cP(A^*)$. Whenever we write $w = a_0
\dots a_{n-1}$ for a word, we are assuming that each of the $a_i$'s is
a letter. If $w = a_0\dots a_{n-1}$ is a word, then we say that $w$
has \emph{length} $n$, denoted $\card w = n$. The fact that~$A^*$ is a
monoid means in particular that $A^*$ is equipped with a biaction of
itself. By discrete duality, it follows that $\cP(A^*)$ also is
equipped with a biaction of~$A^*$ given as follows. The \emph{left
  {\rm(respectively, } right{\rm )} quotient} of a language $L$ by a
word $w$ is the language $w^{-1}L = \{u \in A^* \mid wu \in L\}$
(respectively, $Lw^{-1} = \{u \in A^* \mid uw \in L\}$). We say that a
set of languages is \emph{closed under quotients} if it is invariant
under this biaction.

In the discrete duality between complete atomic Boolean algebras and
sets, the complete Boolean subalgebras closed under quotients of the
powerset $\cP(M)$, where $M$ is a monoid, correspond to the monoid
quotients of $M$~\cite[Theorem~3.26]{Gehrke16} (see
also~\cite[Theorem~8.13]{Dekkers2008}). In particular, given a
language $L \subseteq A^*$, the discrete dual of~$\cB_L$, the complete
Boolean subalgebra closed under quotients generated by $L$, is a
monoid quotient $\mu_L: A^*\twoheadrightarrow M_L$. The monoid $M_L$
is called the \emph{syntactic monoid of~$L$} and $\mu_L$ is the
\emph{syntactic homomorphism} of~$L$. The quotient map $\mu_L$ is
given by the corresponding congruence relation, known as the
\emph{syntactic congruence of~$L$}, which, by the definition of
$\mu_L$, is given by
\[
  u \sim_L v \iff (\forall x,y \in A^*, \ xuy \in L \iff xvy \in L).
\]
Notice that it follows that $L = \mu_L^{-1}(\mu_L[L])$, that is,
\emph{$L$ is recognized by~$M_L$ via~$\mu_L$}. More generally, given a
monoid $M$, a language $L \subseteq A^*$ is said to be
\emph{recognized} by $M$ provided there exists a homomorphism
$\mu\colon A^* \to M$ and a subset $P \subseteq M$ satisfying $L =
\mu^{-1}(P)$, or equivalently, provided $L = \mu^{-1}(\mu[L])$. By
duality, the set of languages recognized by~$\mu$ is a complete
Boolean subalgebra closed under quotients which contains~$L$. Since it
must contain~$\cB_L$, the syntactic homomorphism of~$L$ factors
through any homomorphism~$\mu$ which recognizes~$L$. That is, the
syntactic monoid and homomorphism of $L$ can be thought as the
``optimal'' recognizer of~$L$ .

For a set of languages $\cS \subseteq \cP(A^*)$, it is then natural to
wonder about the existence of an ``optimal'' recognizer for all the
languages of~$\cS$. This leads to the notion of \emph{syntactic
  monoid, homomorphism, and congruence of~$\cS$}, which are defined
via discrete duality for the complete Boolean subalgebra closed under
quotients generated by~$\cS$. Explicitly, the \emph{syntactic
  congruence of $\cS$}, denoted $\sim_\cS$, is the intersection of all
the syntactic congruences of the languages of~$\cS$, the
\emph{syntactic monoid} is the quotient $M_\cS = A^* / {\sim_\cS}$,
and the \emph{syntactic homomorphism} is the canonical projection
$\mu_\cS: A^* \twoheadrightarrow M_\cS$. The homomorphism $\mu_\cS$ is
``optimal'' in the sense that if $\mu: A^* \to M$ is a homomorphism
recognizing every language of $\cS$, then it factors through
$\mu_\cS$.

Finally, \emph{regular languages} are those recognized by
\emph{finite} monoids. It is not hard to see that regular languages
over an alphabet~$A$ form a Boolean algebra closed under quotients.
Beyond regularity, the discrete notion of recognition introduced here
is not adequate. In particular, any infinite monoid recognizes
uncountably many languages.
Recognition of Boolean algebras of (not necessarily regular) languages 
closed under quotients require the introduction of topology. This will be 
addressed in Section~\ref{sec:bims}.

So far, we did not require $A$ to be a finite set. 
In this paper, we will consider languages over profinite alphabets, 
where the topology of the alphabet will be taken into account (finite alphabets
are simply viewed as discrete spaces). Languages over profinite
alphabets will be the subject of Section~\ref{sec:profinite}.
Unless specified otherwise, we will use the letters $A, B, \dots$ to denote
finite alphabets, and $X, Y, \dots$ for profinite (not necessarily finite) ones.

\subsection{lp-strains and lp-varieties of languages}\label{ss:lp}
A homomorphism between free monoids is said to be
\emph{length-preserving} (also called an \emph{lp-morphism}) provided
it maps generators to generators. Therefore, lp-morphisms $A^* \to
B^*$ are in a bijection with functions $A \to B$. Given a map $h:A \to
B$, we use $h^*$ to denote the unique lp-morphism $h^*:A^* \to B^*$
whose restriction to $A$ is $h$. As we will see in
Section~\ref{ss:classessent}, lp-morphisms are important in the
treatment of fragments of logic on words.

Combining Pippenger~\cite{pippenger97} and
Straubing's~\cite{Straubing02} terminology, we call \emph{lp-strain of
  languages} an assignment, for each finite alphabet $A$, of a Boolean
algebra $\cV(A)$ of languages over $A$ such that, for every
lp-morphism $h^*:B^* \to A^*$, if $L \in \cV(A)$ then $(h^*)^{-1}(L)
\in \cV(B)$. Since $(h \circ g)^* = h^*\circ g^*$ whenever $f$ and $g$
are composable set functions, each $lp$-strain of languages defines a
presheaf $\cV: {\bf Set}_{fin}^{\rm op} \to {\bf BA}$ of Boolean
algebras. Moreover, the presheaves obtained in this way are precisely
the sub-presheaves of the presheaf $\cP: {\bf Set}_{fin}^{\rm op} \to
{\bf BA}$ defined by the lp-strain assigning the full powerset
$\cP(A^*)$ to each finite alphabet~$A$. This notion allows us to
extend languages to profinite alphabets
(cf. Section~\ref{sec:profinite}), which is crucial when handling
infinite Boolean algebras of formulas
(cf. Section~\ref{sec:inf-BA}). Notice that, if~$\cV$ is an lp-strain
of languages, then every function $h: B \to A$ induces dually a
continuous function $\widehat h: X_{\cV(B)} \to X_{\cV(A)}$ making the
following diagram commute:
\begin{center}
  \begin{tikzpicture}[node distance = 20mm]
    \node (B) at (0,0) {$B^*$}; \node[right of = B] (A)
    {$X_{\cV(B)}$}; \node[below of = B, yshift = 5mm] (VB) {$A^*$};
    \node[below of = A, yshift = 5mm] (VA) {$X_{\cV(A)}$};
    \draw[->] (B) to node[left, ArrowNode] {$h^*$} (VB); \draw[->] (B)
    to (A); \draw[->, dashed] (A) to node[right, ArrowNode] {$\widehat
      h$} (VA); \draw[->] (VB) to (VA);
  \end{tikzpicture}
\end{center}
where $A^* \to X_{\cV(A)}$ and $B^* \to X_{\cV(B)}$ are the maps with
dense image obtained by dualizing and restricting the embeddings
$\cV(A) \rightarrowtail \cP(A^*)$ and $\cV(B) \rightarrowtail
\cP(B^*)$, respectively.
This is a particular case of Lemma~\ref{l:5}, for which we will
provide a detailed proof.

In the case where $\cV(A)$ is closed under quotients for every
alphabet $A$, we call $\cV$ an \emph{lp-variety of languages}.
\subsection{Logic on words}\label{ss:low}
We shall consider classes of languages that are definable in fragments
of \emph{first-order logic}.
\emph{Variables} are denoted by $x,y,z, \ldots$, and we fix a finite
alphabet~$A$.
Our logic signature consists of:
\begin{itemize}
\item (unary) \emph{letter predicates} $\tP_a$, one for each letter $a
  \in A$;
\item a set $\cN$ of \emph{numerical predicates}~$R$, each of finite
  arity. A \emph{$k$-ary numerical predicate} $R$ is given by a subset
  $R \subseteq \N^k$. For instance, the usual (binary) numerical
  predicate $<$ formally corresponds to the predicate $R = \{(i,j)
  \mid i,j \in \N, \ i < j\}$;
\item a set $\cQ$ of (unary) \emph{quantifiers}~$Q$, each one being
  given by a function $Q:\{0,1\}^* \to \{0,1\}$. For instance, the
  existential quantifier $\exists$ corresponds to the function mapping
  $\varepsilon_0 \dots \varepsilon_{n-1} \in \{0,1\}^*$ to $1$ exactly
  when there exists an $i$ so that $\varepsilon_i = 1$.
\end{itemize}
Then, \emph{(first-order) formulas} are recursively built as
follows:
\begin{itemize}[leftmargin = *]
\item for each \emph{letter predicate}~$\tP_a$ and variable~$x$, there
  is an atomic formula $\tP_a(x)$;
\item if $R \in \cN$ is a $k$-ary numerical predicate and $x_1, \dots,
  x_k$ are (possibly non-distinct) variables, then $R(x_1, \dots,
  x_k)$ is an atomic formula;
\item \emph{Boolean combinations} of formulas are formulas, that is,
  if $\phi$ and $\psi$ are formulas, then so are $\phi \wedge \psi$,
  $\phi \vee \psi$, and $\neg \phi$;
\item if $Q$ is a quantifier, $x$ a variable and $\phi$ a formula,
  then $Q x \ \phi$ is also a formula.
\end{itemize}
An occurrence of a variable $x$ in a formula is said to be \emph{free} 
provided it appears outside of the scope of a quantifier $Q x$. 
A \emph{sentence} is a formula with no free occurrences of variables.

A \emph{context} $\x$ is a finite set of distinct variables. Given
disjoint contexts $\x$ and $\y$, we denote by~$\x\y$ the union of~$\x$
and~$\y$. Whenever we mention union of contexts, we will assume they
are disjoint without further mention. Also, we simply write $x$ to
refer to the context~$\{x\}$, so that by~$\x x$ we mean the context
$\x \cup \{x\}$. We denote by~$\card \x$ the cardinality of~$\x$.
The models of a formula will always be considered in a fixed
context. We say that $\phi$ \emph{is in context~$\x$} provided all the
free variables of $\phi$ belong to~$\x$. Notice that, if $\phi$ is in
context~$\x$, then it is also in every other context containing~$\x$.

Models of sentences in the empty context are words over $A$.
More generally, models of formulas in the context $\bf x$ are given by
\emph{$ \x$-marked words}, that is, words $w \in A^*$ equipped
with an interpretation~$ {\bf i} \in \{0, \dots, \card w -1\}^\x$ of
the context $\x$.
We identify maps from $\x$ to $\{0, \dots, \card w - 1\}$ with $\card
\x$-tuples ${\bf i} = (i_1, \dots, i_{\card \x }) \in \{0, \dots, \card
w -1\}^{\card \x} \subseteq \N^{\card \x}$. We denote the set of all
$ \x$-marked words by $A^* \otimes \N^{ \x}$.
Given a word $w \in A^*$ and a vector ${\bf i} = (i_1, \dots, i_{\card
  \x })\in \{0, \dots, \card w -1\}^{\card \x}$, we denote by $(w,
{\bf i})$ the marked word based on $w$ equipped with the map given by
${\bf i}$.
Moreover, if $\x$ and $\y$ are disjoint contexts, ${\bf i} = (i_1,
\dots, i_{\card \x}) \in \{0, \dots, \card w - 1\}^{\card \x}$ and
${\bf j} = (j_1, \dots, j_{\card \y}) \in \{0, \dots, \card w -
1\}^{\card \y}$, then $(w, {\bf i}, {\bf j})$ denotes the ${ {\bf
    z}}$-marked word $(w, {\bf k})$, where ${\bf z} = \x\y$ and ${\bf
  k} = (i_1, \dots, i_{\card \x}, j_1, \dots, j_{\card \y})$.

The semantics of a formula in context $\x = \{x_1, \dots, x_k\}$ is
defined inductively as follows. We let $w \in A^*$ be a word and ${\bf
  i} = (i_1, \dots, i_k) \in \{0, \dots, \card w - 1\}^{\card \x}$. Then,
\begin{itemize}
\item the marked word $(w, {\bf i})$ satisfies $\tP_{a}(x_j)$ if and
  only if its $i_j$-th letter is an $a$. As a consequence, for all $j
  \in \{1, \dots, k\}$, there exists exactly one letter $a \in A$ such
  that $(w, {\bf i}) \models \tP_a(x_j)$;
\item the marked word $(w, {\bf i})$ satisfies the formula $R(x_{j_1},
  \dots, x_{j_\ell})$ given by a numerical predicate~$R$ if and only
  if $(i_{j_1}, \dots, i_{j_\ell}) \in R$;
\item the Boolean connectives \emph{and} $\wedge$, \emph{or} $\vee$,
  and \emph{not} $\neg$ are interpreted classically;
\item the marked word $(w, {\bf i})$ satisfies the formula $\tQ x \
  \phi$, for a \emph{quantifier} $\tQ: \{0,1\}^* \to \{0,1\}$, if and
  only if $\tQ((w, {\bf i}, 0) \models \phi, \dots, (w, {\bf i}, \card
  w -1)\models \phi ) = 1$, where ``$(w, {\bf i}, i)\models \phi$''
  denotes the truth-value of ``$(w, {\bf i}, i)$ satisfies $\phi$''.
\end{itemize}

Important examples of quantifiers are given by the \emph{existential
  quantifier} $\exists$ mentioned above and its variant $\exists !$
(\emph{there exists a unique}); \emph{modular quantifiers}
$\exists_q^r$, for each $q \in \N$ and $0\leq r<q$, mapping a word of
$\{0,1\}^*$ to $1$ if and only if its number of $1$'s is congruent to
$r$ modulo $q$; and the \emph{majority quantifiers} ${\sf Maj}_k$, for
each $k \in \N$, sending an element of $\{0,1\}^*$ to $1$ if and only
if it has strictly $k$ more occurrences of~$1$ than of~$0$.
Finally, given a formula $\phi$ in context $\x$, we denote by
$L_{\phi}$ the set of marked words that are models of $\phi$.

Given a finite alphabet~$A$, a context~$\bf x$, a set of numerical
predicates $\cN$, and a set of quantifiers~$\cQ$, we denote by $\log
\cQ A {\bf x} \cN$ the corresponding set of first-order formulas in
context~$\bf x$, up to semantic equivalence. The set of sentences
$\log \cQ A \emptyset \cN$ in the empty context is simply denoted by
$\cQ_A[\cN]$.
Notice that, as long as $\log \cQ A \x \cN$ is non-empty, say it
contains a formula $\phi$, it will also contain the \emph{always-true}
formula ${\bf 1}$, which is semantically equivalent to $\phi \vee \neg
\phi$, and the \emph{always-false} formula ${\bf 0}$, which is
semantically equivalent to $\phi \wedge \neg \phi$. In particular, we
have $L_{{\bf 1}} = A^* \otimes \N^{ \x}$ and $L_{{\bf 0}} = \emptyset$, and
taking the set of models of a given formula defines an embedding of
Boolean algebras
\[\log \cQ A \x \cN \rightarrowtail \cP(A^* \otimes \N^{ \x}).\]

In formal language theory one usually considers languages that are
subsets of a free monoid, as presented in the previous
paragraph. However, formulas in a non-empty context $\x = \{x_1,
\dots, x_k\}$ define languages of marked words. Nevertheless, there is
a natural injection $A^* \otimes \N^{ \x} \rightarrowtail (A\times
2^\y)^*$ whenever $\y$ is a context containing~$\x$. Here, we make a
slight abuse of notation by writing $2^{\y}$ to actually mean the
set $\cP(\y)$. Indeed, let $(w, {\bf i})$ be a marked word, where $w =
a_0 \dots a_{n -1}$ and ${\bf i} = (i_1, \dots, i_k)$. Then, $(w, {\bf
  i})$ is completely determined by the word $(a_0, S_0)\dots (a_{n-1},
S_{n-1})$ over $(A \times 2^\y)$, where $S_i = \{x_{j} \mid i_j =
i\}$. Thus, we will often regard the language defined by a formula in
a context $\x$ as a language over an extended alphabet $A \times 2^\y$
for some $\y \supseteq \x$. Note that words of the form $(a_0,
S_0)\dots (a_{n-1}, S_{n-1})$ as above are usually called
\emph{$\x$-structures} (cf. \cite[Chapter~II]{Straubing94})
\subsection{Marked words}\label{ss:markedwords}
We now consider the case where~$\x$ consists of a single variable~$x$,
as this will play an important role in the remaining paper
(cf. Sections~\ref{sec:3} and~\ref{sec:closing-quot}). Note however
that a similar treatment is possible in general.

When $\x = \{x\}$, the set of $x$-marked words may be described as
\begin{equation*}
  A^* \otimes \N = \{(w, i) \mid w \in A^*, \ i \in \{0, \dots, \card
  w - 1\}\},
\end{equation*}
and we have an embedding $A^*\otimes \N \rightarrowtail (A \times
2)^*$ defined by
\[(w, i) \mapsto (a_0, 0) \dots (a_{i-1}, 0) (a_i, 1) (a_{i+1}, 0)
  \dots (a_{n-1}, 0),\]
for every word $w = a_0 \dots a_{n-1}$ and $i \in \{0, \dots, n-1\}$.

Identifying the elements of $A^* \otimes\N$ with their images in $(A
\times 2)^*$, we may compute the Boolean algebra closed under
quotients generated by this language. It is not difficult to see that
$A^* \otimes\N$ is a regular language in $(A \times 2)^*$ and that its
syntactic monoid is the three-element commutative monoid $\{e,m,z\}$
satisfying $m^2 = z$, with~$z$ acting as zero and~$e$ its identity
element. The syntactic morphism of the language $A^* \otimes\N$ is
$\mu: (A \times 2)^*\twoheadrightarrow \{e,m,z\}$ given by $\mu(a,0) =
e$ and $\mu(a, 1) = m$, and we have
\[A^* = \mu^{-1}(e), \quad A^*\otimes \N = \mu^{-1}(m), \quad
  \text{and} \quad A_z = \mu^{-1}(z),\]
where we identify $A^*$ with $(A\times\{0\})^*$ and $A_z= (A \times
2)^* \setminus (A^* \cup A^* \otimes \N)$.

\subsection{Classes of sentences}\label{ss:classessent}
Given a function $\zeta: A \to B$, we may define a map $\zeta^*_\x:
A^* \otimes \N^{ \x} \to B^* \otimes\N^ \x$ by setting $\zeta^*_\x(w,
{\bf i}) = (\zeta^*(w), {\bf i})$, for every $w \in A^*$ and ${\bf i}
\in \{0, \dots, \card w - 1\}^\x$. By discrete duality, $\zeta_\x^*$
yields a homomorphism of Boolean algebras $(\zeta^*_\x)^{-1}: \cP(B^*
\otimes \N^{ \x}) \to \cP(A^* \otimes \N^{ \x})$. By identifying
fragments of logic with the set of languages they define, we may show
that, for every set of quantifiers~$\cQ$ and every set of numerical
predicates~$\cN$, the homomorphism $(\zeta^*_\x)^{-1}$ restricts and
co-restricts to a homomorphism $\log \cQ B \x \cN \to \log \cQ A \x
\cN$. Indeed, for $\x$ consisting of a single variable~$x$, we have
\[\zeta_x(w, i) \models \tP_b(x) \iff (w, i) \models
  \bigvee_{\zeta(a) = b} \tP_a(x),\]
for every $w \in A^*$ and $i \in \{0, \dots, \card w - 1\}$. Note
that, since $A$ is finite, $\bigvee_{\zeta(a) = b} \tP_a(x)$ is a
well-defined formula of $\log \cQ A \x \cN$. With a routine structural
induction on the construction of formulas we can then show that
$(\zeta^*_\x)^{-1}$ sends the language defined by a formula
$\phi\in\log \cQ B \x \cN$ to the language defined by the formula
obtained from~$\phi$ by substituting, for every occurrence of a
predicate $\tP_b(x)$ with $b\in B$, the formula
$\bigvee_{\zeta(a)=b}\tP_a(x)$.  Note that if $b$ is not in the image
of~$\zeta$, then $\tP_b(x)$ is replaced by the empty join, which is
logically equivalent to the \emph{always-false} proposition.
We denote by $\zeta_{\cQ_\x[\cN]}$ this restriction and co-restriction
of $(\zeta_\x^*)^{-1}$, so that the following diagram commutes:
\begin{equation}
  \begin{aligned}
    \begin{tikzpicture}
      [node distance = 20mm, ->] \node at (0,0) (gc1) {$ \log \cQ B \x
        \cN$}; \node[below of = gc1] (gc2) {$\log \cQ A \x \cN$};
      \node[right of = gc1, xshift = 10mm] (pc1) {$\cP(B^*\otimes
        \N^{ \x})$}; \node[below of = pc1] (pc2) {$\cP(A^*
        \otimes \N^{ \x})$};
      \draw[->] (pc1) to node[right, ArrowNode] {$(\zeta_\x^*)^{-1}$}
      (pc2); \draw[>->] (gc1) to (pc1); \draw[>->] (gc2) to (pc2);
      \draw[->] (gc1) to node[left, ArrowNode] {$\zeta_{\cQ_\x[\cN]}$}
      (gc2);
    \end{tikzpicture}
    \label{eq:6}
  \end{aligned}
\end{equation}
In particular, $\cQ_{\_,\x}[\cN]$ defines a presheaf
$\cQ_{\_,\x}[\cN]: {\bf Set}_{fin}^{\rm op} \to {\sf BA}$ and, as so,
it takes right (respectively, left) inverses to left (respectively,
right) inverses and thus surjections (respectively, injections) on
finite sets to embeddings (respectively, quotients) in Boolean
algebras.

A \emph{class of sentences} is, intuitively speaking, the notion
corresponding to that of an lp-strain of languages in the setting of
logic on words.  This notion models what is usually referred to as a
\emph{fragment of logic}.
Formally, a \emph{class of sentences} $\Gamma$ is a map that
associates to each finite alphabet~$A$ a set of sentences
$\Gamma(A)\subseteq\cQ_{A}[\cN]$ which satisfies the following
properties:
\begin{enumerate}[label = (LC.\arabic*), leftmargin = *]
\item\label{item:LC1} Each set $\Gamma(A)$ is closed under Boolean
  connectives $\wedge$ and $\neg$ (and thus, under $\vee$);
\item\label{item:LC2} For each map $\zeta\colon A\to B$ between finite
  alphabets, the homomorphism $\zeta_{\cQ[\cN]}: \cQ_B[\cN] \to
  \cQ_A[\cN]$ restricts and co-restricts to a homomorphism
  $\zeta_{\Gamma}: \Gamma(B)\to \Gamma(A)$.
\end{enumerate}
Since we consider sentences up to semantic equivalence,
by~\ref{item:LC1}, each $\Gamma(A)$ is a Boolean algebra, and
~\ref{item:LC2} simply says that the assignment $A \mapsto
\cV_\Gamma(A) = \{L_\phi\mid\phi\in\Gamma(A)\}$ defines an lp-strain
of languages.

\section{Languages over profinite alphabets}\label{sec:profinite}
As in the regular setting, it is sometimes useful to consider
languages over profinite alphabets. Indeed, we will see in
Section~\ref{sec:substitution} that these appear naturally when
extending \emph{substitution} to arbitrary Boolean algebras.
In this section we see how any lp-strain of languages may be naturally
extended to profinite alphabets. Since lp-strains of languages are the
formal language counterpart of classes of sentences, such extension
will be another key ingredient in the subsequent sections
(cf. Corollary~\ref{cor:SP-global}, but also Theorem~\ref{t:2}).

Let $\cV$ be an lp-strain of languages. We extend the assignment $A
\mapsto \cV(A)$ for~$A$ finite to profinite alphabets as follows. Let
$Y$ be a profinite alphabet. By definition, $Y$ is the projective
limit of all its finite continuous quotients, say $\{h_i : Y
\twoheadrightarrow Y_i\}_{i \in I}$, the connecting morphisms being
all the surjective maps $h_{i,j}: Y_i \twoheadrightarrow Y_j$
satisfying $h_j = h_{i,j} \circ h_i$. Note that $I$ is a directed set
ordered by $i \ge j$ if and only if $h_{i,j}$ is defined.
In turn, each of the maps $h_{i,j}$ uniquely defines an lp-morphism
$h_{i,j}^*: Y_i^* \twoheadrightarrow Y_j^*$. Since $\cV$ is an
lp-strain, for every $i \ge j$, there is a well-defined embedding of
Boolean algebras $\theta_{i,j}: \cV(Y_j) \rightarrowtail \cV(Y_i)$
sending $L \in \cV(Y_j)$ to its preimage $(h_{i,j}^*)^{-1}(L) \in
\cV(Y_j)$.
A similar phenomenon happens with respect to each of the continuous
quotients $h_i: Y \twoheadrightarrow Y_i$, when $Y^*$ is viewed as the
topological space which is the union over $n \ge 0$ of the product
spaces~$Y^n$.   Indeed, since the clopen
subsets of $Y^*$ are the unions of the form $\bigcup_{n \ge 0} C_n$,
where $C_n \subseteq Y^n$ is a clopen subset of $Y^n$, we have that
$h_i^*$ is continuous and thus, it defines an embedding of Boolean
algebras $\iota_i: \cV(Y_i) \rightarrowtail \cl(Y^*)$.
Clearly, the family $\{ \cV(Y_i) \mid i \in I\}$ forms a direct limit
system with connecting morphisms $\{\theta_{i,j}\mid i \le j\}$
satisfying $\iota_j = \iota_i \circ \theta_{i,j}$.
Thus, we may define
\begin{equation}
  \label{eq:3}
  \cV(Y) = \lim_{\longrightarrow} \ \{\cV(A) \mid A \text{ is a
    finite continuous  quotient of }Y\}.
\end{equation}
Note that, by identifying $\cV(Y_i)$ with $\iota_i[\cV(Y_i)] \subseteq
\cl(Y^*)$ for each $i \in I$, we may regard $\cV(Y)$ as a Boolean
subalgebra of $\cl({Y^*}) \subseteq \cP(Y^*)$. More precisely, we have
\[\cV(Y) = \bigcup_{i \in I} \cV(Y_i),\]
and each $\cV(Y_i)$ is a Boolean subalgebra of $\cV(Y)$ (which dually
means that~$X_{\cV(Y_i)}$ is a quotient of~$X_{\cV(Y)}$).
Moreover, given a map $h: Y \to A$, words $u, v\in Y^*$, and a
language $K \subseteq A^*$, a routine computation shows that
\[u^{-1}[(h^*)^{-1}(K)]v^{-1}= (h^*)^{-1}(s^{-1}Kt^{-1}),\]
where $s = h^*(u)$ and $t = h^*(v)$. Thus, if $\cV(A)$ is closed under
quotients, then so is $(h^*)^{-1}(\cV(A))$.
Therefore, we have:
\begin{lemma}\label{l:4}
  A language $L \subseteq Y^*$ belongs to $\cV(Y)$ if and only if
  there is a finite continuous quotient $h: Y \twoheadrightarrow A$
  and a language $K \in \cV(A)$ such that $L =
  (h^*)^{-1}(K)$. Further, if $\cV$ is an lp-variety, then $\cV(Y)$ is
  a Boolean algebra closed under quotients by words of~$Y^*$.
\end{lemma}

\begin{remark}
  As we have seen above, every language over a profinite alphabet may
  be seen as a clopen subset of~$Y^*$. However, even when~$\cV$ is the
  full variety of languages, meaning that $\cV(A) = \cP(A^*)$ for
  every finite alphabet~$A$, the equality $\cV(Y) = \cl(Y^*)$ does not
  hold in general. For instance, if $Y = \N \cup \{\infty\}$ is the
  one-point compactification of~$\N$, we have $Y =
  \lim\limits_{\longleftarrow} \ \{h_i: \N \twoheadrightarrow \{0, 1,
  \dots, i\}\}_{i \in \N}$, where $h_i(j) = j$ for $j \in \{0, 1,
  \dots, i\}$, and $h_i(j) = i$ for $j \in \N_{> i} \cup \{\infty\}$.
  Moreover, the set $L = \{0 1 2\dots i \mid i \in \N\}$ is a clopen
  subset of~$Y^*$ which does not belong to $\cV(Y)$. Indeed, by
  Lemma~\ref{l:4}, this is due to the fact that there is no finite
  continuous quotient~$h:Y \twoheadrightarrow A$ of~$Y$ so that $L$
  belongs to $(h^*)^{-1}(\cP(A^*))$, as the preimage under~$h$ of
  every language over~$A$ necessarily contains some word with the
  letter~$\infty$.
\end{remark}

An interesting open question is then to provide a nice description of
the dual space of $\cV(Y)$ for the full variety~$\cV$.  This would
provide the ``best'' ambient space for studying languages over a
profinite alphabet. Lacking such a description, we will define a
\emph{language over a profinite alphabet~$Y$} to be a clopen subset
of~$Y^*$. Note that we have the following:

\begin{proposition}\label{prop:maxlp}
  Let $Y$ be a profinite set and $L\subseteq Y^*$. Then the following
  conditions are equivalent:
\begin{enumerate}
\item There exists $q\colon Y\twoheadrightarrow A$ finite continuous
  quotient with $L=(q^*)^{-1}(q^*[L])$;
\item There exist $V_0,\dots, V_{n-1}$ a clopen partition of $Y$ such
  that, if $w\in L$, then there is $f:|w|\to n$ with
\[
  w\in V_{f(0)}\dots V_{f(|w| -1)}\subseteq L;
\]
\item There exist $V_0,\dots, V_{n -1}$ clopen subsets of $Y$ such
  that, if $w\in L$, then there is $f:|w|\to n$ with
\[
w\in V_{f(0)}\dots V_{f(|w| - 1)}\subseteq L.
\]
\end{enumerate} 
\end{proposition}

\begin{proof}
  $(a){\Rightarrow}(b)$: Fix an enumeration of
  $A=\{a_0,\dots,a_{n-1}\}$ and take $V_i=q^{-1}(a_i)$. Then
  $V_0,\dots, V_{n-1}$ is a clopen partition of $Y$. Also, for any $w\in
  L$
\[
  w\in (q^*)^{-1}(q^*(w))=q^{-1}(q(w_0))\dots q^{-1}(q(w_{|w|
    -1}))\subseteq (q^*)^{-1}(q^*[L])= L,
\]
where the last equality holds by hypothesis. Now $(b)$ follows as
$q^{-1}(q(w_0))\dots q^{-1}(q(w_{|w| - 1}))=V_{f(0)}\dots V_{f(|w|
  -1)}$, where $f(j)=i$ if and only if $w_j=a_i$. Clearly
$(b){\Rightarrow}(c)$.

For $(c){\Rightarrow}(a)$: Let $W_0,\dots, W_{m-1}$ be the atoms of the
finite Boolean algebra of clopens generated by the clopen subsets
$V_0,\dots, V_{n-1}$ stipulated in $(c)$ and let $q\colon
Y\twoheadrightarrow m$ be the corresponding quotient map. Then $q$ is
continuous and, by $(c)$, for each $w\in L$ we have $f_w:|w|\to n$
with $w\in V_{f_w(0)}\dots V_{f_w(|w| - 1)}\subseteq L$. Since the
$W_i$ are atoms, we have $w_j\in
q^{-1}(q(w_{j}))=W_{q(w_{j})}\subseteq V_{f_w(j)}$, and it follows
that
\[
  L\subseteq (q^*)^{-1}(q^*[L])=\bigcup_{w\in
    L}(q^*)^{-1}(q^*(w))=\bigcup_{w\in L}W_{q(w_{0})}\dots W_{q(w_{|w|
      - 1})}\subseteq\bigcup_{w\in L} V_{f_w(0)}\dots V_{f_w(|w| -
    1)}\subseteq L.\popQED
\]
\end{proof}
Notice that each of the spaces $Y^n$ is a Boolean space, but~$Y^*$ is
not. Nevertheless, since it is $0$-dimensional, it may be naturally
embedded in a Boolean space, namely in the Banaschewski
compactification of~$Y^*$, cf. Section~\ref{ss:comp}. As a
consequence, we have the following:

\begin{lemma}\label{l:18}
  Let $\cB \subseteq \cl(Y^*)$ be a Boolean algebra of languages over
  the profinite alphabet $Y$. Then, the dual of the embedding $\cB
  \rightarrowtail \cl(Y^*)$ restricts to $Y^*$, yielding a continuous
  function $Y^* \to X_\cB$ with dense image, mapping a word $w \in
  Y^*$ to the ultrafilter $\{K \in \cB \mid w \in K\}$.
\end{lemma}
\begin{proof}
  The dual of the embedding $\cB \rightarrowtail \cl(Y^*)$ is a
  quotient $\beta_0(Y^*) \twoheadrightarrow X_\cB$. Since $Y^*$
  densely embeds in $\beta_0(Y^*)$ via the assignment $w \mapsto \{K
  \in \cl(Y^*) \mid w \in K\}$, the claim follows.
\end{proof}

Now, we show that the Banaschewski and \v Cech-Stone compactifications
of~$Y^*$ coincide.
\begin{lemma}\label{l:10}
  Every disjoint union of topological spaces of \v Cech-dimension zero
  has \v Cech-dimension zero.
\end{lemma}

\begin{proof}
  Let $\{X_i\}_{i \in I}$ be a family of topological spaces of \v
  Cech-dimension zero, and $X = \bigcup_{i \in I}X_i$ be its disjoint
  union. Let also $X = U_1 \cup \dots \cup U_n$ be a finite open cover
  of $X$. Then, for every $i \in I$, $X_i = \bigcup_{k = 1}^n(U_k \cap
  X_i)$ is a finite open cover of~$X_i$. Since $X_i$ has \v
  Cech-dimension zero, this open cover admits a finite clopen
  refinement, say $X_i = \bigcup_{j \in F_i} V_{j, i}$, where $F_i$ is
  a finite set and each $V_{j,i}$ is a clopen subset of~$X_i$
  contained in some~$U_k \cap X_i$. For each $k = 1, \dots, n$, we set

  \[V_k = \bigcup \{V_{j, i} \mid i \in I, \ j \in F_i, \ V_{j,
      i}\subseteq U_k \cap X_i\}.\]
  Then, $V_i$ is a clopen subset of~$X$ contained in~$U_k$ and $X
  = \bigcup_{k = 1}^n V_k$. Thus, $\{V_k\}_{k = 1}^n$ is the desired
  finite clopen refinement of the given finite open cover of $X$.
\end{proof}

By a straightforward application of Theorem~\ref{t:1}, we may
conclude that the Banaschewski and \v Cech-Stone compactifications of
$Y^*$ indeed coincide.
\begin{theorem}\label{t:4}
  Let $Y$ be a Boolean space and $Y^*$ be the topological space which
  is the union over $n \ge 0$ of the product spaces~$Y^n$. Then, the
  dual space of $\cl(Y^*)$ is $\beta(Y^*)$, the \v Cech-Stone
  compactification of $Y^*$. In particular, there is a continuous
  embedding $Y^* \rightarrowtail \beta(Y^*)$ with dense image.
\end{theorem}
The following immediate consequence is stated for later reference.
\begin{corollary}\label{c:12}
  Let $X, Y$ be Boolean spaces. Then, $\beta(Y^*) \times X$ is the
  Banaschewski compactification of $Y^* \times X$.
\end{corollary}

A consequence of Lemma~\ref{l:4} is that the extension to profinite
alphabets of an lp-strain of languages is closed under preimages of
continuous lp-morphisms between profinite alphabets:

\begin{lemma}\label{l:17}
  Let $\alpha: Z \to Y$ be a continuous function and $\cV$ an
  lp-strain of languages. Then, the Boolean algebra homomorphism
  $(\alpha^*)^{-1}: \cl(Y^*) \to\cl(Z^*)$ restricts and co-restricts
  to a homomorphism $\cV(Y) \to \cV(Z)$, so that we have the following
  commutative diagram:
  \begin{center}
    \begin{tikzpicture}[node distance = 20mm]
      \node (B) at (0,0) {$\cV(Y)$}; \node[right of = B, xshift = 7mm]
      (A) {${\cl(Y^*)}$}; \node[below of = B, yshift = 5mm] (VB)
      {$\cV(Z)$}; \node[below of = A, yshift = 5mm] (VA) {${\cl(Z^*)}$};
      \draw[>->] (B) to (A); \draw[->, dashed] (B) to (VB); \draw[->]
      (A) to node[right, ArrowNode] {$(\alpha^*)^{-1}$} (VA);
      \draw[>->] (VB) to (VA);
    \end{tikzpicture}
  \end{center}
\end{lemma}
\begin{proof}
  Let $L \in \cV(Y)$. By Lemma~\ref{l:4}, there is a finite continuous
  quotient $h: Y \twoheadrightarrow A$ and a language $K \in \cV(A)$
  such that $L = (h^*)^{-1}(K)$. Consider the factorization of $h
  \circ \alpha$ through its image: $Z
  \overset{\;q\;}{\twoheadrightarrow} B
  \overset{\;e\;}{\rightarrowtail} A$. Then, $B =\im(h \circ \alpha)$
  is a finite continuous quotient of~$Z$ which embeds in $A$. Since
  $\cV$ is an lp-strain of languages, and thus closed under preimages
  of lp-morphisms, the language $(e^*)^{-1}(K)$ belongs to
  $\cV(B)$. Finally, using again Lemma~\ref{l:4}, we may conclude that
  $(q^*)^{-1}((e^*)^{-1}(K))$ belongs to $\cV(Z)$. But, since $e \circ
  q = h \circ \alpha$ and $L = (h^*)^{-1}(K)$, it follows that
  $(q^*)^{-1}((e^*)^{-1}(K)) = (\alpha^*)^{-1}((h^*)^{-1}(K)) =
  (\alpha^*)^{-1}(L)$. Thus, $(\alpha^*)^{-1}(L)$ belongs to $\cV(Z)$,
  as intended.
\end{proof}

The dual statement of Lemma~\ref{l:17}, yields the following:

\begin{lemma}\label{l:5}
  Let $\alpha: Z \to Y$ be a continuous function and $\cV$ an
  lp-strain of languages. Then, there exists a continuous map
  $\widehat \alpha: X_{\cV(Z)} \to X_{\cV(Y)}$ making the following
  diagram commute:
  \begin{center}
    \begin{tikzpicture}[node distance = 20mm]
      \node (B) at (0,0) {$Z^*$}; \node[right of = B] (A) {$X_{\cV(Z)}$};
      \node[below of = B, yshift = 5mm] (VB) {$Y^*$};
      \node[below of = A, yshift = 5mm] (VA) {$X_{\cV(Y)}$};
      \draw[->] (B) to node[ArrowNode] {$f$} (A); \draw[->] (B) to
      node[left, ArrowNode] {$\alpha^*$} (VB); \draw[->, dashed] (A)
      to node[right, ArrowNode] {$\widehat \alpha$} (VA); \draw[->]
      (VB) to node[below,ArrowNode] {$g$} (VA);
    \end{tikzpicture}
  \end{center}
  where $f:Y^* \to X_{\cV(Y)}$ and $g:Z^* \to X_{\cV(Z)}$ are the
  continuous functions with dense image obtained by dualizing and
  restricting the embeddings $\cV(Y) \rightarrowtail \cl(Y^*)$ and
  $\cV(Z) \rightarrowtail \cl(Z^*)$, respectively
  (cf. Lemma~\ref{l:18}).
\end{lemma}
\begin{proof}
  By Lemma~\ref{l:17}, the diagram below restricts correctly:
  \begin{center}
    \begin{tikzpicture}[node distance = 20mm]
      \node (B) at (0,0) {$\cV(Y)$}; \node[right of = B, xshift = 7mm]
      (A) {${\cl(Y^*)}$}; \node[below of = B, yshift = 5mm] (VB)
      {$\cV(Z)$}; \node[below of = A, yshift = 5mm] (VA)
      {${\cl(Z^*)}$};
      \draw[>->] (B) to (A); \draw[->, dashed] (B) to (VB); \draw[->]
      (A) to node[right, ArrowNode] {$(\alpha^*)^{-1}$} (VA);
      \draw[>->] (VB) to (VA);
    \end{tikzpicture}
  \end{center}
  Let $\widehat \alpha$ be the dual of the homomorphism
  $(\alpha^*)^{-1}: \cV(Y) \to \cV(Z)$. We claim that $\widehat
  \alpha$ makes the desired diagram commute. Indeed, given $K \in
  \cV(Y)$ and $w \in Z^*$, by definition of dual map and by
  Lemma~\ref{l:18}, we have the following:
  \begin{align*}
    K \in \widehat \alpha \circ f(w)
    & \iff (\alpha^*)^{-1}(K) \in f(w) \iff w \in (\alpha^*)^{-1}(K)
    \\ & \iff \alpha^* (w) \in K \iff K \in g \circ \alpha^*(w).
  \end{align*}
  Thus, $ \widehat \alpha \circ f = g \circ \alpha^*$ as required.
\end{proof}

\section{Recognition of Boolean algebras closed under
  quotients}\label{sec:bims}
Quotients by words are natural operations on languages. In the regular
setting they dually correspond to the multiplication in profinite
monoids~\cite{GehrkeGrigorieffPin08}, and one is often interested in
studying Boolean algebras that are closed under quotients.
Beyond regularity, the closure under quotients will bring additional
algebraic structure, and thus computational power, to the dual space
of the Boolean algebra considered. We may thus think of the dual space
of a Boolean algebra closed under quotients as having ``almost'' a
monoid structure. This intuitive idea is captured by the \emph{Boolean
  spaces with internal monoids}, which were introduced
in~\cite{GehrkePetrisanReggio16} as an alternative to the semiuniform
monoids of~\cite{GehrkeGrigorieffPin10}. In Section~\ref{sec:1} we
will introduce \emph{Boolean spaces with an internal monoid}
(\bim{}'s) and define its variant of a \bim{}-stamp.

Then, in Section~\ref{sec:3}, we will discuss recognition of languages
that are defined by some formula with a free variable, that is,
languages of marked words. We already saw in
Section~\ref{ss:markedwords} that the set~$A^* \otimes \N$ of all
marked words is not equipped with a monoid structure. For this reason,
\bim{}'s will not appear as a natural notion of
recognizer. Nevertheless, the fact that $A^* \otimes \N$ embeds in the
free monoid $(A \times 2)^*$ allows us to identify some monoid actions
that will be crucial when defining semidirect products in
Section~\ref{sec:closing-quot}.
\subsection{Languages over profinite alphabets}\label{sec:1}
Let $Y$ be a profinite alphabet, and $\cB \subseteq \cl(Y^*)$ be a
Boolean algebra of languages. If $\cB$ is closed under quotients, then
for every word $w$ over~$Y$, the homomorphisms
\[\ell_{w}:\cl(Y^*) \to \cl(Y^*), \qquad L \mapsto w^{-1}L\]
and
\[r_w:\cl(Y^*) \to \cl(Y^*), \qquad L \mapsto Lw^{-1}\]
restrict and co-restrict to endomorphisms of~$\cB$.
Thus, we have the two following commutative diagrams:
\begin{center}
  \begin{tikzpicture}[node distance = 15mm, ->]
    \node (A) {$\cl(Y^*)$};
    \node[right of = A, xshift = 20mm] (B) {$\cl(Y^*)$};
    \node[below of = A] (C) {$\cB$};
    \node[below of = B] (D) {$\cB$};
    \draw (A) to node[above, ArrowNode] {$\ell_w$} (B); \draw (C) to
    node[below, ArrowNode] {$\ell_w$} (D); \draw[>->] (C) to (A);
    \draw[>->] (D) to (B);
    \node[right of = B, xshift = 25mm] (A') {$\cl(Y^*)$}; \node[right
    of = A', xshift = 20mm] (B') {$\cl(Y^*)$}; \node[below of = A']
    (C') {$\cB$}; \node[below of = B'] (D') {$\cB$};
    \draw (A') to node[above, ArrowNode] {$r_w$} (B'); \draw (C') to
    node[below, ArrowNode] {$r_w$} (D'); \draw[>->] (C') to (A');
    \draw[>->] (D') to (B');
  \end{tikzpicture}
\end{center}
Dually, we have the following commutative diagrams of continuous functions
\begin{equation}
  \begin{aligned}
    \begin{tikzpicture}[node distance = 15mm, ->]
      \node (A) {$Y^*$};
      \node[right of = A, xshift = 20mm] (B) {$Y^*$};
      \node[below of = A] (C) {$X_\cB$};
      \node[below of = B] (D) {$X_\cB$};
      \draw (A) to node[above, ArrowNode] {$\widetilde\ell_w$} (B);
      \draw (C) to node[below, ArrowNode] {$\widetilde\ell_w$} (D);
      \draw (A) to node[left, ArrowNode] {$\pi$} (C); \draw (B) to
      node[right, ArrowNode] {$\pi$} (D);
      \node[right of = B, xshift = 25mm] (A') {$Y^*$}; \node[right of =
      A', xshift = 20mm] (B') {$Y^*$}; \node[below of = A'] (C')
      {$X_\cB$}; \node[below of = B'] (D') {$X_\cB$};
      \draw (A') to node[above, ArrowNode] {$\widetilde r_w$} (B');
      \draw (C') to node[below, ArrowNode] {$\widetilde r_w$} (D');
      \draw (A') to node[left, ArrowNode] {$\pi$}(C'); \draw (B') to
      node[right, ArrowNode] {$\pi$} (D');
    \end{tikzpicture}
  \end{aligned}\label{eq:12}
\end{equation}
We make a few remarks about these diagrams. First, notice that
concatenation of words over~$Y$ turns $Y^*$ into a topological
monoid. Indeed, the multiplication $m: Y^* \times Y^* \to Y^*$ is
continuous because, if $C_1, \dots, C_n \subseteq Y$ are (cl)open
subsets of~$Y$, then
\[m^{-1}(C_1 \times \dots \times C_n) = \bigcup_{i = 0}^n (C_1 \times
  \dots C_i) \times (C_{i+1} \times \dots \times C_n) \]
is a (cl)open subset of $Y^* \times Y^*$.  Moreover, since for every
$u, v \in Y^*$ we have $\ell_u \circ r_v = r_v \circ \ell_u$, we also
have $\widetilde\ell_u \circ \widetilde r_v = \widetilde r_v \circ
\widetilde\ell_u$. Therefore, the monoid structure on~$Y^*$ induces a
biaction with continuous components of~$Y^*$ on~$X_\cB$ which is given
by
\begin{equation}
  Y^* \times X_\cB \to X_\cB, \quad (u, x) \mapsto
  \widetilde\ell_u(x)\qquad\text{and}\qquad X_\cB \times Y^* \to X_\cB, \quad
  (x, v) \mapsto \widetilde r_v(x).\label{eq:10}
\end{equation}
However, $X_\cB$ itself does not necessarily inherit a monoid
structure. Indeed, as it was shown in~\cite{GehrkeGrigorieffPin10,
  Gehrke16}, when $A$ is a finite alphabet, this is case if and only
if $\cB$ consists of regular languages. 
Nevertheless, \eqref{eq:10} induces a monoid structure on the dense
subspace $M= \pi[Y^*]$ of $X_\cB$, which is given by
\begin{equation}
  \pi(u) \cdot \pi(v) = \pi(uv),\label{eq:11}
\end{equation}
for every $u,v \in Y^*$.  Indeed, with a routine computation we may
derive the following equalities:
\begin{equation}
  \widetilde\ell_u(\pi(v)) = \pi(uv) = \widetilde r_v(\pi(u)).\label{eq:17}
\end{equation}
These not only show that~\eqref{eq:11} is well-defined in the sense
that $\pi(uv) = \pi(u'v')$ whenever $\pi(u) = \pi(u')$ and $\pi(v) =
\pi(v')$, but also that the monoid structure on~$M$ is indeed
inherited from~\eqref{eq:10}.
In fact, it is not hard to see that~$\pi: Y^* \twoheadrightarrow M$ is
precisely the syntactic homomorphism of~$\cB$ as defined in
Section~\ref{ss:languages}.
Moreover, since $M$ is dense in $X_\cB$, \eqref{eq:17} also shows that
whenever $\pi(u) = \pi(u')$ (respectively, $\pi(v) = \pi(v')$), the
continuous functions $\widetilde\ell_u$ and $\widetilde\ell_{u'}$
(respectively, $\widetilde r_v$ and $\widetilde r_{v'}$) coincide on
all of~$X_\cB$. Therefore, the natural biaction of $M$ on
itself extends to a biaction of~$M$ on~$X_\cB$ given by
\[M \times X_\cB \to X_\cB, \quad (m, x) \mapsto
  \lambda_m(x)\qquad\text{and}\qquad X_\cB \times M \to X_\cB, \quad
  (x, m) \mapsto \rho_m(x),\]
with continuous components at each $m \in M$. Here, for each $m \in
M$, $\lambda_m$ (respectively, $\rho_m$) denotes the continuous
function $\widetilde \ell_u$ (respectively, $\widetilde r_u$) where $u
\in Y^*$ is any word satisfying $\pi(u)= m$.
\begin{example}
  Let $A$ be a finite alphabet and ${\rm Reg}(A)$ denote the set of
  regular languages over $A$ (which, as already observed, is a Boolean
  algebra closed under quotients). Let also $\widehat {A^*}$ denote
  the free profinite monoid over~$A$ (see e.g.~\cite{Almeida05} for a
  description of~$\widehat{A^*}$ both as a projective limit and as a
  completion of a suitable metric space $(A^*, d)$). Then, ${\rm
    Reg}(A)$ is a Boolean algebra isomorphic to the Boolean algebra of
  clopen subsets of $\widehat{A^*}$, via the assignment that sends a
  regular language to its topological
  closure~\cite[Section~3.6]{Almeida94}.  In particular, the Stone
  dual of ${\rm Reg} (A)$ is the underlying topological space
  of~$\widehat{A^*}$ and the map $\pi: A^* \to \widehat{A^*}$ from
  diagrams~\eqref{eq:12} is an embedding (with dense image). The
  biaction of $A^*$ on $\widehat{A^*}$ is obtained by restriction of
  the multiplication on
  $\widehat{A^*}$~\cite[Section~6]{GehrkeGrigorieffPin08} (see also
  \cite[Section~4.2]{Gehrke16}).
\end{example}

The following definition, which captures duality for Boolean
subalgebras of languages closed under the quotient operations,
originates in~\cite{GehrkePetrisanReggio16}, where it was used for
recognition over finite alphabets. The above considerations verify
that it remains the appropriate notion for recognition over profinite
alphabets.
\begin{definition}
  A \emph{Boolean space with an internal monoid} (\bim{}) is a triple
  $(M, p, X)$ where $X$ is a Boolean space equipped with a biaction of
  a monoid $M$ whose right and left components at each $m\in M$ are
  continuous and an injective function $p\colon M\rightarrowtail X$
  which has dense image and is a morphism of sets with
  $M$-biactions. That is, for each $m\in M$, there are continuous
  functions $\lambda_m, \rho_m: X \to X$ that make the following
  diagrams commute:
  \begin{center}
    \begin{tikzpicture}[->,node distance=3.4cm, scale=1]
      \node(M1r) at (1.5,1.2) {$M$}; \node(Xr) at (3.5,0) {$X$};
      \node(M2r) at (1.5,0) {$M$}; \node(Yr) at (3.5,1.2) {$X$};

      \draw[->] (M1r) to node[left, ArrowNode] {$(\_) \cdot m$} (M2r);
      \draw[->] (Yr) to node[right, ArrowNode] {$\rho_m$} (Xr);
      \draw[>->] (M1r) to node[above, ArrowNode] {$p$} (Yr);
      \draw[>->] (M2r) to node[below, ArrowNode] {$p$} (Xr);

      \node(M1) at (-3.5,1.2) {$M$}; \node(X) at (-1.5,0) {$X$};
      \node(M2) at (-3.5,0) {$M$}; \node(Y) at (-1.5,1.2) {$X$};

      \draw[->] (M1) to node[left, ArrowNode] {$m\cdot (\_)$}
      (M2); \draw[->] (Y) to node[right, ArrowNode]
      {$\lambda_m$} (X); \draw[>->] (M1) to
      node[above, ArrowNode] {$p$} (Y); \draw[>->] (M2) to
      node[below, ArrowNode] {$p$} (X);
    \end{tikzpicture}
  \end{center}
  A morphism $(M, p, X) \to (N, q, W)$ of $\bim$s is a pair $(g,
  \varphi)$, where $g: M \to N$ is a monoid homomorphism, $\varphi: X
  \to W$ is a continuous function, and the following diagram commutes:
  \begin{center}
    \begin{tikzpicture}[->,node distance=3.4cm, scale=1]

      \node(M1) at (-3.5,1.2) {$M$}; \node(X) at (-1.5,0) {$X$};
      \node(M2) at (-3.5,0) {$N$}; \node(Y) at (-1.5,1.2) {$W$};

      \draw[->] (M1) to node[left, ArrowNode] {$g$} (M2); \draw[->]
      (Y) to node[right, ArrowNode] {$\varphi$} (X); \draw[>->] (M1)
      to node[above, ArrowNode] {$p$} (Y); \draw[>->] (M2) to
      node[below, ArrowNode] {$q$} (X);
    \end{tikzpicture}
  \end{center}
  We say that $(g, \varphi)$ is a quotient provided $g$ is surjective
  (and consequently, so is~$\varphi$).
\end{definition}

\begin{example}
  Let $Y$ be a profinite alphabet. Since $\cl(Y^*)$ is a Boolean
  algebra closed under quotients, the \v Cech-Stone compactification
  of~$Y^*$, defines a \bim{} $(Y^*, \, e,\, \beta(Y^*))$
  (cf. Theorem~\ref{t:4}).
\end{example}

Given a profinite alphabet $Y$, we say that a language $L \subseteq
\cl(Y^*)$ is \emph{recognized} by the \bim{} $(M, p, X)$ if there
exists a homomorphism $\mu: Y^* \to M$ such that $p \circ \mu$ is
continuous, and a clopen subset $C \subseteq X$ satisfying $L = (p
\circ \mu)^{-1}(C)$. Notice that, since $X$ is a Boolean space, $p
\circ \mu$ may be uniquely extended to a continuous function $\varphi:
\beta(Y^*) \to X$, and so, if $e: Y^* \rightarrowtail \beta(Y^*)$
denotes the \v Cech-Stone compactification of~$Y^*$, then the
homomorphisms $\mu: Y^* \to M$ for which $p \circ \mu$ is continuous
are in a bijective correspondence with the \bim{} morphisms $(Y^*,\,
e,\, \beta(Y^*)) \to (M, p, X)$. Moreover, the dual of the
function~$\varphi$ is the homomorphism of Boolean algebras $(p \circ
\mu)^{-1}: \cl(X) \to \cl(\beta(Y^*)) = \cl(Y^*)$ whose image consists
of the languages recognized by~$(M, p, X)$ via~$\mu$. Finally,
since~$p$ is a morphism of sets with $M$-biactions and the biaction
of~$M$ on~$X$ has continuous components, the Boolean algebra of
languages recognized by~$(M, p, X)$ via~$\mu$ is closed under
quotients. Indeed, if $L = (p \circ \mu)^{-1}(C)$ for some clopen
subset $C \subseteq X$, and $u, v, w \in Y^*$, then we have
\begin{align*}
  w \in u^{-1}Lv^{-1}
  & \iff p\circ \mu(uwv) \in C \iff p (\mu(u)\, \mu(w)\,\mu(v)) \in C
  \\ & \iff \lambda_{\mu(u)}\circ \rho_{\mu(v)} \circ p \circ \mu(w)
       \in C
  \\ & \iff w \in (p \circ \mu)^{-1}((\lambda_{\mu(u)}\circ
       \rho_{\mu(v)})^{-1}(C)).
\end{align*}
Thus, $u^{-1}Lv^{-1} = (p \circ \mu)^{-1}((\lambda_{\mu(u)}\circ
\rho_{\mu(v)})^{-1}(C))$ is recognized by the clopen
$(\lambda_{\mu(u)}\circ \rho_{\mu(v)})^{-1}(C) \subseteq X$.

In order to handle lp-varieties (that is, lp-strains closed under
quotients), it is useful to refine the notion of \bim{}. In the
regular setting, this corresponds to the notion of
\emph{stamp}~\cite{PinStraubing05}. The idea is to consider \bim{}s
with a chosen profinite set of generators for the monoid component,
and constrain the set of languages we recognize accordingly.
Formally, a \emph{\bim{}-stamp} (called \emph{\bim{} presentation}
in~\cite{BorlidoCzarnetzkiGehrkeKrebs17}) is a tuple $\tR = (Y,\mu,M,
p, X)$, where $\mu: Y^* \twoheadrightarrow M$ is a monoid quotient,
$(M, p, X)$ is a \bim{}, and $p \circ \mu$ is a continuous
function. By the observations made in the preceding paragraph,
\bim{}-stamps $(Y,\mu,M, p, X)$ may also be understood as \bim{}
quotients $(Y^*, e, \beta (Y^*)) \twoheadrightarrow (M, p, X)$.
We remark that each Boolean algebra closed under quotients $\cB
\subseteq \cl(Y^*)$ naturally defines a \bim-stamp
\[Y^* \twoheadrightarrow M = \pi[Y^*] \rightarrowtail X_\cB,\]
as described in the beginning of this section. This is called the
\emph{syntactic \bim-stamp of~$\cB$} and it is denoted by $\synt\cB =
(Y, \mu_\cB, M_\cB, p_\cB, X_\cB)$.

\bim{}-stamps encode two types of behavior: algebraic behavior given
by the  triple $(Y, \mu, M)$ and topological behavior given by the
continuous function $p \circ \mu: Y^* \to X$.
The interplay between these two plays a central role in this paper.
We say that a language $L$ is \emph{recognized} by the \bim{}-stamp
$\tR = (Y, \mu, M, p, X)$ provided $L$ is recognized by the \bim{}
$(M, p, X)$ via $\mu$.  As observed above, the set of all languages
recognized by a \bim{}-stamp forms a Boolean algebra closed under
quotients.
Similarly to what happens in the regular setting, the syntactic
\bim-stamp of a given Boolean algebra closed under quotients is the
``optimal'' \bim{}-stamp recognizing that Boolean algebra, in the
following sense:
\begin{proposition}\label{p:5}
  A \bim{}-stamp $\tR = (Y, \mu, M, p, X)$ recognizes a Boolean
  algebra closed under quotients~$\cB$ if and only if the syntactic
  \bim{}-stamp of~$\cB$ factors through~$\tR$, that is, there is a
  commutative diagram:
  \begin{center}
    \begin{tikzpicture}[->,scale = 1.2]
      \node(A) at (0,0) {$Y^*$}; \node(M) at (1.7,0) {$M$}; \node(N)
      at (1.7,-1.2) {$M_\cB$}; \node(X) at (3.4,0) {$X$}; \node(Y) at
      (3.4,-1.2) {$X_\cB$};
      \draw[->>] (A) to node[left, yshift = -2mm, ArrowNode]
      {$\mu_\cB$} (N); \draw[dashed](M) to node[right, ArrowNode]
      {$g$} (N); \draw[dashed] (X) to node[right, ArrowNode]
      {$\varphi$} (Y); \draw[->>] (A) to node[above, ArrowNode]
      {$\mu$} (M); \draw[>->] (M) to node[above, ArrowNode] {$p$} (X);
      \draw[>->] (N) to node[below, ArrowNode] {$p_\cB$} (Y);
    \end{tikzpicture}
  \end{center}
  where $g$ is a homomorphism and~$\varphi$ a continuous function.
\end{proposition}
\begin{proof}
  It is clear that if $\synt\cB$ factors through $\tR$, then every
  language recognized by $\synt\cB$ is also recognized by~$\tR$ and,
  in particular, $\cB$ is recognized by $\tR$.  Conversely, suppose
  that~$\cB$ is recognized by~$\tR$.  Then, the kernel of $\mu$ is
  contained in the syntactic congruence ${\sim_\cB}$, and therefore,
  the syntactic morphism $\mu_\cB$ factors through $\mu$, say $g \circ
  \mu = \mu_\cB$. On the other hand, since $p \circ \mu$ has dense
  image, there are embeddings $\cB \rightarrowtail \cl(X)$ and $\cl(X)
  \rightarrowtail \cl(Y^*)$, or dually, continuous quotients $\pi:
  \beta(Y^*) \twoheadrightarrow X$ and $\varphi: X \twoheadrightarrow
  X_\cB$, where the restriction of $\pi$ to $Y^*$ is $p \circ \mu$,
  and $p_\cB = \varphi \circ p$. Finally, since $\mu$ is surjective
  and $\mu_\cB = g \circ \mu$, it follows that $\varphi \circ p =
  p_\cB \circ g$ as intended.
\end{proof}

A morphism between \bim{}-stamps $\tR = (Y, \mu, M, p, X)$ and $\tS =
(Z, \nu, N, q, W)$ is a triple $\Phi = (h, g, \varphi)$, where $h: Y^*
\to Z^*$ is a continuous homomorphism and $(g, \varphi)$ is a morphism
between the corresponding \bim{} components of $\tR$ and $\tS$, so
that the following diagram commutes:
\begin{center}
  \begin{tikzpicture}[->, scale=1.2]
    \node(A) at (0,0) {$Y^*$}; \node(M) at (1.7,0) {$M$}; \node(B) at
    (0,-1.2) {$Z^*$}; \node(N) at (1.7,-1.2) {$N$}; \node(X) at
    (3.4,0) {$X$}; \node(Y) at (3.4,-1.2) {$W$};

    \draw[->] (A) to node[left, ArrowNode] {$h$} (B);
    \draw[->] (M) to node[right, ArrowNode] {$g$} (N);
    \draw[->] (X) to node[right, ArrowNode] {$\varphi$} (Y);
    \draw[->>] (A) to node[above, ArrowNode] {$\mu$} (M);
    \draw[->>] (B) to node[below, ArrowNode] {$\nu$} (N);
    \draw[>->] (M) to node[above, ArrowNode] {$p$} (X);
    \draw[>->] (N) to node[below, ArrowNode] {$q$} (Y);
  \end{tikzpicture}
\end{center}
When $h$ is an lp-morphism, we say that $\Phi$ is an \emph{lp-morphism
  of \bim{}-stamps}.  Note that when $h$ is onto, then so are $g$ and
$\varphi$. If moreover~$h$ is length-preserving, then we will
call~$\Phi$ an \emph{lp-quotient}.
It is worth mentioning that every lp-quotient $\Phi = (h, g, \varphi):
\tR \twoheadrightarrow \tS$ induces the following commutative diagram
of continuous functions:
\begin{center}
  \begin{tikzpicture}[node distance = 15mm, ->>]
    \node (A) {$Y^*$}; \node[right of = A, xshift = 25mm] (B) {$X$};
    \node[below of = A] (C) {$Z^*$}; \node[below of = B] (D) {$W$};
    \draw (A) to node[left, ArrowNode] {$h$} (C); \draw (B) to
    node[right, ArrowNode] {$\varphi$} (D); \draw[->] (A) to
    node[above, ArrowNode] {$p \circ \mu$} (B); \draw[->] (C) to
    node[below, ArrowNode] {$q \circ \nu$} (D);
  \end{tikzpicture}
\end{center}
Thus, taking preimages yields the following commutative diagram of
homomorphisms of Boolean algebras:
\begin{center}
  \begin{tikzpicture}[node distance = 15mm, >->]
    \node (A) {$\cl(X)$}; \node[right of = A, xshift = 25mm] (B)
    {$\cl(Y^*)$}; \node[below of = A] (C) {$\cl(W)$}; \node[below of =
    B] (D) {$\cl(Z^*)$};
    \draw (C) to node[left, ArrowNode] {$\varphi^{-1}$} (A); \draw (D) to
    node[right, ArrowNode] {$h^{-1}$} (B); \draw[->] (A) to
    node[above, ArrowNode] {$(p \circ \mu)^{-1}$} (B); \draw[->] (C)
    to node[below, ArrowNode] {$(q \circ \nu)^{-1}$} (D);
  \end{tikzpicture}
\end{center}
Therefore, $\tS$ being an lp-quotient of $\tR$ means that every
language recognized by~$\tS$ is also recognized by~$\tR$, under the
identification $\cl(Z^*) \subseteq \cl(Y^*)$.

Finally, we show that \bim{}-stamps over profinite alphabets can be
described in terms of certain projective limits of \bim{}-stamps over
finite alphabets. Although a bit technical, the proof mostly uses
standard arguments used in the computation of projective limits in
set-based categories.

\begin{proposition}\label{p:1}
  Projective limits exist in the category of \bim{}-stamps with
  lp-morphisms. Moreover, each \bim{}-stamp is the projective limit of
  its \bim{}-stamp lp-quotients over finite alphabets.
\end{proposition}
\begin{proof}
  We give the general idea of the proof. All the missing details are
  routine computations. Let $\cF = \{\tR_i = (Y_i, \mu_i, M_i, p_i,
  X_i)\}_{i \in I}$ be a projective system in the category of
  \bim{}-stamps with lp-morphisms. For $i \ge j$, we denote by
  $\Phi_{i,j} = (h_{i,j}^*, g_{i,j}, \varphi_{i,j}): \tR_i \to \tR_j$
  the corresponding connecting morphism, with $h_{i,j}: Y_i \to Y_j$ a
  continuous function.  Then, each of the families $\{Y_i\}_{i \in
    I}$, $\{M_i\}_{i \in I}$, and $\{X_i\}_{i \in I}$ forms itself a
  projective system, with the maps $h_{i,j}$, $g_{i,j}$ and $\varphi_{i,j}$ 
  as connecting morphisms, respectively. We set
  \begin{equation}
    Y = \lim\limits_{\longleftarrow}\ \{Y_i \mid i \in I\}, \qquad M_0
    = \lim\limits_{\longleftarrow}\ \{M_i \mid i \in I\}, \qquad
    \text{and}\qquad X_0 = \lim\limits_{\longleftarrow}\ \{X_i \mid i
    \in I\},\label{eq:14}
  \end{equation}  
  and for each $i \in I$, we denote by $h_i: Y \twoheadrightarrow
  Y_i$, $\zeta_i: M_0 \twoheadrightarrow M_i$, and $\pi_i: X_0
  \twoheadrightarrow X_i$ the corresponding projections.
  
  Using the explicit description of projective limits displayed
  in~\eqref{eq:4}, we may check that there are well-defined maps
  \begin{equation}
    \mu_0: Y^* \to M_0, \ z \mapsto (\mu_i \circ h_i^*(z))_{i \in
      I}\qquad\text{and} \qquad p_0: M_0 \rightarrowtail X_0, \ m \mapsto
    (p_i\circ \zeta_i(m))_{i \in I}.\label{eq:8}
  \end{equation}
  Graphically, we have the following commutative diagram:
  \begin{center}
    \begin{tikzpicture}[node distance = 20mm]
      \node (Y) at (0,0) {$Y^*$}; \node[right of = Y](M) {$M_0$};
      \node[right of = M](X) {$X_0$}; \node[below of = Y, yshift = 5mm](Yi){$Y^*_i$};
      \node[right of = Yi](Mi){$M_i$}; \node[right of = Mi](Xi)
      {$X_i$};
      \draw[->] (Y) to node[above, ArrowNode] {$\mu_0$} (M);
      \draw[>->] (M) to node[above, ArrowNode] {$p_0$} (X); \draw[->>]
      (Y) to node[left, ArrowNode] {$h^*_i$} (Yi); \draw[->>] (M) to
      node[right, ArrowNode] {$\zeta_i$} (Mi); \draw[->>] (X) to
      node[right, ArrowNode] {$\pi_i$} (Xi); \draw[->>] (Yi) to
      node[below, ArrowNode] {$\mu_i$} (Mi); \draw[>->] (Mi) to
      node[below, ArrowNode] {$p_i$} (Xi);
    \end{tikzpicture}
  \end{center}
  Moreover, by continuity of $p_i\circ \mu_i \circ h_i^*$, for every
  $i \in I$, the map $p_0 \circ \mu_0$ is also continuous.
  We set
  \[M = \mu_0 [Y^*] \qquad \text{and} \qquad X = \overline{p_0[M]}.\]
  Co-restricting $\mu_0$ to $M$, we have an onto homomorphism $\mu:
  Y^* \twoheadrightarrow M$, and the restriction and co-restriction of
  $p_0$ to $M$ and to $X$, respectively, yields a map $p:M
  \rightarrowtail X$ with dense image. We claim that $\tR = (Y, \mu,
  M, p, X)$ is the projective limit of $\cF$. The fact that $\tR$
  satisfies the universal property of projective limits is inherited
  from the universal properties satisfied by $Y$, $M_0$ and $X_0$.

  Thus, it remains to check that $M$ continuously bi-acts on~$X$.

  We denote by $\lambda_{i,m_i}: X_i \to X_i$ the (continuous)
  component at $m_i \in M_i$ of the left action of~$M_i$
  on~$X_i$. Then, for $m = (m_i)_{i \in I} \in M$, setting
  \[\lambda_m: X \to X, \qquad x = (x_i)_{i \in I} \mapsto
    (\lambda_{i,m_i}(x_i))_{i \in I}\]
  defines a continuous map which induces a left action of $M$ on
  $X$. Indeed, for every $m, m' \in M$, we may compute $\lambda_m
  \circ p(m') = p (mm')$, using the analogous property for each
  $\lambda_{i,m_i}$. This proves not only that $\lambda_m$ is
  well-defined (because $\lambda_m[X] =
  \lambda_m\left[\overline{p[M]}\right] \subseteq \overline{p[M]} =
  X$), but also that $p$ is a morphism of sets with a left $M$-action.
  Similarly, we can define the right action of~$M$ on~$X$. The
  compatibility between the left and right actions is inherited from
  the compatibility between the left and right actions for each
  $\tR_i$. Thus, $\tR$ is a \bim{}-stamp.

  Finally, that each \bim-stamp is the projective limit of all its
  lp-quotients over finite alphabets is an easy consequence of the
  explicit computation of projective limits just made.
\end{proof}

We remark that, even though each of the maps $\mu_i \circ h_i^*$ is
onto, since $Y^*$ is not compact, the map $\mu_0$ defined
in~\eqref{eq:8} is not onto in general (cf.~\cite[Lemma
1.2]{AlmeidaWeil1995}). To see this, consider for instance, for each
integer $n$, the \bim{}-stamp $(\{0\}^* \twoheadrightarrow \mathbb Z_n
\xrightarrow {id} \mathbb Z_n)$. Then, we have a projective limit
system $\cF = \{(\{0\}^* \twoheadrightarrow \mathbb Z_n \xrightarrow
{id} \mathbb Z_n)\}_{n \in \N}$ and $M_0$ (which equals $X_0$) defined
in~\eqref{eq:14} is the additive profinite group~$\widehat \Z$ (seen
as a monoid). The map $\mu_0$ is then the inclusion $\{0\}^*
\rightarrowtail \widehat \Z$, which is not surjective. Finally, the
projective limit of~$\cF$ is the \bim{}-stamp $(\{0\}^*
\twoheadrightarrow \N \rightarrowtail \widehat \N)$, where $\widehat
\N$ denotes the closure of~$\N$ in~$\widehat \Z$.

\subsection{Languages of marked words and quotienting
  operations}\label{sec:3}

Recall that the set $A^* \otimes\N$ of marked words may be seen as a
regular language in $(A \times 2)^*$. The fact that the syntactic
morphism of this language is $\mu: (A\times 2)^*\to\{e,m,z\}$ given by
$(a,0)\mapsto e$ and $(a,1)\mapsto m$ corresponds to saying that
\[
  (A\times 2)^*\cong A^* \,\uplus\, (A^* \otimes\N)\,\uplus\, A_z,
\]
as defined in Section~\ref{ss:markedwords}, and that concatenation on
$(A \times 2)^*$ decomposes as follows: it yields biactions of $A^*$
on each of these components, thus accounting for all pairs involving
an element of~$A^*$, and any other pair is sent to $A_z$.

Dually, the above disjoint union decomposition of $(A\times 2)^*$ yields 
the following Cartesian product decomposition
\[
\cP((A\times 2)^*)\cong \cP(A^*) \,\times\, \cP(A^* \otimes\N)\,\times\, \cP(A_z),
\]
We are interested in Boolean subalgebras $\cD$ of $\cP((A \times
2)^*)$.  Dual to the inclusions of the disjoint components, we get
projections of $\cD$ onto
\[
  \cD_0=\{L\cap A^*\mid L\in \cD\}, \quad \cD_1=\{L\cap (A^* \otimes\N)\mid
  L\in \cD\},\quad \text{and}\quad\cD_z=\{L\cap A_z\mid L\in \cD\},
\]
respectively. While $\cD$ is always contained in
$\cD_0\times\cD_1\times\cD_z$ it is not difficult to see that we have
equality if and only if $\cD$ contains the language $A^*\otimes \N$
and its orbit under the biaction of $(A\times 2)^*$.

We will be particularly interested in such algebras for which $\cD_z=2$. In this
case, any component of the dual biaction by quotients on $\cP((A\times 2)^*)$ 
with domain $\cP(A_z)$ just sends the bounds to bounds of the appropriate 
component. The remaining components are as follows. First, we have the 
biactions of~$A^*$ on~$\cP(A^*)$ and on~$\cP(A^* \otimes \N)$ given, for 
$u\in A^*$, respectively, by
\begin{gather}
  \begin{aligned}
    \ell^0_u: \cP(A^*) &\to \cP(A^*) \qquad\qquad
    & r^0_u: \cP(A^*)  &\to \cP(A^*), \\
     L \ \ &\mapsto u^{-1}L & L\ \  &\mapsto Lu^{-1}
  \end{aligned}\label{eq:25}
\end{gather}
and by
\begin{gather}
  \begin{aligned}
    \ell^1_u: \cP(A^* \otimes \N) &\to \cP(A^* \otimes \N)
    \qquad\qquad
    & r^1_u: \cP(A^* \otimes \N)  &\to \cP(A^* \otimes \N). \\
    L\ \  &\mapsto u^{-1}L & L\ \  &\mapsto Lu^{-1}
  \end{aligned}\label{eq:13}
\end{gather}
For every marked word $(w,
i)$, the biaction of~$A^*$ on~$A^* \otimes \N$ also induces two
functions $A^* \to A^* \otimes \N$, which are given by left and by
right multiplication by~$(w,i)$. Dually, these define
\begin{gather}
  \label{eq:7}
  \begin{aligned}
    \ell^1_{(w,i)}:\cP(A^* \otimes \N) & \to \cP(A^*)\qquad\qquad&
    r^1_{(w,i)}: \cP(A^* \otimes \N) & \to \cP(A^*).\\
    L\ \ &\mapsto (w,i)^{-1}L & L\ \ &\mapsto L(w,i)^{-1}
  \end{aligned}
\end{gather}
\begin{lemma}\label{l:12}
  Let $\cD \subseteq \cP((A \times 2)^*)$ be a Boolean subalgebra such
  that $\cD_z = 2$. Then, the following are equivalent:
  \begin{enumerate}
  \item\label{item:1} $\cD$ is closed under quotients and $(A^*
    \otimes \N) \in \cD$ or $A_z\in\cD$;
  \item\label{item:4} $\cD\cong\cD_0\times\cD_1\times 2$, $\cD_0$ and
    $\cD_1$ are closed under the respective biactions of $A^*$, and
    for all $L\in \cD_1$ and all $(w,i) \in A^* \otimes\N$ we have
    $(w,i)^{-1}L,\ L(w,i)^{-1}\in\cD_0$.
  \end{enumerate}
\end{lemma}
\begin{proof}
  Suppose we have~\ref{item:1}. If $A_z \in \cD$ then, as $\cD$ is
  closed under quotients,
  \[
  (a,1)^{-1}A_z=(A^* \otimes \N)\cup A_z\in\cD,
  \]
  and thus $A^*\otimes \N\in\cD$. Since~$\cD$ contains the language
  $A^*\otimes \N$ and is closed under quotients, it also contains its
  orbit under the biaction of $(A\times 2)^*$. Thus, we have that
  $\cD$ is isomorphic to~$\cD_0\times\cD_1\times \cD_z$ which, by
  hypothesis, is $\cD_0\times\cD_1\times 2$. The rest of the assertion
  in~\ref{item:4} follows from~$\cD$ being closed under quotients and
  the splitting of the biaction of $(A\times 2)^*$ on~$\cP((A\times
  2)^*)$ as given by~\eqref{eq:25}, \eqref{eq:13}, and~\eqref{eq:7}.

  Conversely, let us assume that~\ref{item:4} holds. Notice that
  having $\cD\cong\cD_0\times\cD_1\times \cD_z$ amounts to having
  that, for every $L_1, L_2, L_3 \in \cD$, there exists some $L \in
  \cD$ such that
  \[L_1 \cap A^* = L \cap A^*, \qquad L_2 \cap (A^* \otimes \N) = L
    \cap (A^*\otimes \N), \qquad \text{and} \qquad L_3 \cap A_z = L
    \cap A_z.\]
  So, in particular, by taking $L_1 = L_3 = \emptyset$ and $L_2 = (A
  \times 2)^*$, we may conclude that $A^* \otimes \N \in \cD$.  To show 
  that~$\cD$ is closed under quotients, it suffices to consider quotients by 
  letters. Also, since quotient operations are homomorphisms, it suffices to
  consider languages belonging to each component. By hypothesis $\cD_0$
  and $\cD_1$ are closed under quotients by letters $a\in A$ and clearly so 
  is $2$ within $\cP(A_z)$. Also by hypothesis, $(a,1)^{-1}L,\ L(a,1)^{-1}\in
  \cD_0$ for $L\in\cD_1$. Finally, for $L\in\cD_0$ we have
  \[
    (a,1)^{-1} L=L(a,1)^{-1}=\emptyset\in\cD \qquad \text{and} \qquad
    (a,1)^{-1} A_z=A_z(a,1)^{-1}=(A^*\otimes\N)\cup A_z\in\cD.\popQED
  \]
\end{proof}
Notice that in the case where $\cD$ satisfies the equivalent
conditions of Lemma~\ref{l:12}, the Boolean subalgebra with atoms
$A^*$, $A^*\otimes\N$ and $A_z$ is a subalgebra of $\cD$. Also note
that the Boolean subalgebra of $\cP((A \times 2)^*)$ generated
by~$A^*$ is closed under quotients and it is such that $\cD_z = 2$,
but it does not satisfy the equivalent conditions of the the lemma.
\begin{corollary}\label{c:11}
  Let $\cC\subseteq \cP(A^* \cup (A^* \otimes \N))$ be a Boolean
  subalgebra. Then, the Boolean subalgebra~$\cD$ of~$\cP((A \times
  2)^*)$ closed under quotients generated by~$\cC$ is generated, as a
  lattice, by $\cS_0 \cup \cS_1 \cup \{A_z\}$, where
  \begin{multline*}
    \cS_0 = \{u^{-1}Lv^{-1} \mid u, v \in A^*, \ L \in \cC_0\} \\
    \cup\ \{(w,i)^{-1} L u^{-1},\ u^{-1}L(w,i)^{-1} \mid u \in A^*,
    (w,i)\in A^* \otimes \N, \ L \in\cC_1\} \subseteq \cP(A^*)
  \end{multline*}
  and
  \[\cS_1 = \{u^{-1}L v^{-1}\mid u, v \in A^*, \ L \in \cC_1\} \subseteq \cP(A^*
    \otimes \N).\]
\end{corollary}
\begin{proof}
  Since the quotient operations are Boolean homomorphisms,
  ${\cD}$ is generated as a Boolean algebra by the quotients
  of the languages in $\cC$. Also, since, for $a\in A$ and $L\subseteq
  A^* \cup (A^* \otimes \N)$, we have
  \[
  a^{-1}L\subseteq A^* \cup (A^* \otimes \N)\ \text{ and }\ (a,1)^{-1}L\subseteq A^*,
  \]
  it follows that ${\cD}_z=2$. Also, since $\cC$ is a
  Boolean subalgebra of $\cP(A^* \cup (A^*\otimes \N))$, it contains
  $A^* \cup (A^*\otimes \N)$ and thus ${\cD}$ contains $A_z$.  It
  follows that Lemma~\ref{l:12} applies to ${\cD}$ and the conclusion
  of the corollary easily follows.
\end{proof}

Another immediate consequence of the equivalence between \ref{item:1}
and \ref{item:4} in Lemma~\ref{l:12} is the following:

\begin{corollary}\label{c:9}
  Let $\cD \subseteq \cP((A \times 2)^*)$ be a Boolean subalgebra
  closed under quotients that contains $A^* \otimes \N$ and such that
  $\cD_z = 2$. We let $\pi: (A \times 2)^* \to M_\cD$ denote the
  syntactic morphism of~$\cD$ and we set
  \[M = \pi[A^*] \qquad\text{and}\qquad T= \pi[A^* \otimes \N].\]
  Then,
  \begin{enumerate}
  \item\label{item:2} the biaction of~$M_{\cD}$ on $X_{\cD}$ restricts
    and co-restricts to a biaction of $M$ on $X_{\cD_0}$ and to a
    biaction of $M$ on $X_{\cD_1}$;
  \item\label{item:3} for every $t \in T$, the components of the
    biaction of $M_\cD$ on itself at $t$ restrict and co-restrict as
    follows
    \[\lambda_{t}: X_{\cD_0} \to
      X_{\cD_1} \qquad \text{and}\qquad\rho_{t}: X_{\cD_0} \to
      X_{\cD_1}.\]
  \end{enumerate}
\end{corollary}
We remark that, if $\cD \subseteq \cP((A \times 2)^*)$ is a Boolean
subalgebra closed under quotients then~$\cD_0$, seen as a Boolean
subalgebra of $\cP(A^*)$, is also closed under quotients: it is so
because, for every $u, v \in A^*$ and $L \in \cD$, the following
equality holds:
\[u^{-1}(L \cap A^*) v^{-1} = (u^{-1}Lv^{-1}) \cap A^*.\]
Moreover, since the quotient $\cP((A \times 2)^*) \twoheadrightarrow
\cP(A^*)$ restricts and co-restricts to a quotient $\cD
\twoheadrightarrow \cD_0$, the syntactic morphism of $\cD_0 \subseteq
\cP(A^*)$ is a restriction and co-restriction of the syntactic
morphism of $\cD \subseteq \cP((A \times 2)^*)$. In particular, the
monoid $M$ defined in Corollary~\ref{c:9} is the syntactic
monoid~$M_{\cD_0}$ of~$\cD_0$ and the biaction of~$M$ on~$X_{\cD_0}$
is simply the natural biaction of~$M_{\cD_0}$ on~$X_{\cD_0}$.
\section{The Substitution Principle}\label{sec:substitution}
The concept of substitution for the study of logic on words, as laid
out by Tesson and Th\'erien in~\cite{TessonTherien07}, is quite
different from substitution in predicate logic.  Substitution in
predicate logic works on terms, whereas the notion of substitution
in~\cite{TessonTherien07} works at the level of predicates. As such it
provides a method for decomposing complex formulas into simpler
ones. The core idea is to enrich the alphabet over which the logic is
defined in order to be able to substitute large subformulas through
letter predicates.

In this section we start by defining substitution maps with respect to
finite Boolean algebras, and we prove a local version of the
Substitution Principle (cf. Corollary~\ref{c:4}), whose main
ingredient is the duality between finite sets and finite Boolean
algebras. Then, by proving that the substitution maps form a direct
limit system, we are able to extend the construction to arbitrary
Boolean algebras using full fledged Stone duality. In this process,
we are naturally led to consider profinite alphabets and we are able
to state a global version of the Substitution Principle
(cf. Corollary~\ref{cor:SP-global}). Finally, in Section~\ref{sec:app}, we show
how in practical terms these techniques may be useful in the study of
fragments of logic.

\subsection{Substitution with respect to finite Boolean
  algebras}\label{sec:sub-BA-fin}

As an example, consider the sentence $\psi = \exists x \
\phi(x)$. Then, $\psi$ may be obtained from the sentence $\exists x \
\tP_b(x)$ by replacing $\tP_b(x)$ by $\phi(x)$, and thus,
understanding $\psi$ amounts to understanding both the sentence
$\exists x \ \tP_b(x)$ and the formula $\phi(x)$.  If we want to
substitute away several subformulas in this way, we must account for
their logical relations. For instance, suppose that $\phi(x)$ is the
conjunction $\phi_1(x)\wedge\phi_2(x)$, and that $\phi_1(x)$ and
$\phi_2(x)$ are the simpler subformulas we wish to consider. Should
$\psi$ be obtained from a simpler sentence as above, such a sentence
would be $\exists x \ (\tP_{b_1}(x) \wedge\tP_{b_2}(x))$ for two
letters $b_1, b_2$. But, since $\phi_1$ and $\phi_2$ may be related we
would then need to impose relations on letters. Then, no complexity is
removed.  Thus instead, we will consider a \emph{finite Boolean
  algebra of formulas} to be substituted away. These can all be
accounted for by having \emph{a letter predicate for each atom} of
this finite Boolean algebra. The fact that letter predicates model the
atoms of a finite Boolean algebra of formulas in a free variable is
built into their interpretation.
This explains why, when substituting a formula of a given finite set
$\cF$ for each occurrence of a letter predicate from a corresponding
alphabet in a sentence, we should only consider sets $\cF$ of formulas
that have the same logical behavior as letter predicates, meaning
that $\cF$ satisfies
\begin{enumerate}[label = (A.\arabic*), leftmargin = *]
\item\label{item:A1} $\bigvee_{\phi \in \cF}\ \phi$ is the
  \emph{always-true} proposition;
\item\label{item:A2} for every $\phi_1, \phi_2 \in \cF$
  distinct, $\phi_1 \wedge \phi_2 $ is the \emph{always-false}
  proposition.
\end{enumerate}
In other words, we require that $\cF$ is the set of atoms of the
finite Boolean algebra it generates.

We now formalize this concept of substitution.
Throughout this section we fix a context~$\x$ and a variable~$x$ which
does not belong to~$\x$. Moreover, $\Gamma$ will be a fixed class of
sentences and $\Delta \subseteq \log \cQ A{\x x}\cN$ a finite Boolean
subalgebra.
We regard the set of atoms of $\Delta$ as a finite alphabet, and in
order to emphasize both the fact that it is an alphabet and the fact
that it is determined by~$\Delta$, we will denote it
by~$C_{\Delta}$. On the other hand, when we wish to view an element
$c$ of $C_{\Delta}$ as a formula of $\Delta$, we will write~$\phi_c$
instead of~$c$.
\begin{definition}
  The \emph{$\Gamma$-substitution given by $\Delta$ (with respect to
    the variable $x$)} is the map
  \[\sigma_{\Gamma, \Delta}: \Gamma(C_{\Delta}) \to\log\cQ A \x\cN\]
  sending a sentence to the formula in context~$\x$ obtained by
  substituting for each occurrence of a letter predicate $\tP_c(z)$,
  the formula $\phi_c[x/z]$ (that is, the formula obtained by
  substituting $z$ for $x$ in the formula $\phi_c\in
  \At(\Delta)$). Note that here we assume (without loss of generality)
  that the variables occurring in $\psi \in \Gamma(C_\Delta)$, such as
  $z$, do not occur in the formulas of $\Delta$. When~$\Gamma$ is
  clear from the context, as will very often be the case, we will
  simply write~$\sigma_\Delta$ instead of~$\sigma_{\Gamma, \Delta}$.
\end{definition}

Since the only constraints on the interpretation of letters in a word
are given by the properties~\ref{item:A1} and~\ref{item:A2}, it
follows by a simple structural induction that $\sigma_{\Delta}$ is a
homomorphism of Boolean algebras. We denote the image of this
homomorphism by~$\Gamma \vcirc \Delta$. In the sequel we will consider
$\sigma_\Delta$ as denoting the co-restriction of the substitution to
its image, that is,
\[\sigma_\Delta: \Gamma(C_\Delta) \twoheadrightarrow \Gamma \vcirc
  \Delta.\]

\begin{remark}
  Let us compare our substitution map with that
  of~\cite[Definition~2.6]{TessonTherien07}. In~\cite{TessonTherien07},
  given a finite set of formulas $\Phi \subseteq \cQ_{A, x}[\cN]$, the
  authors define the associated $\Phi$-substitution to be the function
  $\sigma_\Phi: \Gamma(2^\Phi) \to \cQ_A[\cN]$ that sends a formula
  $\psi \in \Gamma(2^\Phi)$ to the formula obtained from $\psi$ by
  replacing each letter predicate $\tP_S(z)$, with $S \subseteq \Phi$,
  by the formula $\varphi_S = \bigwedge_{\varphi \in S} \varphi[x/z]
  \wedge \bigwedge_{\varphi \notin S} \neg \varphi[x/z]$. Let $\Delta$
  be the Boolean subalgebra of $\cQ_{A, x}[\cN]$ generated by
  $\Phi$. It is not hard to see that, for every subset $S \subseteq
  \Phi$, either $\varphi_S$ is the always-false formula or it is an
  atom of~$\Delta$, and that every atom of $\Delta$ is semantically
  equivalent to the formula $\varphi_S$, where $S = \{\varphi \in \Phi
  \mid \phi \leq \varphi\}$. In particular, we have an embedding
  $\zeta: \At(\Delta) \rightarrowtail 2^\Phi$ mapping an atom $\phi
  \in \Delta$ to the set $\{\varphi \in \Phi \mid \phi \leq
  \varphi\}$. Finally, since $\Gamma$ is a class of sentences, $\zeta$
  induces a Boolean algebra quotient $\zeta_\Gamma:\Gamma(2^\Phi)
  \twoheadrightarrow \Gamma(C_\Delta)$ which, by definition of our
  substitution map $\sigma_\Delta$ and of the substitution map
  $\sigma_\Phi$ from~\cite{TessonTherien07}, is such that
  $\sigma_\Delta \circ \zeta_\Gamma = \sigma_\Phi$. In particular, the
  images of $\sigma_\Delta$ and of $\sigma_\Phi$ are the same.
\end{remark}

Next we describe the languages of $\Gamma \vcirc \Delta$ via those of
$\Gamma$ and those of $\Delta$.  Since $\Gamma(C_\Delta)$ and
$\Gamma\vcirc \Delta$ define, respectively, a Boolean algebra of
languages over~$C_\Delta$ and a Boolean algebra of marked words
over~$A$, we have embeddings

\begin{center}
  \begin{tikzpicture}[->,auto]
    \node[xshift = -5mm](G) at (0,0) {$\Gamma(C_\Delta)$};
    \node(GL) at (3,0) {$\Gamma \vcirc \Delta$};
    \node[xshift = -5mm](PC) at (0, -1.5) {$\cP(C_\Delta^*)$};
    \node(PA) at (3, -1.5) {$\cP(A^* \otimes \N^{ \x})$}; \draw[>->] (G) to (PC);
    \draw[>->] (GL) to (PA);
    \draw[->>] (G) to node[above, ArrowNode] {$\sigma_{\Delta}$} (GL);
  \end{tikzpicture}
\end{center}
Our first goal is to prove that the substitution map $\sigma_\Delta$
extends to a complete homomorphism of Boolean algebras
$\cP(C^*_\Delta)\to \cP(A^* \otimes \N^{ \x})$.  Formulated
dually, this means we are looking for a map on the level of (marked)
words $\tau_\Delta: A^* \otimes \N^{ \x} \to C^*_\Delta$ making
the following diagram commute:
\begin{center}
  \begin{tikzpicture}[->]
    \node[xshift = 10mm](XVG) at (6.5,0) {$X_{\Gamma(C_\Delta)}$};
    \node[xshift = 15mm](XVGL) at (9.5,0) {$X_{\Gamma \vcirc \Delta}$};
    \node[xshift = 10mm](C) at (6.5, -1.5) {$C_\Delta^*$};
    \node[xshift = 15mm](A) at (9.5,-1.5) {$A^* \otimes \N^{
        \x}$}; \draw[->] (C) to node[left, yshift = -2pt, ArrowNode]
    {$p_\Delta$}(XVG); \draw (A) to node[right, yshift = -2pt,
    ArrowNode] {$q_\Delta$} (XVGL);
    \draw[<-] (XVG) to node[above, ArrowNode] {$\Sigma_{\Delta}$}
    (XVGL);
    \draw[<-, dashed] (C) to node[below, ArrowNode] {$\tau_\Delta$} (A);
  \end{tikzpicture}
\end{center}
Here the maps $p_\Delta\colon C_{\Delta}^*\to X_{\Gamma{(C_{\Delta})}}$ and
$q_\Delta\colon A^* \otimes \N^{ \x} \to X_{\Gamma \vcirc \Delta}$ are,
respectively, the restrictions of the dual maps of the embeddings
$\Gamma{(C_{\Delta})}\rightarrowtail \cP(C_{\Delta}^*)$ and
$\Gamma\vcirc \Delta \rightarrowtail\cP(A^* \otimes \N^{ \x})$.

In order to be able to define the map $\tau_\Delta$, we need to
understand the Boolean algebra~$\Delta$ via duality. Recall that $\log
\cQ A{\x x}\cN$ embeds in $\cP(A^* \otimes \N^{ {\x x}})$ as a
Boolean subalgebra, and therefore, so does~$\Delta$:
\[\Delta \rightarrowtail \log \cQ A{\x x}\cN \rightarrowtail
  \cP(A^* \otimes \N^{ {\x x}}).\]
Applying discrete duality to this composition, we obtain a map
\[\xi_\Delta: A^* \otimes \N^{ {\x x}} \twoheadrightarrow C_\Delta =
  \At(\Delta) \]
defined, for $w\in A^* \otimes \N^{ \x}$, $i<|w|$, $c\in
C_\Delta$, and $\phi_c$ the atom of $\Delta$ corresponding to $c$, by
\begin{equation}
  \xi_\Delta(w, i) = c \ \iff \ (w, i)\in L_{\phi_c}
  \ \iff \ 
  (w, i)\vDash \phi_c .\label{eq:5}
\end{equation}

\begin{proposition}\label{p:transduccao}
  Let $\Gamma$ be a class of sentences, $\Delta$ a finite Boolean
  subalgebra of $\log\cQ A {\x x}\cN$, and
  $\sigma_\Delta:\Gamma(C_\Delta)\to\Gamma\vcirc\Delta$ the associated
  substitution as defined above.
  Then, the function ${\tau_\Delta: A^* \otimes \N^{ \x} \to
    C_{\Delta}^*}$ defined by $\tau_\Delta(w) = \xi_\Delta(w, 0)
  \cdots \xi_\Delta(w, \card w - 1)$ makes the following diagram
  commute:
  \begin{center}
  \begin{tikzpicture}[->]
    \node[xshift = -5mm](G) at (0,0) {$\Gamma(C_\Delta)$};
    \node(GL) at (3,0) {$\Gamma \vcirc \Delta$};
    \node[xshift = -5mm](PC) at (0, -1.5) {$\cP(C_\Delta^*)$};
    \node(PA) at (3, -1.5) {$\cP(A^* \otimes \N^{ \x})$}; \draw[>->] (G) to (PC);
    \draw[>->] (GL) to (PA);
    \draw[->>] (G) to node[above, ArrowNode] {$\sigma_{\Delta}$} (GL);
    \draw[->,dashed] (PC) to node[below, ArrowNode] {$\tau_\Delta^{-1}$}(PA);
  \end{tikzpicture}
\end{center}
\end{proposition}

\begin{proof}
 We show that the dual diagram   
  \begin{center}
    \begin{tikzpicture}[->]
      \node(A) at (1.5,0) {$A^* \otimes \N^{ \x}$}; \node(SGB) at (5,1.5)
      {$X_{\Gamma{(C_{\Delta}})}$}; \node(SGA) at (1.5,1.5)
      {$X_{\Gamma \vcirc \Delta}$}; \node(SB) at (5,0)
      {$C_{\Delta}^*$}; \draw[->,dashed] (A) to node[below, ArrowNode]
      {$\tau_\Delta$} (SB); \draw[->] (SGA) to node[above, ArrowNode]
      {$\Sigma_\Delta$} (SGB); \draw[->] (A) to node[left, ArrowNode]
      {$p_\Delta$} (SGA); \draw[->] (SB) to node[right, ArrowNode]
      {$q_\Delta$} (SGB);
    \end{tikzpicture}
 \end{center}
 commutes. To this end, let $w \in A^* \otimes \N^{ \x}$,
 $i<|w|$, and $c\in C_\Delta$ with $\phi_c$ the corresponding atom
 of~$\Delta$. First, we argue that
 \[(\tau_\Delta(w), i) \models \tP_c(x) \ \iff \ (w, i) \models
   \phi_c.\]
 This is so because, by the definition of letter predicates, the
 marked word over $C_\Delta^*$, $(\tau_\Delta(w), i)$, is a model of
 $\tP_c(x)$ if and only if its $i$-th letter is a~$c$. By definition
 of~$\tau_\Delta$, this is equivalent to having $\xi_\Delta(w, i) =
 c$, which, by~\eqref{eq:5}, means that $(w, i) \models \phi_c$, as
 required.

  Now, since the validity in a marked word of a quantified formula $Qx\
  \psi$ is fully determined once we know the truth value of the given
  formula $\psi$ at each point of the marked word, and since
  $\psi\in\Gamma(C_\Delta)$ and $\sigma_\Delta(\psi)$ are built up
  identically once the substitutions of $\tP_c(x)$ by $\phi_c$ have
  been made, it follows that, for all $\psi\in\Gamma(C_\Delta)$, we
  have
\begin{equation}\label{eq:TT}
  \tau_\Delta(w) \in L_\psi   \  \iff \    w \in L_{\sigma_\Delta(\psi)}.
\end{equation}
However
\[
  w \in L_{\sigma_\Delta(\psi)}\quad\iff\quad L_{\sigma_\Delta(\psi)}\in
  p_\Delta(w)\quad\iff\quad L_\psi\in \Sigma_\Delta( p_\Delta(w))
\]
so that
\[
  \tau_\Delta(w) \in L_\psi \ \iff \ L_\psi\in \Sigma_\Delta(
  p_\Delta(w))
\]
and thus $\Sigma_\Delta( p_\Delta(w))=q_\Delta(\tau_\Delta(w))$ as required.
\end{proof}

The existence of the map $\tau_\Delta$ defined in
Proposition~\ref{p:transduccao} yields the next result.
\begin{corollary}[Substitution Principle - local version]\label{c:4}
  Let $\Gamma$ be a class of sentences and $\Delta \subseteq \log\cQ
  A{\x x}\cN$ be a finite Boolean subalgebra. Then,
  the languages definable by a formula of $\Gamma \vcirc \Delta$ are
  precisely those of the form $\tau_{\Delta}^{-1}(K)$, where $K$ is a
  language definable in $\Gamma{(C_{\Delta})}$.
\end{corollary}
\begin{proof}
  This is an immediate consequence of the commutativity of the diagram
  dual to that of Proposition~\ref{p:transduccao}.
\end{proof}

\begin{remark}\label{r:1}
  We warn the reader that the operator $(\_\vcirc\_)$ just defined
  does not coincide with the operator $(\_\circ\_)$ considered both
  in~\cite{BorlidoCzarnetzkiGehrkeKrebs17} and in the regular version
  of~\cite{TessonTherien07}. The relationship between  these operators
  is expressed by the following equality:
  \[\Gamma \circ \Delta = \langle (\Gamma \vcirc \Delta) \cup
    \Delta_\x\rangle_{\sf BA},\]
  where $\Delta_\x$ denotes the set of formulas of $\Delta$ in
  context~$\x$ (i.e., those whose free variables belong to~$\x$).
  We choose to first study the operator~$(\_\vcirc\_)$ in order to
  emphasize the role of Stone duality in the \emph{Substitution
    Principle}. It should be clear for the reader how to state the
  corresponding results for~$(\_\circ\_)$.
\end{remark}

\subsection{The extension to arbitrary Boolean algebras using
  profinite alphabets}\label{sec:inf-BA}

As the reader may have noticed, the context~$\x$ fixed along the
previous section is not playing an active role. For that reason, and
in order to simplify the notation, we now assume that $\x$ is the
empty context. Later, in Section~\ref{sec:encode}, we will see that
this assumption may be done without loss of generality.

Here we show that substitution as defined in
Section~\ref{sec:sub-BA-fin} extends to arbitrary Boolean algebras in
a meaningful way. Fix a class of sentences $\Gamma$. We start by
comparing the substitution maps obtained for two finite Boolean
algebras, one contained in the other.  To this end, suppose $\Delta_1
\rightarrowtail \Delta_2$ is such an inclusion of finite Boolean
subalgebras of $\log\cQ A{x}\cN$ and let $\zeta: C_2\twoheadrightarrow
C_1$ be the dual of the inclusion. Recall that the semantics of logic
on words provides embeddings of $\Gamma(C_i)$ in $\cP(C_i^*)$ and
that, since $\Gamma$ is a class of sentences, the surjection $\zeta$
yields an embedding $\zeta_\Gamma\colon\Gamma(C_1)\rightarrowtail
\Gamma(C_2)$ making the following diagram commute
(cf. diagram~\eqref{eq:6} and~\ref{item:LC2}):
\begin{equation}
  \begin{aligned}
    \begin{tikzpicture}
      [node distance = 18mm, ->] \node at (0,0) (gc1) {$\Gamma(C_1)$};
      \node[below of = gc1] (gc2) {$\Gamma(C_2)$}; \node[right of =
      gc1, xshift = 10mm] (pc1) {$\cP(C_1^*)$}; \node[below of = pc1]
      (pc2) {$\cP(C_2^*)$};
      \draw[>->] (pc1) to node[right, ArrowNode] {$(\zeta^*)^{-1}$}
      (pc2); \draw[>->] (gc1) to (pc1); \draw[>->] (gc2) to (pc2);
      \draw[>->] (gc1) to node[left, ArrowNode] {$\zeta_\Gamma$}
      (gc2);
    \end{tikzpicture}
  \end{aligned}\label{eq:9}
\end{equation}
We also have:
\begin{lemma}\label{l:2}
  The following  diagram is commutative:
  \begin{center}
    \begin{tikzpicture}[node distance = 20mm, ->]
      \node (pc2) {$C_2^*$}; \node[below of = pc2] (pc1) {$C_1^*$};
      \draw[->>] (pc2) to node[left, ArrowNode] {$\zeta^*$} (pc1);
      \node[right of = pc2, yshift = -10mm] (pa)
      {$A^*$};
      \draw (pc1) to node[below, yshift = -1mm, ArrowNode]
      {$\tau_{\Delta_1}$} (pa); \draw (pc2) to node[above, ArrowNode]
      {$\tau_{\Delta_2}$} (pa);
    \end{tikzpicture}
  \end{center}
\end{lemma}
\begin{proof}
  We first recall from Proposition~\ref{p:transduccao} that, for $i =
  1, 2$, and $w \in A^*$, we have
  \[\tau_{\Delta_i}(w) = \xi_i(w, 0) \dots \xi_i(w, \card w -1),\]
  where $\xi_i: A^* \otimes \N \twoheadrightarrow C_i$ is the quotient
  dual to the embedding $\Delta_i \rightarrowtail \cP(A^* \otimes
  \N)$.  On the other hand, by hypothesis, we have a commutative
  diagram
  \begin{center}
    \begin{tikzpicture}[>=stealth', >->, node distance = 18mm]
      \node(B) {$\Delta_1$}; \node[below of = B](C) {$\Delta_2$};
      \node[right of = B, yshift = -9mm, xshift = 10mm] (Sub)
      {$\cP(A^* \otimes \N)$};
      \draw (B) to (C); \draw (B) to (Sub); \draw (C) to (Sub);
    \end{tikzpicture}
  \end{center}
  which dually gives
  \begin{center}
    \begin{tikzpicture}[>=stealth', <<-, node distance = 18mm]
      \node(B) {$C_1$}; \node[below of = B](C) {$C_2$}; \node[right of
      = B, yshift = -9mm, xshift = 10mm] (Sub){$A^* \otimes \N$};
      \draw (B) to node[left, ArrowNode] {$\zeta$}(C); \draw (B) to
      node[above, ArrowNode] {$\xi_1$} (Sub); \draw (C) to node[below,
      ArrowNode] {$\xi_2$} (Sub);
    \end{tikzpicture}
  \end{center}
  It then follows that
  \begin{align*}
    \tau_{\Delta_1}(w)
    & = \xi_1(w, 0) \dots \xi_1(w, \card w  - 1)
    \\ & = \zeta^*(\xi_2(w, 0) \dots \xi_2(w, \card w - 1)) =
         \zeta^* (\tau_{\Delta_2}(w)),
  \end{align*}
  that is, the statement of the lemma is true.
\end{proof}

As a straightforward consequence of diagram~\eqref{eq:9},
Lemma~\ref{l:2}, and Proposition~\ref{p:transduccao}, we have the
following:
\begin{theorem}\label{t:5}
  Let $\Delta_1 \rightarrowtail \Delta_2$ be an embedding of finite
  Boolean subalgebras of $\log\cQ A{x}\cN$, and let $\zeta: C_2
  \twoheadrightarrow C_1$ be its dual map. Then, the following diagram
  commutes
  \begin{center}
    \begin{tikzpicture}[node distance = 15mm, ->]
      \node at (0,0) (gc1) {$\Gamma(C_1)$}; \node[below of = gc1,
      yshift = -30mm] (gc2) {$\Gamma(C_2)$}; \node[below of = gc1,
      xshift = 20mm] (pc1) {$\cP(C_1^*)$}; \node[below of = pc1] (pc2)
      {$\cP(C_2^*)$};
      \draw[>->] (pc1) to node[left, ArrowNode] {$(\zeta^*)^{-1}$}
      (pc2); \draw[>->] (gc1) to (pc1); \draw[>->] (gc2) to (pc2);
      \draw[>->] (gc1) to node[left, ArrowNode] {$\zeta_\Gamma$}
      (gc2);
      \node[right of = pc1, yshift = -7.5mm, xshift = 15mm] (pa)
      {$\cP(A^*)$};
      \draw (pc1) to node[above, ArrowNode] {$\tau_{\Delta_1}^{-1}$}
      (pa); \draw (pc2) to node[below, ArrowNode]
      {$\tau_{\Delta_2}^{-1}$} (pa);
      \node[right of = pa, xshift = 15mm] (qan) {$\cQ_A[\cN]$};
      \draw[>->] (qan) to (pa); \draw[bend left = 20] (gc1) to
      node[above, ArrowNode] {$\sigma_{\Delta_1}$} (qan); \draw[bend right = 20]
      (gc2) to node[below, ArrowNode] {$\sigma_{\Delta_2}$} (qan);
    \end{tikzpicture}
  \end{center}
\end{theorem}

\begin{corollary}
  \label{c:3}
  If $\Delta_1 \subseteq \Delta_2$ are finite Boolean subalgebras of
  $\log\cQ A{x}\cN$, then the inclusion $\Gamma \vcirc \Delta_1
  \subseteq \Gamma \vcirc \Delta_2$ holds.  In particular, for every
  Boolean subalgebra $\Delta\subseteq \log\cQ A{x}\cN$, the family
  $\{\Gamma \vcirc \Delta' \mid \Delta' \subseteq \Delta\text{ is a
    finite Boolean subalgebra}\}$ forms a direct limit system.
\end{corollary}
\begin{proof}
  Let $\Delta_1 \rightarrowtail \Delta_2$ be an embedding of finite
  Boolean subalgebras of $\log\cQ A{x}\cN$, and $\zeta: C_2
  \twoheadrightarrow C_1$ its dual map. Commutativity of the outer
  triangle of Theorem~\ref{t:5} yields
  \begin{center}
    \begin{tikzpicture}[->,>=stealth',shorten >=1pt,auto,node
      distance=3.4cm, scale=1] \node(B) at (0,0) {$\Gamma(C_1)$};
      \node(C) at (0,-1.6) {$\Gamma(C_2)$}; \node(Sub) at (3,-0.8)
      {$\cQ_A[\cN]$};

      \draw[>->,>=stealth'] (B) to node[left, ArrowNode] {$\zeta_\Gamma$}
      (C); \draw[->,>=stealth'] (B) to node[above, ArrowNode]
      {$\sigma_{\Delta_1}$} (Sub); \draw[->,>=stealth'] (C) to
      node[below, ArrowNode] {$\sigma_{\Delta_2}$} (Sub);
    \end{tikzpicture}
  \end{center}
  Thus, we have a direct system of Boolean subalgebras of~$\cQ_A[\cN]$
  and $\Gamma\vcirc \Delta_1\subseteq\Gamma\vcirc \Delta_2$.
\end{proof}

This leads to the following definition.
\begin{definition}\label{def:circ}
  For an arbitrary Boolean subalgebra $\Delta \subseteq \log\cQ
  A{x}\cN$, we set
  \begin{equation}
    \Gamma \vcirc \Delta = \lim\limits_{\longrightarrow} \
    \{\Gamma\vcirc \Delta'\mid \Delta' \subseteq \Delta \text{ is a
      finite Boolean subalgebra}\}.\label{eq:1} 
  \end{equation}
  Note that $\Gamma \vcirc \Delta$ is simply given by the union of all
  Boolean subalgebras $\Gamma\vcirc \Delta' \subseteq \cQ_A[\cN]$. In
  particular, it is also a Boolean subalgebra of $\cQ_A[\cN]$.
\end{definition}

In turn, on the side of Boolean spaces, we have a projective limit
system formed by the maps~$\tau_{\Delta'}$.

\begin{corollary}\label{c:5}
  For every Boolean subalgebra $\Delta \subseteq \log\cQ A{x}\cN$, the
  set of maps
  \begin{equation}
  \{\tau_{\Delta'}: A^* \to C_{\Delta'}^* \mid \Delta' \subseteq
    \Delta\text{ is a finite Boolean subalgebra}\}\label{eq:2}
  \end{equation}
  forms a projective limit system, where the connecting morphisms are
  the homomorphisms of monoids $C_{\Delta_{2}}^* \twoheadrightarrow
  C_{\Delta_1}^*$ induced by the dual maps of inclusions $\Delta_1
  \rightarrowtail \Delta_2$ of finite Boolean subalgebras of $\Delta$.
  Moreover, the limit of this system is the map $\tau_\Delta: A^* \to
  X_{\Delta}^*$ sending a word $w \in A^*$ to the word $\gamma_0
  \cdots \gamma_{\card w - 1}$ with $\gamma_i = \{\phi \in \Delta \mid
  (w, i) \text{ satisfies } \phi\}$. In other words, $\gamma_i$ is the
  projection to $X_{\Delta}$ of the principal (ultra)filter of
  $\cP(A^*\otimes \N)$ generated by $(w, i)$.
\end{corollary}
\begin{proof}
  The fact that~\eqref{eq:2} forms a projective limit system follows
  from Lemma~\ref{l:2}.

  Now, we remark that the maps $\tau_{\Delta'}$ are all length
  preserving, thus the above system factors into projective limit
  systems $\{\tau_{\Delta'}: A^n \to (C_{\Delta'} )^n \mid \Delta'
  \subseteq \Delta\text{ is a finite Boolean subalgebra}\}$ for each
  $n \ge 0$, and each of these systems has itself a projective limit,
  which is induced by the projections $A^* \otimes \N \to
  X_\Delta$. Thus, the projective limit of~\eqref{eq:2} is the map
  $\tau_\Delta: A^* \to X_\Delta^*$ defined in the statement, where
  the space $X^*_{\Delta}$ is seen as the (disjoint) union over $n \ge
  0$ of the spaces $X_{\Delta}^n$, each one of those being a Boolean
  space when equipped with the product topology.
\end{proof}

The reader should now recall the extension of an lp-strain of
languages~$\cV$ to profinite alphabets presented in
Section~\ref{sec:profinite}, and in particular that such extension
yields an embedding $\cV(Y) \rightarrowtail \cl(Y^*)$ for every
profinite alphabet~$Y$.

\begin{corollary}\label{c:6}
  Let $\Gamma$ be a class of sentences and $\Delta \subseteq \log\cQ
  A{x}\cN$ be a Boolean subalgebra.  Then, the Boolean algebra $\Gamma
  \vcirc \Delta$ is a quotient of $\cV_\Gamma(X_{\Delta})$, or
  equivalently, $X_{\Gamma \vcirc \Delta}$ is isomorphic to a closed
  subspace of $X_{\cV_\Gamma(X_\Delta)}$. More precisely, there exist
  commutative diagrams
  \begin{center}
    \begin{tikzpicture}[node distance = 15mm, ->]
      \node (A1) {$\cl(X_\Delta^*)$}; \node[right of = A1, xshift =
      15mm] (XD1) {$\cP(A^*)$}; \node[below of = XD1] (XV1) {${\Gamma
          \vcirc \Delta}$}; \node[below of = A1] (XG1)
      {${\cV_\Gamma(X_\Delta)}$};
      \draw (A1) to node[above, ArrowNode]{ $\tau^{-1}_\Delta$} (XD1);
      \draw[<-<] (XD1) to (XV1); \draw[<-<] (A1) to (XG1); \draw[->>]
      (XG1) to (XV1);
      \node[right of = XD1, xshift = 15mm] (A) {$\beta(A^*)$};
      \node[right of = A, xshift = 15mm] (XD) {$\beta(X_\Delta^*)$};
      \node[below of = XD] (XV) {$X_{\cV_\Gamma(X_\Delta)}$};
      \node[below of = A] (XG) {$X_{\Gamma \vcirc \Delta}$};
      \draw (A) to node[above, ArrowNode]{$\beta\tau_\Delta$} (XD);
      \draw[->>] (XD) to (XV); \draw[->>] (A) to (XG); \draw[>->] (XG)
      to (XV);
    \end{tikzpicture}
  \end{center}
  which are dual to each other.
\end{corollary}
\begin{proof}
  We prove that the left-hand  diagram commutes.  Let $L$ be a
  language of $\Gamma\vcirc \Delta$. By Definition~\ref{def:circ} of
  $\Gamma \vcirc \Delta$, this means that there exists a finite
  Boolean subalgebra $\Delta' \subseteq \Delta$ such that $L$ belongs
  to $\Gamma \vcirc \Delta'$. Denote by $C_{\Delta'}$ the alphabet
  consisting of the atoms of~$\Delta'$. Using the local version of the
  Substitution Principle (cf. Corollary~\ref{c:4}), this is equivalent
  to the existence of some language $K \in \cV_\Gamma(C_{\Delta'})$
  such that $L = \tau_{\Delta'}^{-1}(K)$. But by Lemma~\ref{l:4} and
  Corollary~\ref{c:5}, this amounts to having that $L$ belongs to
  $\tau_\Delta^{-1}(\cV_\Gamma(X_\Delta))$. So, we just proved that
  the image of the following composition
   \begin{center}
     \begin{tikzpicture}[node distance = 30mm, >->]
       \node (A) {$\cV_\Gamma(X_\Delta)$}; \node[right of = A] (B) {$
         \cl(X_\Delta^*)$}; \node[right of = B] (C) {$\cP(A^*)$};
       \draw (A) to (B); \draw (B) to node [above, ArrowNode] {$\tau_\Delta^{-1}$} (C);
    \end{tikzpicture}
  \end{center}
  is precisely $\Gamma \vcirc \Delta$. That is, we have a commutative
  diagram as in the left-hand side of the statement.
\end{proof}

We can now state the already announced global version of the
Substitution Principle, which is simply a rephrasing of the
commutativity of the left-hand  diagram of Corollary~\ref{c:6}.

\begin{corollary}
  [Substitution Principle - global version] \label{cor:SP-global} Let
  $\Gamma$ be a class of sentences and $\Delta \subseteq \log\cQ
  A{x}\cN$ be a Boolean subalgebra.  Then, the languages (over~$A$)
  definable in $\Gamma \vcirc \Delta$ are precisely the languages of
  the form $\tau_\Delta^{-1}(K)$, where $K \in
  \cV_\Gamma(X_{\Delta})$.
\end{corollary}

\subsection{Using an alphabet to encode free
  variables}\label{sec:encode}

We now explain how free variables may be encoded in an alphabet,
provided the logic at hand is expressive enough. The idea is that a
formula~$\phi\in \log \cQ A x \cN$ may be regarded as a sentence over
the extended alphabet~$(A \times 2)$, in the same way we may regard
marked words as words over $(A \times 2)$ via the embedding $A^*
\otimes \N \rightarrowtail (A \times 2)^*$.
More generally, given two disjoint contexts~$\x$ and~$\y$, we will
define two functions
\[\varepsilon_{\x, \y} :\log \cQ {A} {\x\y} \cN \to \cQ_{A \times
    2^\x, \y}[\cN] \quad \text{and}\quad \delta_{\x,\y}: \cQ_{A \times
    2^\x, \y}[\cN] \to \log \cQ {A} {\x\y} \cN,\]
which we call, respectively, the \emph{encoding} and the
\emph{decoding} function. As the name suggests, $\varepsilon_{\x \y}$
``encodes'' the variables from $\x$ in the extended alphabet~$A \times
2^\x$, while $\delta_{\x,\y}$ reverts this process (see
Corollary~\ref{c:7} for a precise statement).

Given contexts $\x$ and $\y$, we let
\[\iota_{\x}:A^* \otimes \N^{ \x} \rightarrowtail (A \times
  2^\x)^*\]
denote the natural embedding of $ \x$-marked words into the set of
words over $(A \times 2^ \x)$, and
\[\iota_{\x,\y}: A^* \otimes \N^{ {\x\y}} \rightarrowtail (A \times 2^\x)^*
  \otimes \N^{ \y}\]
be defined by $\iota_{\x, \y}(w, {\bf i}, {\bf j}) = (\iota_\x(w, {\bf
  i}), {\bf j})$, for every $(w, {\bf i}, {\bf j}) \in A^* \otimes
\N^{ {\x\y}}$.

In order to keep the notation simple, we first consider the case where
$\x$ consists of a single variable~$x$.

To define the encoding function, we need to assume that the quantifier
$\exists!$ (\emph{there exists a unique}) and the numerical predicate
$=$ (formally given by $\{(n,n) \mid n \in \mathbb N\}$) are
expressible in our logic. This assumption is needed so that we are
able to define the set of all marked words $A^* \otimes \N$ as a
language over $(A \times 2)$. Indeed, one can easily check that the
models of the sentence
\begin{equation}
  \Phi= \exists !  z\ \bigvee_{a \in A} \tP_{(a,1)}(z)\label{eq:28}
\end{equation}
are precisely the words of $(A \times 2)^*$ having exactly one letter
of the form $(a, 1)$, that is, those that belong to the image of
$\iota_x: A^* \otimes \N \rightarrowtail (A \times 2)^*$. More
generally, if $\Phi$ is seen as a formula in a given context $\y$,
then its models are the $\y$-marked words in the image of $\iota_{x,
  \y}: A^* \otimes \N^{ {x\y}} \rightarrowtail (A \times 2)^* \otimes
\N^{ \y}$.

\begin{definition}\label{sec:enc}
  Let $\phi$ be a formula in context $x\y$ over the alphabet $A$, and
  assume that all occurrences of~$x$ in~$\phi$ are free. We denote by
  $\phi^x$ the formula obtained from~$\phi$ by:
  \begin{itemize}
  \item replacing each predicate $\tP_a(x)$ by $\exists z\
    \tP_{(a,1)}(z)$,
  \item replacing each predicate $\tP_a(z)$ by $\tP_{(a,1)} (z) \vee
    \tP_{(a, 0)}(z)$, if $z \neq x$,
  \item replacing each predicate $R(\_, x, \_)$ by $\exists z \
    (\bigvee_{a \in A}\tP_{(a,1)}(z) \wedge R(\_, z,\_))$.
  \end{itemize}

\end{definition}

By a routine structural induction on the construction of a formula, we
can show the following:
\begin{lemma}\label{l:3}
  Let $\phi$ be a formula in context $x\y$ over the alphabet
  $A$. Then, for every $w \in A^*$, $i \in \{0, \dots, \card w - 1\}$,
  and ${\bf j} \in \{0, \dots,\card w - 1\}^{\card \y}$, we have
  \[(w, i, {\bf j}) \models \phi \iff \iota_{x, \y}(w, i,{\bf j})
    \models \phi^x.\]
\end{lemma}
\begin{corollary}\label{c:1}
  Let $\Phi$ be the formula defined in~\eqref{eq:28}. Then, there is a
  well-defined lattice homomorphism
  \[\varepsilon_{x, \y}: \log \cQ A {x\y} \cN \to \log \cQ {A\times 2}
    \y \cN, \qquad \phi \mapsto \phi^x \wedge \Phi,\]
  which makes the following diagram commute:
  \begin{center}
    \begin{tikzpicture}[node distance = 20mm, ->]
      \node (A) {$ \log \cQ A {x\y} \cN$}; \node[right of = A, xshift
      = 20mm] (B) {$\cP(A^* \otimes \N^{{x \y}})$}; \node[below
      of = A] (C) {$\log \cQ {A \times 2} {\y} \cN$}; \node[below of = B] (D)
      {$\cP((A \times 2)^* \otimes \N^{ \y})$};
      \draw[>->] (A) to (B); \draw[>->] (C) to (D); \draw (A) to
      node[left,ArrowNode] {$\varepsilon_{x,\y}$} (C); \draw[->] (B)
      to node[right,ArrowNode] {$\iota_{x,\y}[\_]$} (D);
    \end{tikzpicture}
  \end{center}
  That is, for every $\phi \in \log \cQ A {x\y} \cN$, we have
  $L_{\varepsilon_{x,\y}(\phi)} = \iota_{x,\y}[L_{\phi}]$.  
\end{corollary}
\begin{proof}
  We first note that $\varepsilon_{x, \y}$ is well-defined, that is,
  that $\phi^x \wedge \Phi$ and $(\phi')^x \wedge \Phi$ are
  semantically equivalent formulas provided so are $\phi$ and
  $\phi'$. Indeed, since the models of $\Phi$, when seen as a formula
  in the context $\y$, are the $\y$-marked words over $(A \times 2)$
  that belong to the image of $\iota_{x, \y}$, that is a consequence
  of Lemma~\ref{l:3}. It is easily seen that $\varepsilon_{x, \y}$ is
  a lattice homomorphism. Commutativity of the diagram also follows
  from Lemma~\ref{l:3}.
\end{proof}
We call \emph{encoding function} (with respect to~$x, \y$) the map
$\varepsilon_{x \y}: \log \cQ A {x\y} \cN \to \log \cQ {A\times 2} \y
\cN$ of Corollary~\ref{c:1}. Note that, if $\y'$ is a context
containing $\y$ (so that $ \log \cQ A {x\y} \cN \subseteq \log \cQ {A}
{x\y'} \cN$ and $ \log \cQ {A \times 2} {\y} \cN \subseteq \log \cQ
{A\times 2} {\y'} \cN$) then $\varepsilon_{x, \y}$ is the restriction
and co-restriction of $\varepsilon_{x, \y'}$.

Conversely, a formula $\psi$ over~$(A \times 2)$ in context~$\y$ may
be regarded as a formula over~$A$ in context~$x\y$ as follows:

\begin{definition}
  \label{sec:dec}
  Let $\psi$ be a formula in context $\y$ over the alphabet $A \times
  2$. We denote by $\psi_x$ the formula obtained from $\psi$ by
  replacing:
  \begin{itemize}
  \item each predicate $\tP_{(a,1)}(z)$ by $\tP_a(z) \wedge (x = z)$,
  \item each predicated $\tP_{(a, 0)}(z)$ by $\tP_a(z) \wedge (x \neq
    z)$.
  \end{itemize}
  Of course, we are assuming that things have been arranged so that
  the variable~$x$ does not occur in~$\psi$.
\end{definition}

Again, a structural induction on construction of formulas shows the
following:
\begin{lemma}\label{l:9}
  Let $\psi$ be a formula in context $\y$ over the alphabet $(A \times
  2)$. Then, for every $w \in A^*$, $i \in \{0, \dots, \card w - 1\}$,
  and ${\bf j} \in \{0, \dots,\card w - 1\}^{\card \y}$, we have
  \[(w, i,{\bf j}) \models \psi_x \iff \iota_{x, \y}(w, i, {\bf j})
    \models \psi.\]
\end{lemma}

The next result is a straightforward consequence of Lemma~\ref{l:9}.
\begin{corollary}  \label{c:2}
  There is a well-defined Boolean algebra homomorphism
  \[\delta_{x,\y}: \log \cQ {A\times 2} \y \cN \to \log \cQ A {x\y}
    \cN, \qquad\psi \mapsto \psi_x,\]
  which makes the following diagram commute:
  \begin{center}
    \begin{tikzpicture}[node distance = 20mm, ->]
      \node (A) {$ \log \cQ {A\times 2} {\y} \cN$}; \node[right of =
      A, xshift = 20mm] (B) {$\cP((A \times 2)^* \otimes \N^{
          \y})$}; \node[below of = A] (C) {$\log \cQ A {x\y} \cN$};
      \node[below of = B] (D) {$\cP(A^* \otimes \N^{ {x\y}})$};
      \draw[>->] (A) to (B); \draw[>->] (C) to (D); \draw (A) to
      node[left,ArrowNode] {$\delta_{x,\y}$} (C); \draw[->] (B) to
      node[right,ArrowNode] {$\iota_{x,\y}^{-1}$} (D);
    \end{tikzpicture}
  \end{center}
  That is, for every $\psi \in \log \cQ {A\times 2} {\y} \cN$, we have
  $L_{\delta_{x,\y} (\phi)} = \iota_{x,\y}^{-1}(L_{\phi})$.
\end{corollary}
The \emph{decoding function} (with respect to $x, \y$) is then the map
$\delta_{x,\y}: \log \cQ {A\times 2} \y \cN \to \log \cQ A {x\y} \cN$
defined in Corollary~\ref{c:2}.  Similarly to what happens
for~$\varepsilon_{x, \y}$, if $\y'$ is a context containing $\y$, then
$\delta_{x, \y}(\psi)$ is the restriction and co-restriction
of~$\delta_{x, \y'}(\psi)$.

Let us now define $\varepsilon_{\x, \y}$ and $\delta_{\x, \y}$, for
all contexts~$\x$ and~$\y$. We write $\x = \{x_1, \dots, x_k\}$ and,
for each $i \in \{0, \dots, k\}$, we let $\x_i$ denote the context
$\{x_{1}, \dots, x_i\}$ (in particular, $\x_0$ denotes the empty
context and $\x_k = \x$). Then, for each $i \in \{1, \dots, k\}$, we
have defined
\[\varepsilon_{x_i, \x_{i-1}\y}: \log \cQ
  {A\times 2^{k-i}} {\x_i\y} \cN \to \log \cQ {A\times 2^{k-i+1}}
  {\x_{i-1}\y} \cN\] and
\[\delta_{x_i, \x_{i-1}\y}:\log \cQ {A\times 2^{k-i+1}} {\x_{i-1}\y} \cN
  \to \log \cQ {A\times 2^{k-i}} {\x_i\y} \cN.\]
An iteration of these maps yields two functions
\[
  \varepsilon_{\x, \y}: \log \cQ {A} {\x\y} \cN \to \cQ_{A \times
    2^\x, \y}[\cN] \quad \text{and}\quad \delta_{\x,\y}: \cQ_{A \times
    2^\x, \y}[\cN] \to \log \cQ {A} {\x\y} \cN,\]
which are given, respectively, by $\varepsilon_{\x, \y} =
\varepsilon_{x_1,\x_0\y} \circ \dots \circ
\varepsilon_{x_k,\x_{k-1}\y}$ and by $\delta_{\x,\y} =
\delta_{x_k,\x_{k-1}\y} \circ \dots \circ \delta_{x_1,\x_0\y}$. In
particular, $\varepsilon_{\x, \y}$ is a lattice homomorphism, while
$\delta_{\x, \y}$ is a Boolean algebra homomorphism.
Moreover, when $\x$ consists of a single variable~$x$, the functions
$\varepsilon_{x, \y}$ and $\delta_{x, \y}$ are precisely those
previously defined.

By Corollaries~\ref{c:1} and~\ref{c:2}, we have:

\begin{proposition}\label{p:2}
  Let $\x$ and $\y$ be disjoint contexts. Then, the following diagrams
  commute:
  \begin{center}
    \begin{tikzpicture}[node distance = 20mm, ->]
      \node (A) {$ \log \cQ A {\x\y} \cN$}; \node[right of = A, xshift
      = 20mm] (B) {$\cP(A^* \otimes \N^{ {\x\y}} )$}; \node[below
      of = A] (C) {$\log \cQ {A \times 2^\x} {\y} \cN$}; \node[below
      of = B] (D) {$\cP((A \times 2^\x)^* \otimes \N^{ \y})$};
      \draw[>->] (A) to (B); \draw[>->] (C) to (D); \draw (A) to
      node[left,ArrowNode] {$\varepsilon_{\x,\y}$} (C); \draw[->] (B)
      to node[right,ArrowNode] {$\iota_{\x,\y}[\_]$} (D);
      %
      %
      \node[right of = B, xshift = 20mm] (A') {$ \log \cQ {A\times
          2^\x} {\y} \cN$}; \node[right of = A', xshift = 20mm] (B')
      {$\cP((A \times 2^\x)^* \otimes \N^{\y})$}; \node[below of
      = A'] (C') {$\log \cQ {A} {\x\y} \cN$}; \node[below of = B']
      (D') {$\cP(A^* \otimes \N^{ {\x\y}})$};
      \draw[>->] (A') to (B'); \draw[>->] (C') to (D'); \draw (A') to
      node[left,ArrowNode] {$\delta_{\x,\y}$} (C'); \draw[->] (B') to
      node[right,ArrowNode] {$\iota_{\x,\y}^{-1}$} (D');
    \end{tikzpicture}
  \end{center}
\end{proposition}

An immediate, but important, consequence of Proposition~\ref{p:2} is
the following:
\begin{corollary}\label{c:7}
  For every $\phi\in \log \cQ A {\x\y} \cN$ and $\psi\in\log \cQ
  {A\times 2^\x} {\y} \cN$, we have the following:
  \begin{enumerate}
  \item $\phi$ is semantically equivalent to
    $\delta_{\x,\y}\varepsilon_{\x,\y}(\phi)$,
  \item the models of $\psi$ that belong to $\im(\iota_{\x,\y})$ are
    precisely the models of $\varepsilon_{\x,\y}\delta_{\x,\y}(\psi)$.
  \end{enumerate}
  In particular, $\varepsilon_{\x, \y}$ is a lattice embedding, and
  $\delta_{\x, \y}$ is a quotient of Boolean algebras.
\end{corollary}

We finally show that it suffices to consider substitution with respect
to a single variable.

\begin{proposition}\label{p:3}
  Let $\Gamma$ be a class of sentences, and let $\Delta \subseteq
  \log\cQ A{\x x}\cN$ be a finite Boolean subalgebra. Then, there
  exist a finite Boolean subalgebra $\Delta' \subseteq \log\cQ {A
    \times 2^\x} x \cN$, and an embedding $\zeta:\At(\Delta)
  \rightarrowtail \At(\Delta')$, for which the following diagram
  commutes:
  \begin{center}
    \begin{tikzpicture}[node distance = 40mm]
      \node (A) at (0,0) {$\Gamma(C_{\Delta'})$}; \node[right of = A]
      (B) {$\cQ_{A \times 2^\x}[\cN]$}; \node[below of = A, yshift =
      25mm] (C) {$\Gamma(C_\Delta)$}; \node[right of = C] (D)
      {$\cQ_{A, \x}[\cN]$};
      \draw[->] (A) to node[above, ArrowNode] {$\sigma_{\Delta'}$}
      (B); \draw[->>] (A) to node[left, ArrowNode] {$\zeta_\Gamma$}
      (C); \draw[->] (C) to node[below, ArrowNode] {$\sigma_\Delta$}
      (D); \draw[->] (B) to node[right, ArrowNode] {$\delta_{\x,
          \emptyset}$} (D);
    \end{tikzpicture}
  \end{center}
  where $C_{\Delta} = \At(\Delta)$ and $C_{\Delta'} =
  \At(\Delta')$. In particular, since $\zeta_\Gamma$ is surjective, we
  have:
  \[\Gamma \vcirc \Delta = \delta_{\x, \emptyset}[\Gamma \vcirc
    \Delta'].\]
\end{proposition}
\begin{proof}
  We take $\Delta' = \{\psi \in \log\cQ {A \times 2^\x} x \cN \mid
  \delta_{\x x}(\psi) \in \Delta\}$. Since $\delta_{\x x}$ is a
  Boolean algebra homomorphism, $\Delta'$ is a Boolean
  algebra. Moreover, $\delta_{\x x}$ restricts and co-restricts to a
  Boolean algebra quotient $\delta_{\x x}': \Delta' \twoheadrightarrow
  \Delta$. We let $\zeta: \At(\Delta) \rightarrowtail \At(\Delta')$ be
  its dual. Let us show that the diagram of the statement
  commutes. Let $w \in A^* \otimes \N^{ \x}$ and $\phi \in
  \Gamma(C_{\Delta'})$. Then, we have
  \[w \models \sigma_\Delta \circ \zeta_\Gamma(\phi) \just \iff
    {\text{Proposition~\ref{p:transduccao}}} \tau_\Delta(w) \models
    \zeta_\Gamma(\phi) \just \iff {\eqref{eq:6}} \zeta^* \circ
    \tau_\Delta(w) \models \phi,\] and
  \[w \models \delta_{\x, \emptyset} \circ \sigma_{\Delta'}(\phi)
    \just \iff {\text{Proposition~\ref{p:2}}} \iota_{\x}(w) \models
    \sigma_{\Delta'}(\phi)
    \just\iff{\text{Proposition~\ref{p:transduccao}}}
    \tau_{\Delta'}\circ \iota_{\x}(w) \models \phi.\]
  Therefore, the diagram commutes provided the equality
  \begin{equation}
    \label{eq:33}
    \zeta^* \circ \tau_\Delta(w) = \tau_{\Delta'}\circ \iota_{\x}(w)
  \end{equation}
  holds. We first note that both sides of this equality are words over
  $C_{\Delta'} = \At(\Delta')$ of length $\card w$. So, we fix $i \in
  \{0, \dots, \card w -1\}$. We let $c_\ell$ and $c_r$ denote,
  respectively, the $i$-th letter of the left and of the right-hand
  side of~\eqref{eq:33}. We also let $\psi_{c_\ell}$ and $\psi_{c_r}$
  denote the atoms of~$\Delta'$ represented by the letters $c_\ell$
  and $c_r$, respectively. Now, since $\zeta$ is dual to $\delta_{\x,
    x}$, we have that $\psi_{c_\ell}$ is determined by the condition
  $(\tau_\Delta(w))_i \leq \delta_{\x, x}(\psi_{c_\ell})$ where,
  by~\eqref{eq:5}, $(\tau_\Delta(w))_i$ the unique atom of~$\Delta$
  satisfying $(w,i)\models (\tau_\Delta(w))_i$. In particular, we have
  $(w, i) \models \delta_{\x, x}(\psi_{c_\ell})$. On the other hand,
  by~\eqref{eq:5}, $\psi_{c_r}$ is the unique atom of~$\Delta'$
  satisfying $(\iota_{\x}(w), i) \models \psi_{c_r}$. Since
  $(\iota_{\x}(w), i) = \iota_{\x, x}(w,i)$, by Proposition~\ref{p:2},
  this is equivalent to $(w, i) \models \delta_{\x,
    x}(\psi_{c_r})$. Thus, we have $\psi_{c_\ell} = \psi_{c_r}$ as
  required.
\end{proof}

Using the definition of~$(\Gamma\vcirc \Delta)$ when $\Delta$ is an
arbitrary Boolean algebra, we may easily derive the following result
from Proposition~\ref{p:3} and from its proof:
\begin{corollary}\label{c:8}
  Let $\Gamma$ be a class of sentences, and let $\Delta \subseteq
  \log\cQ A{\x x}\cN$ be a Boolean subalgebra. Then, there exists a
  Boolean subalgebra $\Delta' \subseteq \log\cQ {A \times 2^\x} x
  \cN$, namely, $\Delta' = \{\psi \in \log\cQ {A \times 2^\x} x \cN
  \mid \delta_{\x, x}(\psi) \in \Delta\}$, such that
  \[\Gamma \vcirc \Delta = \delta_{\x, \emptyset}[\Gamma \vcirc
    \Delta'].\]
\end{corollary}
\subsection{Applications to logic on words}\label{sec:app}
In this section we show the utility of the Substitution Principle in
handling the application of one layer of quantifiers to a Boolean
algebra of formulas. As we will see, this allows for a systematic
study of fragments of logic via smaller subfragments
(cf. Proposition~\ref{prop:ApplSubst} below). Before stating it, we
show how we can assign to a given set of quantifiers a class of
sentences.
\begin{lemma}\label{l:1}
  Let $\cQ$ be a set of quantifiers. The assignment
  \[\Gamma_\cQ: A \mapsto  \cQ_A[\emptyset] =
    \left\langle \left\{\tQ x \bigvee_{a \in B} \tP_a(x) \mid \tQ \in
        \cQ, \ B \subseteq A\right\}\right\rangle_{\sf BA}\]
  defines a class of sentences.
\end{lemma}
\begin{proof}
  By definition, $\Gamma_\cQ(A)$ is a Boolean algebra, so
  Property~\ref{item:LC1} holds. To check~\ref{item:LC2} it is enough
  to observe that replacing a letter predicate by a disjunction of
  letter predicates yields formulas without numerical
  predicates. Thus, every map $\zeta: A \to B$ defines a homomorphism
  $\zeta_\Gamma: \cQ_{B}[\emptyset] \to \cQ_{A}[\emptyset]$.
\end{proof}

The following is an immediate consequence of the definitions
of~$\Gamma_\cQ$ and of~$(\_\vcirc\_)$, and it shows that
$\Gamma_\cQ$-substitutions model the application of a layer of
quantifiers to a given Boolean algebra.
\begin{lemma}\label{l:8}
  Let $A$ be a finite alphabet, $\mathcal Q$ a set of quantifiers,
  $\mathcal N$ a set of numerical predicates, and $\Delta \subseteq
  \cQ_{A, x}[\cN]$ a Boolean subalgebra. Then,
  \[\Gamma_\cQ\vcirc \Delta = \langle\{\tQ x\ \phi \mid \phi \in
    \Delta\}\rangle_{\sf BA}.\]
\end{lemma}

Finally, an iteration of Lemma~\ref{l:8} yields the following:
\begin{proposition}
  \label{prop:ApplSubst}
  Let $A$ be a finite alphabet, $\mathcal Q$ a set of quantifiers, and
  $\mathcal N$ a set of numerical predicates. For each $n > 0$, let
  $\Delta_n$ be the Boolean algebra of quantifier-free formulas in the
  context $\x = \{x_1, \ldots, x_n\}$.  Then, a sentence of
  $\cQ_A[\cN]$ has quantifier-depth $n$ if and only if it belongs to
  \[\underbrace{\Gamma_{\cQ} \vcirc (\dots \vcirc(\Gamma_{\cQ}}_{n\text{
        times}} \vcirc \ \Delta_n)\dots).\]
  In particular, we have
  \[\mathcal Q_A [\mathcal N] = \bigcup_{n \in \N}
    \underbrace{\Gamma_{\cQ} \vcirc (\dots \vcirc(\Gamma_{\cQ}}_{n\text{
        times}} \vcirc \ \Delta_n)\dots).\]
\end{proposition}
\begin{proof}
  Let $n$ be a positive integer. For each $k \geq 0$, we denote
  \[\Gamma_\cQ^k \vcirc \Delta_{n} = \underbrace{\Gamma_{\cQ} \vcirc (\dots
      \vcirc(\Gamma_{\cQ}}_{k\text{ times}} \vcirc \
    \Delta_n)\dots).\]
  We prove by induction on~$k\ge 0$ that $\Gamma_\cQ^k \vcirc
  \Delta_{n}$ consists of the formulas of quantifier-depth~$k$ in the
  context $\{x_{k+1}, \dots,x_n\}$. The case $k = 0$ follows from the
  definition of~$\Delta_n$. Suppose the claim is valid for~$k$, and
  let ${\bf x} = \{x_{k+2}, \dots, x_n\}$. By induction hypothesis, we
  have $\Gamma_\cQ^k \vcirc \Delta_n \subseteq \log \cQ A {{\bf
      x}x_{k+1}} \cN$. Thus, by Corollary~\ref{c:8}, the Boolean
  algebra $\Delta= \{\psi \in \log\cQ {A \times 2^\x} {x_{k+1}} \cN
  \mid \delta_{\x, x_{k+1}}(\psi) \in \Gamma_\cQ^k \vcirc \Delta_n\}$
  is such that
  \[\Gamma_\cQ^{k+1} \vcirc \Delta_n = \Gamma_\cQ \vcirc (\Gamma_\cQ^k \vcirc \Delta_n) =
    \delta_{\x, \emptyset}\ [\Gamma_\cQ \vcirc \Delta].\]
  Since $\delta_{\x, \emptyset}$ is a homomorphism of Boolean
  algebras, by Lemma~\ref{l:8}, it follows that $\Gamma_\cQ^{k+1}
  \vcirc \Delta_n$ is the Boolean algebra generated by the formulas of
  the form $\delta_{\x, \emptyset}(Q x_{k+1}\ \psi)$, where $\psi$
  belongs to~$\Delta$. Finally, since
  \[ \delta_{\x, \emptyset}(Q x_{k+1}\ \psi) = Q x_{x+1} \ \delta_{\x,
      x_{k+1}}(\psi),\]
  by definition of~$\Delta$ and by induction hypothesis, we may
  conclude that $\Gamma_\cQ^{k+1}\vcirc \Delta_n$ consists of the
  quantifier-depth $k+1$ formulas in the context $\{x_{k+2}, \dots,
  x_n\} = \x$ as we intended to show.
\end{proof}

\section{Semidirect products}\label{sec:closing-quot}
Let $\Gamma$ be a class of sentences and $\Delta \subseteq \log\cQ
A{x}\cN$ a Boolean subalgebra of formulas in one free variable. By
Corollary~\ref{cor:SP-global}, the languages definable by a formula of
$\Gamma \vcirc \Delta$ may be described using the Boolean algebra of
languages $\cV_\Gamma(X_\Delta)$, where $X_\Delta$ is the dual space
of~$\Delta$, and the map $\tau_\Delta: A^* \to X_{\Delta}^*$ as
defined in Corollary~\ref{c:5}. On the other hand, the
map~$\tau_\Delta$ depends only on the Boolean subalgebra $\cD_\Delta=
\{L_{\phi} \mid \phi \in \Delta\} \subseteq \cP(A^* \otimes \N)$ of
the languages definable by a formula in~$\Delta$.  The following
definition is an abstract version of this construction.

\begin{definition}\label{sec:recogn-form-obta-1}
  Let $\cC \subseteq \cP(A^* \otimes \N)$ be a Boolean subalgebra and
  $\cW \subseteq \cl(X_\cC^*)$ be a Boolean subalgebra of languages
  over~$X_\cC$,  the dual space of $\cC$, viewed as a profinite alphabet. 
  Also, let $\pi: A^* \otimes \N \to X_{\cC}$ be the (restriction of the) dual 
  of the embedding $\cC\rightarrowtail \cP(A^* \otimes \N)$.  Define 
  the map
  $\tau_{\cC}: A^* \to X_{\cC}^*$ by
  \[\tau_{\cC}(w) = \pi(w, 0) \dots \pi(w, \card w - 1),\]
  and let
  \[\cW \vcirc \cC =\{\tau_{\cC}^{-1}(K) \mid K \in
    \cW\}.\]
\end{definition}
Clearly, $\cW \vcirc \cC$ is a Boolean subalgebra of~$\cP(A^*)$ and,
by duality, we have that the dual space of~$\cW \vcirc \cC$ is a
closed subspace of~$X_{\cW}$ given by the dual of the map
$\tau_{\cC}^{-1}\colon \cW\to\cP(A^*).$ That is, we have commutative
diagrams
\begin{center}
  \begin{tikzpicture}[node distance = 15mm, ->]
    \node (A1) {$\cl(X_{\cC}^*)$}; \node[right of = A1, xshift = 15mm]
    (XD1) {$\cP(A^*)$}; \node[below of = XD1] (XV1) {${\cW \vcirc
        \cC}$}; \node[below of = A1] (XG1) {${\cW}$};
    \draw (A1) to node[above, ArrowNode]{ $\tau^{-1}_\cC$} (XD1);
    \draw[<-<] (XD1) to (XV1); \draw[<-<] (A1) to (XG1); \draw[->>]
    (XG1) to (XV1);
    \node[right of = XD1, xshift = 15mm] (A) {$\beta(A^*)$};
    \node[right of = A, xshift = 15mm] (XD) {$\beta(X_\cC^*)$};
    \node[below of = XD] (XV) {$X_{\cW}$};
    \node[below of = A] (XG) {$X_{\cW \vcirc \cC}$};
    \draw (A) to node[above, ArrowNode]{$\beta\tau_\cC$} (XD);
    \draw[->>] (XD) to (XV); \draw[->>] (A) to (XG); \draw[>->] (XG)
    to (XV);
  \end{tikzpicture}
\end{center}
which are dual to each other. Comparing these diagrams with those of
Corollary~\ref{c:6}, we readily see that
Definition~\ref{sec:recogn-form-obta-1} is indeed an abstraction of
the construction in Section~\ref{sec:inf-BA}.  More precisely, by
Corollary~\ref{cor:SP-global}, as discussed above, if $\cC=\cD_\Delta$
and $\cW=\cV_\Gamma(X_\Delta)$, then $\cW \vcirc \cC$ is the Boolean
algebra of languages over $A$ given by the logic fragment $\Gamma
\vcirc \Delta$.  Apart from this logic application, which is our focus
here, note that this definition is useful more widely in the theory of
formal languages, see~\cite{Weil1992} for some examples.

Our purpose now is to obtain a description of the \bim{} dual to the
Boolean algebra closed under quotients generated by a Boolean algebra
of the form $\cV(X_\cC) \vcirc \cC$, for an lp-variety of
languages~$\cV$ and~$\cC$ coming from a Boolean subalgebra closed
under quotients of $\cP((A \times 2)^*)$ as studied in
Section~\ref{sec:3}.

Thus, we let~$\cD$ be a Boolean subalgebra of $\cP((A \times 2)^*)$ which 
is closed under the quotient operations and is generated by the union of its 
retracts $\cD_0\subseteq\cP(A^*)$ and $\cD_1\subseteq\cP(A^*\otimes\N)$ 
as considered in Lemma~\ref{l:12}.
Recall that $\cD_0$ and $\cD_1$ are, respectively, the following
quotients of~$\cD$:
\[
  \cD_0=\{L\cap A^*\mid L\in \cD\} \qquad \text{and}\qquad
  \cD_1=\{L\cap (A^* \otimes\N)\mid L\in \cD\}.
\]
Let $\pi: (A \times 2)^* \to X_\cD$ be the restriction to $(A \times
2)^*$ of the map dual to the embedding $\cD \rightarrowtail \cP((A
\times 2)^*)$. As in Section~\ref{sec:3}, we denote
\[M= \pi[A^*]\qquad \text{and}\qquad T = \pi[A^* \otimes \N].\]
These are, respectively, a dense monoid in $X_{\cD_0}$ and a dense
subset of~$X_{\cD_1}$.  By Corollary~\ref{c:9}\ref{item:2}, the
natural biaction of the syntactic monoid~$M_{{\cD}}$ of~$\cD$
on~$X_{{\cD}}$ restricts and co-restricts to a biaction of~$M$
on~$X_{\cD_0}$ and to a biaction of~$M$ on~$X_{\cD_1}$. In particular,
$M$ is a monoid quotient of $A^*$, and $T$ comes equipped with a
biaction of~$M$.

Now, for each $m \in M_\cD$, we let~$\lambda_m$ and~$\rho_m$ denote,
respectively, the left and right components of the biaction
of~$M_{{\cD}}$ on~$X_{{\cD}}$. Then, there is also a biaction of~$M$
on~$X_{\cD_1}^*$ with continuous components given by
\begin{equation}
  \lambda_m^* (x_0 \dots x_{k-1})  = \lambda_m(x_0)
  \dots \lambda_m(x_{k-1}) \quad \text{and}\quad    \rho_m^*(x_0 \dots x_{k-1})  = \rho_m(x_0)
  \dots \rho_m(x_{k-1})\label{eq:15}
\end{equation}
for every $m \in M$ and $x_0 \dots x_{k-1}\in X_{\cD_1}^*$. Moreover,
since $T$ is invariant under the biaction of $M$ on $X_{\cD_1}$, the
$T$-generated free monoid~$T^*$ is invariant under the above biaction
of $M$ on $X_{\cD_1}^*$, thereby yielding a monoid biaction of~$M$
on~$T^*$.  Thus, we have a well-defined semidirect product $T^*{**}M$
given by this monoid biaction,
cf. Section~\ref{ss:semidirect}. Explicitly, the multiplication
on~$T^*{**}M$ is given by
\begin{equation}
  (\underline t, m)(\underline t', m') = (\rho_{m'}(t_0) \dots
  \rho_{m'}(t_{k-1})\lambda_m(t_0') \dots \lambda_m(t_{\ell-1}'), \,
  mm'),\label{eq:19}
\end{equation}
for every $(\underline t, m) = (t_0 \dots t_{k-1}, m)$ and $(\underline
t', m') = (t_0' \dots t_{\ell-1}', m')$ in $T^* {**} M$.

Our next goal is to show that $T^*{**}M$ has a monoid quotient that is
part of a \bim{} having $X_{\cV(X_{\cD_1})} \times X_{\cD_0}$ as space
component. This \bim{} is relevant because, via a suitable morphism,
it exactly recognizes the lattice $\cV\circ \cD$ generated by the
languages in $\cV(X_{\cD_1}) \vcirc \cD_1$ and in~$\cD_0$
(cf. Theorem~\ref{t:2}). We proceed in two steps. First we show that
the multiplication on $T^*{**}M$ naturally extends to a biaction of
$T^*{**}M$ on $\beta(X_{\cD_1}^*) \times X_{\cD_0}$ with continuous
components, so that the inclusion
\[T^*{**}M \rightarrowtail \beta(X_{\cD_1}^*) \times X_{\cD_0}\]
has a \bim{} structure (cf. Lemma~\ref{l:6} and
Proposition~\ref{p:4}). Then, for every Boolean subalgebra $\cW$ of
$\cl(X_{\cD_1}^*)$, we have a continuous quotient $\eta:
\beta(X_{\cD_1}^*) \twoheadrightarrow X_\cW$, and hence, a continuous
quotient $(\eta \times id): \beta(X_{\cD_1}^*) \times
X_{\cD_0}\twoheadrightarrow X_\cW \times X_{\cD_0}$.  The second step
is then to show that, if $\cW = \cV(X_{\cD_1})$ for some
lp-variety~$\cV$, then $(\eta \times id)$ defines a \bim{} quotient
\begin{equation}  \begin{aligned}
    \begin{tikzpicture}[node distance = 40mm] \node (A) at (0,0)
      {$\beta(X_{\cD_1}^*) \times X_{\cD_0}$}; \node [left of = A] (B)
      {$T^*{**}M$}; \node [below of = A, yshift = 20mm] (C)
      {$X_{\cV(X_{\cD_1})} \times X_{\cD_0}$}; \node [left of = C] (D)
      {$N$};
      \draw[>->] (B) to (A); \draw[>->] (D) to node[above, ArrowNode]
      {$\iota$} (C); \draw[->>] (B) to (D); \draw[->>] (A) to
      node[right, ArrowNode] {$\eta \times id$} (C);
      \end{tikzpicture}
    \end{aligned}\label{eq:26}
\end{equation}
where $N$ denotes the image of $T^*{**}M$ under $(\eta \times id)$,
that is, $N = (\eta \times id)[T^*{**}M]$, and $\iota$ is the
inclusion map. This is a consequence of Corollary~\ref{l:14}. In
Lemma~\ref{l:11} we show that $N$ is given by a semidirect product
$\eta[T^*]{**} M$. The \bim{} $(\eta[T^*]{**} M, \, \iota, \,
X_{\cV(X_{\cD_1})} \times X_{\cD_0})$ will be denoted by $\rvd$.  For
the reader's convenience, we list the notation just introduced in the
following table:

\begin{table}[H]
\centering
\resizebox{\textwidth}{!}{%
\begin{tabular}{@{}|c|c|@{}}
  \toprule
  \textbf{Notation}
  & \textbf{Definition}
  \\ \bottomrule
  $\cD_0$                                                                       & $\{L \cap A^* \mid L \in \cD\}$ (this is a quotient of $\cD$)                 \\ \midrule
  $\cD_1$                                                                       & $\{L \cap (A^*\otimes\N) \mid L \in \cD\}$ (this is a quotient of $\cD$)      \\ \midrule
  $\cV \circ \cD$
  & lattice generated by $\cV(X_{\cD_1})\vcirc \cD_1 \cup\cD_0$ (cf. Definition~\ref{sec:recogn-form-obta-1}) \\ \midrule
  $\pi:(A\times2)^* \to X_\cD$               & restriction to $(A \times 2)^*$ of dual of the embedding $\cD \rightarrowtail \cP((A\times 2)^*)$ \\ \midrule
  $M$                                                                           & $\pi[A^*]$ (this is a dense monoid in $X_{\cD_0}$)                            \\ \midrule
  $T$                                                                           & $\pi[A^*\otimes \N]$ (this is a dense subset of $X_{\cD_1}$)                  \\ \midrule
  $\eta: \beta(X_{\cD_1}^*) \twoheadrightarrow X_{\cV(X_{\cD_1})}$ & dual of the embedding $\cV(X_{\cD_1}) \rightarrowtail \cl(X_{\cD_1}^*)$                           \\ \midrule
  $\iota: \eta[T^*] \times M \rightarrowtail X_{\cV(X_{\cD_1})} \times
  X_{\cD_0}$ & inclusion map (this has dense image)
  \\ \midrule $\rvd$ & $(\eta[T^*] {**}
                                          M,\, \iota,\, X_{\cV(X_{\cD_1})} \times
                                          X_{\cD_0})$ (this is a
                                          \bim{}, cf. Proposition~\ref{p:6})\\ \bottomrule
\end{tabular}%
}
\caption{Notation}
\label{tab:my-table}
\end{table}

For the first part, we first observe that, using the equality
$\rho_{m'}(t) = \lambda_t(m')$ valid for every $t, m' \in M_\cD$ (and
in particular, for every $t \in T$ and $m' \in M$), \eqref{eq:19} may
be rewritten as follows:
\[(\underline t, m)(\underline t', m') = (\lambda_{t_0}(m') \dots
  \lambda_{t_{k-1}}(m')\lambda_m(t_0') \dots \lambda_m(t_{\ell-1}'), \,
  mm').\]
This provides a natural way of defining an element
$\widetilde\lambda_{(\underline t, m)}(\underline x, x_0)$ when
$(\underline x, x_0)$ belongs to $X_{\cD_1}^* \times X_{\cD_0}$,
namely,
\begin{equation}
  \tilde\lambda_{(\underline t, m)}(\underline x, x_0) =
  (\lambda_{t_1}(x_0) \dots \lambda_{t_k}(x_0) \lambda_{m}(x_1) \dots
  \lambda_{m}(x_\ell),\, \lambda_{m}(x_0)),\label{eq:21}
\end{equation}
where $\underline x = x_1 \dots x_\ell$. Similarly, we define
\begin{equation}
  \tilde\rho_{(\underline t, m)}(\underline x, x_0) = (\rho_{m}(x_1) \dots
  \rho_{m}(x_\ell)\rho_{t_1}(x_0) \dots \rho_{t_k}(x_0),\,
  \rho_{m}(x_0)).\label{eq:24}  
\end{equation}

\begin{lemma}\label{l:6}
  The functions
  \begin{align*}
    \tilde\lambda_{(\underline t, m)}: X_{\cD_1}^* \times X_{\cD_0}
    & \to X_{\cD_1}^* \times X_{\cD_0}
    \\ (x_1 \dots x_\ell, x_0)
    & \mapsto (\lambda_{t_0}(x_0) \dots \lambda_{t_{k-1}}(x_0)
      \lambda_{m}(x_1) \dots \lambda_{m}(x_\ell), \, \lambda_{m}(x_0)),
  \end{align*}
  and
  \begin{align*}
    \tilde\rho_{(\underline t, m)}: X_{\cD_1}^* \times X_{\cD_0}
    & \to X_{\cD_1}^* \times X_{\cD_0}
    \\ (x_1 \dots x_\ell, x_0)
    & \mapsto (\rho_{m}(x_1) \dots
      \rho_{m}(x_\ell)\rho_{t_0}(x_0) \dots \rho_{t_{k-1}}(x_0),\,
      \rho_{m}(x_0)),
  \end{align*}
  define a biaction of $T^*{**}M$ on the space $X_{\cD_1}^* \times
  X_{\cD_0}$.
\end{lemma}
\begin{proof}
  This is a consequence of the fact that the family $\{\lambda_m,
  \rho_m\}_{m \in M_{\cD}}$ defines a biaction of $M_\cD$ on
  $X_\cD$. We illustrate the computations involved by showing that
  $\tilde\lambda_{(\underline t, m)} \circ \tilde \lambda_{(\underline
    t', m')} = \tilde \lambda_{(\underline t, m)(\underline t', m')}$
  for every $(\underline t, m), (\underline t', m') \in T^*{**}M$, and
  leave the rest for the reader. Let us write $\underline t = t_0
  \dots t_{k-1}$ and $\underline t' = t'_0 \dots t'_{\ell-1}$, and
  pick $(\underline x, x_0) \in X_{\cD_1}^* \times X_{\cD_0}$. We will
  use the following notation: if $f_0, \dots, f_{i-1}$ are functions
  $P \to Q$ then, for every $p \in P$, we denote by $\langle f_0,
  \dots, f_{i-1}\rangle(p)$ the word $f_0(p) \dots f_{i-1}(p)$
  over~$Q$. Then, we may compute
  \begin{align*}
    \tilde\lambda_{(\underline t, m)} \circ \tilde \lambda_{(\underline
    t', m')} (\underline x, x_0)
    & =   \tilde\lambda_{(\underline t, m)} ( \langle \lambda_{t_0'},
      \dots, \lambda_{t_{\ell-1}'}\rangle(x_0)\, \lambda^*_{m'}(\,\underline
      x\,), \lambda_{m'}(x_0))
    \\ & = (\langle \lambda_{t_0}, \dots \lambda_{t_{k-1}}\rangle
         (\lambda_{m'}(x_0)) \lambda_m^*( \langle \lambda_{t_0'},
         \dots, \lambda_{t_{\ell-1}'}\rangle(x_0)\, \lambda^*_{m'}(\,\underline
         x\,)), \lambda_m(\lambda_{m'}(x_0)))
    \\ & = (\langle \lambda_{t_0}\lambda_{m'}, \dots
         \lambda_{t_{k-1}}\lambda_{m'}, \lambda_m\lambda_{t_0'},
         \dots, \lambda_m\lambda_{t_{\ell-1}'}\rangle(x_0)\, \lambda_m^*(\lambda^*_{m'}(\,\underline
         x\,)), \lambda_m(\lambda_{m'}(x_0)))
    \\ & \just ={(*)}  (\langle \lambda_{t_0 m'}, \dots
         \lambda_{t_{k-1} m'}, \lambda_{m t_0'},
         \dots,\lambda_{m t_{\ell-1}'}\rangle(x_0)\, \lambda^*_{mm'}(\,\underline
         x\,), \lambda_{mm'}(x_0))
    \\  & = \tilde \lambda_{(\underline t, m)(\underline t',
          m')}(\underline x, x_0).
  \end{align*}
  Here, every equality follows from the appropriate definitions, except for the
  equality marked with $(*)$, which uses the fact that $\{\lambda_m\}_{m \in
    M_{\cD}}$ is itself a left action.
\end{proof}

We now see that the biaction defined in Lemma~\ref{l:6}
further extends to a biaction of $T^*{**}M$ on $\beta(X_{\cD_1}^*)
\times X_{\cD_0}$ with continuous components, thereby defining a
\bim{}
\[T^* {**} M \rightarrowtail \beta(X_{\cD_1}^*) \times X_{\cD_0}.\]
Indeed, that may be seen as a consequence of the following technical
lemma.
\begin{lemma}\label{l:19}
  Let $L_0 \in \cD_0$ and $L_1, \dots, L_n \in \cD_1$. Then, for every
  $\underline t = t_0 \dots t_{k-1} \in T^*$ and $m \in M$, the
  following equalities hold
  \begin{equation}
    \label{eq:18}
    \widetilde \lambda_{(\underline t, m)}^{-1} (\widehat L_1 \times
    \dots \times \widehat L_n \times \widehat L_0) =
    \begin{cases} \lambda_m^{-1}(\widehat L_{k+1}) \times \dots \times
      \lambda_m^{-1}(\widehat L_n) \times (\lambda_m^{-1}(\widehat
      L_0) \cap \bigcap_{i = 1}^k \lambda_{t_{i-1}}^{-1}(\widehat
      L_i)), \quad \text{if $k \leq n$;}\\ \emptyset, \quad \text{
        otherwise;}
    \end{cases}
  \end{equation} and
  \[ \widetilde \rho_{(\underline t, m)}^{\,-1} (\widehat L_1 \times
    \dots \times \widehat L_n \times \widehat L_0) =
    \begin{cases} \rho_m^{-1}(\widehat L_{1}) \times \dots \times
      \rho_m^{-1}(\widehat L_{n-k}) \times (\rho_m^{-1}(\widehat L_0)
      \cap \bigcap_{i = 1}^k \rho_{t_{i-1}}^{-1}(\widehat L_{n-k+i})),
      \; \text{if $k \leq n$;}\\ \emptyset, \quad \text{
        otherwise.}
    \end{cases}
  \]
  (Here, we adopt the convention that $\widehat L_{k+1} \times \dots
  \times \widehat L_n = \{\varepsilon\} = \widehat L_{1} \times \dots
  \times \widehat L_{n-k}$ when $k = n$.)

  In particular, the biaction of $T^*{**}M$ on $X_{\cD_1}^* \times
  X_{\cD_0}$ has continuous components.
\end{lemma}
\begin{proof} We will only show the first equality, as the second one
  is analogous. First note that the pair $(\underline x, x_0) \in
  X_{\cD_1}^* \times X_{\cD_0}$ belongs to the left-hand side
  of~\eqref{eq:18} if and only if
  \begin{equation}
    \label{eq:34}
    \lambda_{t_0}(x_0) \dots \lambda_{t_{k-1}}(x_0)
    \lambda^*_{m}(\underline x) \in \widehat L_1 \times \dots \times
    \widehat L_n \quad\text{ and }\quad \lambda_{m}(x_0) \in \widehat L_0.
  \end{equation}
  Thus, the left-hand side of~\eqref{eq:18} is empty unless $k \leq
  n$. When that is the case, \eqref{eq:34} is equivalent to
  \[\lambda_{t_{i-1}}(x_0) \in \widehat L_i, \text{ for $i = 1, \dots,
      k$,} \quad \lambda_m^*(\underline x) \in \widehat L_{k+1} \times
    \dots \times \widehat L_n, \quad \text{and}\quad \lambda_{m}(x_0)
    \in \widehat L_0,\] that is,
  \[(\underline x, x_0) \in \lambda_m^{-1}(\widehat L_{k+1}) \times
    \dots \times \lambda_m^{-1}(\widehat L_n) \times
    (\lambda_m^{-1}(\widehat L_0) \cap \bigcap_{i = 1}^k
    \lambda_{t_{i-1}}^{-1}(\widehat L_i)).\]
  This shows that~\eqref{eq:18} holds.
\end{proof}

\begin{proposition}\label{p:4}
  The biaction of $T^* {**} M$ on $X_{\cD_1}^* \times X_{\cD_0}$
  extends to a biaction of $T^*{**}M$ on $\beta(X_{\cD_1}^*) \times
  X_{\cD_0}$ with continuous components, and the inclusion $T^* {**} M
  \rightarrowtail \beta(X_{\cD_1}^*) \times X_{\cD_0}$ admits a \bim{}
  structure with respect to such biaction.
\end{proposition}
\begin{proof}
  By Lemma~\ref{l:19}, the biaction of $T^*{**}M$ on $X_{\cD_1}^*
  \times X_{\cD_0}$ has continuous components. Therefore, since
  $\beta(X_{\cD_1}^*) \times X_{\cD_0}$ and $\beta_0(X_{\cD_1}^*
  \times X_{\cD_0})$ are homeomorphic (cf. Corollary~\ref{c:12}), each
  of its components may be uniquely extended to a continuous
  endofunction of $\beta(X_{\cD_1}^*) \times X_{\cD_0}$. Since
  $X_{\cD_1}^* \times X_{\cD_0}$ is dense in $\beta(X_{\cD_1}^*)
  \times X_{\cD_0}$, it follows that $T^* {**} M \rightarrowtail
  \beta(X_{\cD_1}^*) \times X_{\cD_0}$ is a \bim{}, as required.
\end{proof}

The extensions of the maps defined in Lemma~\ref{l:6} to continuous
endofunctions of $\beta(X_{\cD_1}^*) \times X_{\cD_0}$ will also be
denoted by $\widetilde\lambda_{(\underline t, m)}$ and by
$\widetilde\rho_{(\underline t, m)}$, respectively. We will now give
an expression for $\widetilde\lambda_{(\underline t, m)} (\gamma,
x_0)$ and $\widetilde\rho_{(\underline t, m)} (\gamma, x_0)$ when
$(\gamma, x_0) \in \beta(X_{\cD_1}^*) \times X_{\cD_0}$.  Recall from
Section~\ref{sec:1} that the embedding $X_{\cD_1}^*\rightarrowtail
\beta(X_{\cD_1}^*)$ defines a \bim{}. We let $\{\widetilde \ell_w,
\widetilde r_w\}_{w \in X_{\cD_1}^*}$ be the components of the
biaction of $X_{\cD_1}^*$ on $\beta(X_{\cD_1}^*)$. On the other hand,
since $\lambda_m$ and $\rho_m$ are continuous functions and the
topological space $X^*_{\cD_1}$ is the union over $n \ge 0$ of the
product spaces~$X_{\cD_1}^n$, the maps $(\lambda_m^*)^{-1}$ and
$(\rho_m^*)^{-1}$ define endomorphisms of
$\cl(X_{\cD_1}^*)$. Therefore, $\lambda_m^*$ and $\rho_m^*$ extend to
continuous functions $\beta(\lambda_m^*): \beta(X_{\cD_1}^*) \to
\beta(X_{\cD_1}^*)$ and $\beta(\rho_m^*): \beta(X_{\cD_1}^*) \to
\beta(X_{\cD_1}^*)$ which are, respectively, the duals of
$(\lambda_m^*)^{-1}$ and of $(\rho_m^*)^{-1}$.

We now observe that, given $(\underline x, x_0) \in X_{\cD_1}^* \times
X_{\cD_0}$, we may rewrite~\eqref{eq:21} as follows
\[\tilde\lambda_{(\underline t, m)}(\underline x, x_0) = (\tilde
  \ell_{\lambda_{\underline t}(x_0)}\circ \lambda_m^*(\underline x),\,
  \lambda_{m}(x_0)),\]
where $\lambda_{\underline t}(x_0) = \lambda_{t_0}(x_0)\dots
\lambda_{t_{k-1}}(x_0)$. Therefore, the function $f:
\beta(X_{\cD_1}^*) \times X_{\cD_0} \to \beta(X_{\cD_1}^*) \times
X_{\cD_0}$ defined by
\[f(\gamma, x_0) = (\tilde \ell_{\lambda_{\underline t}(x_0)}\circ
  \beta(\lambda_m^*)(\gamma),\, \lambda_{m}(x_0)),\]
for every $(\gamma, x_0) \in \beta(X_{\cD_1}^*) \times X_{\cD_0}$, is
an extension of the restriction~$\tilde\lambda_{(\underline t, m)}:
X_{\cD_1}^* \times X_{\cD_0} \to X_{\cD_1}^* \times X_{\cD_0}$. Since
$X_{\cD_1}^* \times X_{\cD_0}$ is a dense subspace of
$\beta(X_{\cD_1}^*) \times X_{\cD_0}$ and
$\widetilde\lambda_{(\underline t, m)}: \beta(X_{\cD_1}^*) \times
X_{\cD_0} \to \beta(X_{\cD_1}^*) \times X_{\cD_0}$ is, up to
isomorphism, the map dual to $\widetilde\lambda^{-1}_{(\underline t,
  m)}: \cl (X_{\cD_1}^* \times X_{\cD_0}) \to \cl (X_{\cD_1}^* \times
X_{\cD_0})$, in order to show that $f$ coincides with
$\tilde\lambda_{(\underline t, m)}: \beta(X_{\cD_1}^*) \times
X_{\cD_0} \to \beta(X_{\cD_1}^*) \times X_{\cD_0}$, it suffices to
show that $f^{-1}$ induces a well-defined map $\cl (\beta(X_{\cD_1}^*)
\times X_{\cD_0}) \to \cl (\beta(X_{\cD_1}^*) \times X_{\cD_0})$
(thus, in particular, $f$ is continuous) that is isomorphic to
$\widetilde\lambda^{-1}_{(\underline t, m)}$. Given a clopen subset $K
\subseteq X_{\cD_1}^*$, we let $cl(K)$ denote the closure of~$K$ in
$\beta(X_{\cD_1}^*)$, so that the assignment $K \mapsto cl(K)$
establishes an isomorphism between $\cl(X_{\cD_1}^*)$ and
$\cl(\beta(X_{\cD_1}^*))$. Let us fix $L_1, \dots, L_n \in\cD_1$ and
$L_0 \in \cD_0$. Then, for every $\gamma \in \beta(X_{\cD_1}^*) =
X_{\cl(X_{\cD_1}^*)}$ and $x_0 \in X_{\cD_0}$, we have
\begin{align*}
  (\gamma, x_0) \in f^{-1}
  (cl(\widehat L_1 \times \dots \times \widehat L_n) \times \widehat L_0) 
  & \iff (\tilde
    \ell_{\lambda_{\underline t}(x_0)}\circ
    \beta(\lambda_m^*)(\gamma),\, \lambda_{m}(x_0)) \in cl(\widehat L_1 \times \dots \times \widehat L_n) \times \widehat L_0
  \\ & \iff \widehat L_1 \times \dots \times \widehat L_n \in \tilde
       \ell_{\lambda_{\underline t}(x_0)}\circ
       \beta(\lambda_m^*)(\gamma) \quad \text{and} \quad x_0 \in
       \lambda_m^{-1}(\widehat L_0)
  \\ & \iff (\lambda_m^*)^{-1}(\lambda_{\underline t}(x_0)^{-1}(\widehat L_1 \times \dots \times \widehat L_n))
       \in \gamma  \quad \text{and} \quad x_0 \in
       \lambda_m^{-1}(\widehat L_0),
\end{align*}
where $\lambda_{\underline t}(x_0)^{-1}(\widehat L_1 \times \dots
\times \widehat L_n)$ denotes the left quotient of $\widehat L_1
\times \dots \times \widehat L_n \subseteq X_{\cD_1}^*$ by the word
$\lambda_{\underline t}(x_0) \in X_{\cD_1}^*$.  Now, either $k \leq n$
and $\lambda_{\underline t}(x_0) \in \widehat L_1 \times \dots \times
\widehat L_k$, and then we have
\[(\lambda_m^*)^{-1}(\lambda_{\underline t}(x_0)^{-1}(\widehat L_1
  \times \dots \times \widehat L_k)) = \lambda_m^{-1}(\widehat
  L_{k+1}) \times \dots \times \lambda_m^{-1}(\widehat L_{n}),\]
or else,
\[(\lambda_m^*)^{-1}(\lambda_{\underline t}(x_0)^{-1}(\widehat L_1
  \times \dots \times \widehat L_k)) = \emptyset.\]
Therefore, $f^{-1} (cl(\widehat L_1 \times \dots \times \widehat L_n)
\times \widehat L_0) $ is non-empty if and only if $k \leq n$ and, in
that case, it is given by
\begin{equation}
  f^{-1} (cl(\widehat L_1 \times \dots
  \times \widehat L_n) \times
  \widehat L_0) = cl(\lambda_m^{-1}(\widehat L_{k+1}) \times \dots \times
  \lambda_m^{-1}(\widehat L_{n})) \times (\lambda_m^{-1}(\widehat L_0)
  \cap \bigcap_{i = 1}^k \lambda_{t_{i-1}}^{-1}(\widehat
  L_i)).\label{eq:16}    
\end{equation}
By Lemma~\ref{l:19}, we may then conclude that $f^{-1}$ induces a
function isomorphic to $\widetilde\lambda^{-1}_{(\underline t, m)}$,
as claimed.

Proceeding similarly with $\widetilde\rho_{(\underline t, m)}$, we may
conclude that, for every $(\gamma, x_0) \in \beta(X_{\cD_1}^*) \times
X_{\cD_0}$, the following equalities hold
\begin{equation} \tilde\lambda_{(\underline t, m)}(\gamma, x_0) =
  (\tilde \ell_{\lambda_{\underline t}(x_0)}\circ
  \beta(\lambda_m^*)(\gamma),\, \lambda_{m}(x_0)) \qquad \text{and}
  \qquad\tilde\rho_{(\underline t, m)}(\gamma, x_0) = (\tilde
  r_{\rho_{\underline t}(x_0)} \circ \beta(\rho_m^*) (\gamma),\,
  \rho_{m}(x_0)),\label{eq:22}
\end{equation}

As promised, we now prove that the continuous quotient $(\eta \times
id) :\beta(X_{\cD_1}^*) \times X_{\cD_0}\twoheadrightarrow
X_{\cV(X_{\cD_1})} \times X_{\cD_0}$ induces a \bim{} quotient as
in~\eqref{eq:26}. This is a consequence of the following lemma:
\begin{lemma}\label{l:20}
  If $L_1, \dots, L_n \in \cD_1$ are such that $\widehat L_1 \times
  \dots \times \widehat L_n$ belongs to $\cV(X_{\cD_1}) \subseteq
  \cl(X_{\cD_1}^*)$ then, for every $0 \leq k \leq n$, the language
  $\lambda_m^{-1}(\widehat L_{k+1}) \times \dots
  \lambda_m^{-1}(\widehat L_n)$ also belongs to
  $\cV(X_{\cD_1})$.
\end{lemma}
\begin{proof}
  Let us denote $K = \lambda_m^{-1}(\widehat L_{k+1}) \times \dots
  \times \lambda_m^{-1}(\widehat L_{n})$.  Since, by
  Corollary~\ref{c:9}, we have a continuous function $\lambda_m:
  X_{\cD_1} \to X_{\cD_1}$, by Lemma~\ref{l:17}, we have $K \in
  \cV(X_{\cD_1}) $ provided the language $L = \widehat L_{k+1} \times
  \dots \times \widehat L_{n} \subseteq X_{\cD_1}^*$ belongs to
  $\cV(X_{\cD_1})$. To show this, we pick a word $w \in \widehat L_1
  \times \dots \times \widehat L_{k}$. Then, since $ \widehat L_1
  \times \dots \times \widehat L_n$ belongs to $\cV(X_{\cD_1})$, by
  Lemma~\ref{l:4}, the quotient $w^{-1}( \widehat L_1 \times \dots
  \times \widehat L_n)$ also does. But this quotient is precisely
  $L$. Indeed, for every word $v \in X_{\cD_1}^*$, we have
  \[v \in w^{-1}( \widehat L_1 \times \dots \times \widehat L_n) \iff
    w v \in \widehat L_1 \times \dots \times \widehat L_n \iff v \in
    \widehat L_{k+1} \times \dots \times \widehat L_{n},\]
  where the last equivalence follows from the choice of~$w$.

\end{proof}
\begin{corollary}\label{l:14} Let $\cV$ be an lp-variety
  of languages. Then, for every $(\underline t, m) \in T^* {**} M$,
  there are continuous functions $\lambda_{(\underline t, m)}$ and
  $\rho_{(\underline t, m)}$ making the following diagrams commute:
  \begin{equation}
    \begin{aligned}
      \begin{tikzpicture}[node distance = 44mm, ->] \node (A) at (0,0)
        {$X_{\cV(X_{\cD_1})} \times X_{\cD_0}$}; \node [left of = A]
        (B) {$\beta(X_{\cD_1}^*) \times X_{\cD_0}$}; \node [below of = A,
        yshift = 20mm] (C) {$X_{\cV(X_{\cD_1})} \times X_{\cD_0}$};
        \node [left of = C] (D) {$\beta(X_{\cD_1}^*) \times X_{\cD_0}$};
        \draw[->>] (B) to node[above, ArrowNode] {$\eta \times id$}
        (A); \draw[->>] (D) to node[below, ArrowNode] {$\eta \times
          id$} (C); \draw (B) to node[left, ArrowNode] {$\tilde
          \lambda_{(\underline t, m)}$} (D); \draw[dashed] (A) to
        node[right, ArrowNode] {$\lambda_{(\underline t, m)}$} (C);
        \node[right of = A, xshift = 37mm] (A') at (0,0)
        {$X_{\cV(X_{\cD_1})} \times X_{\cD_0}$}; \node [left of = A']
        (B') {$\beta(X_{\cD_1}^*) \times X_{\cD_0}$}; \node [below of
        = A', yshift = 20mm] (C') {$X_{\cV(X_{\cD_1})} \times
          X_{\cD_0}$}; \node [left of = C'] (D')
        {$\beta(X_{{\cD_1}}^*) \times X_{\cD_0}$};
        \draw[->>] (B') to node[above, ArrowNode] {$\eta \times id$}
        (A'); \draw[->>] (D') to node[below, ArrowNode] {$\eta \times
          id$} (C'); \draw (B') to node[left, ArrowNode] {$\tilde
          \rho_{(\underline t, m)}$} (D'); \draw[dashed] (A') to
        node[right, ArrowNode] {$\rho_{(\underline t, m)}$} (C');
      \end{tikzpicture}
    \end{aligned}\label{eq:27}
  \end{equation}
  where $\tilde \lambda_{(\underline t, m)}$ and $\tilde
  \rho_{(\underline t, m)}$ are, respectively, the left and right
  components at $(\underline t, m)$ of the biaction of $T^*{**}M$ on
  $\beta(X_{\cD_1}^*)\times X_{\cD_0}$
  (cf. Proposition~\ref{p:4}). Moreover, the family
  $\{\lambda_{(\underline t, m)}, \rho_{(\underline t,
    m)}\}_{(\underline t, m)}$ defines a biaction of $T^*{**}M$ on
  $X_{\cV(X_{\cD_1})} \times X_{\cD_0}$.
\end{corollary}
\begin{proof}
  By Lemmas~\ref{l:19} and~\ref{l:20}, the Boolean algebra
  homomorphism
  \[\widetilde \lambda_{(\underline t, m)}^{-1}:
    \cl(\beta(X_{\cD_1}^*) \times X_{\cD_0}) \to
    \cl(\beta(X_{\cD_1}^*) \times X_{\cD_0})\]
  restricts and co-restricts to a Boolean algebra
  homomorphism
  \[\cl(X_{\cV(X_{\cD_1})} \times X_{\cD_0}) \to
    \cl(X_{\cV(X_{\cD_1})} \times X_{\cD_0}).\]
  Let $\lambda_{(\underline t, m)}$ be the dual of the latter. Then,
  $\lambda_{(\underline t, m)}$ is a continuous function that makes
  the left-hand side of~\eqref{eq:27} commute. Existence
  of~$\rho_{(\underline t, m)}$ is shown similarly.

  The fact that $\{\lambda_{(\underline t, m)}, \rho_{(\underline t,
    m)}\}_{(\underline t, m)}$ defines a biaction of $T^*{**}M$ on
  $X_{\cV(X_{\cD_1})} \times X_{\cD_0}$ follows from having that
  $\{\tilde\lambda_{(\underline t, m)}, \tilde\rho_{(\underline t,
    m)}\}_{(\underline t, m)}$ defines a biaction of $T^*{**}M$ on
  $\beta(X_{\cD_1}^*) \times X_{\cD_0}$, together with surjectivity of
  $(\eta \times id)$ and commutativity of the diagrams
  in~\eqref{eq:27}.
\end{proof}

By Corollary~\ref{l:14}, we have that the restriction of $(\eta \times
id): \beta(X_{\cD_1}^*) \times X_{\cD_0} \to X_{\cV(X_{\cD_1})} \times
X_{\cD_0}$ to $T^* {**} M$ is a morphism of sets with
$(T^*{**}M)$-biactions. Therefore, $N = (\eta \times id)[T^*{**}M]$
comes equipped with a monoid structure induced by the monoid structure
of $T^*{**} M$ and we have a \bim{} $\rvd = (N,\, \iota,\,
X_{\cV(X_{\cD_1})} \times X_{\cD_0})$ which is a quotient of $(T^*
{**} M \rightarrowtail X_{\cD_1}^* \times X_{\cD_0})$ as
in~\eqref{eq:26}. We will now give a precise description of the
monoid~$N$. Note that the underlying set of~$N$ is $\eta[T^*] \times
M$. Moreover, since $\cV$ is an lp-variety, by Lemma~\ref{l:4},
$\cV(X_{\cD_1}^*)$ is closed under quotients and thus,
$\eta[X_{\cD_1}^*]$ is a monoid and the restriction and co-restriction
of~$\eta$ to a map $X_{\cD_1}^* \twoheadrightarrow \eta[X_{\cD_1}^*]$
is a monoid quotient. Since $T^*$ is a submonoid of~$X_{\cD_1}^*$, we
have that $\eta[T^*]$ is also a monoid. We will show that~$M$ biacts
on~$\eta[T^*]$ and that $N$ is the semidirect product $\eta[T^*] {**}
M$ defined by this biaction (Lemma~\ref{l:11}). In fact, we show the
following slightly more general fact: $M$ biacts on
$\eta[X_{\cD_1}^*]$ and the ensuing semidirect product
$\eta[X_{\cD_1}^*]{**} M$ is a monoid quotient of $X^*_{\cD_1}{**}M$
(recall that $M$ biacts on the free monoid $X_{\cD_1}^*$ so that we
have a semidirect product $X_{\cD_1}^*{**} M$ of which $T^*{**} M$ is
a submonoid, cf.~\eqref{eq:15}).

\begin{lemma}\label{l:11}
  For every $m \in M$, the assignments
  \[\ell_m: \eta[X_{\cD_1}^*]
    \to \eta[X_{\cD_1}^*], \qquad \eta(\underline x) \mapsto
    \eta(\lambda_m^*(\underline x))\] and
  \[r_m: \eta[X_{\cD_1}^*] \to \eta[X_{\cD_1}^*], \qquad
    \eta(\underline x) \mapsto \eta(\rho_m^*(\underline x))\]
  are well-defined functions, which define a monoid biaction of $M$ on
  $\eta[X_{\cD_1}^*]$. In particular, the monoid $N$ is a semidirect
  product $\eta[X_{\cD_1}^*] {**} M$ whose multiplication is defined
  by
  \begin{equation}
    (\eta(\underline x), m) (\eta(\underline x'), m') =
    (\eta(\rho_{m'}^*(\underline x)\lambda_{m}^*(\underline x')),
    mm')\label{eq:20}
  \end{equation}
  for every $(\underline x, m), (\underline x', m') \in
  X_{\cD_1}^*{**} M$, and $(\eta \times id): X_{\cD_1}^*{**} M
  \twoheadrightarrow \eta[X_{\cD_1}^*]{**} M$ is a monoid quotient.
\end{lemma}
\begin{proof}
  We first show that $\ell_m$ is well-defined. Let $\underline x,
  \underline x' \in X_{\cD_1}^*$ be such that $\eta(\underline x) =
  \eta(\underline x')$. We need to show that
  $\eta(\lambda_m^*(\underline x)) = \eta(\lambda_m^*(\underline
  x'))$. By definition of dual map, having $\eta(\underline x) =
  \eta(\underline x')$ is equivalent to having that, for every $K \in
  \cV(X_{\cD_1})$,
  \[\underline x \in K \iff \underline x' \in K.\]
  Since $\lambda_m$ is a continuous function and $\cV$ is an lp-strain
  of languages, by Lemma~\ref{l:5}, we then have that, for every $K
  \in \cV(X_{\cD_1})$,
  \[\underline x \in (\lambda^*_m)^{-1}(K) \iff \underline x' \in
    (\lambda^*_m)^{-1}(K),\]
  and this is equivalent to having $\eta(\lambda_m^*(\underline x)) =
  \eta(\lambda_m^*(\underline x'))$ as required.

  Similarly, one can show that $r_m$ is well-defined. The fact that
  $\{\ell_m, r_m\}_{m \in M}$ defines a monoid biaction is inherited
  from the fact that $\{\lambda_m, \rho_m\}_{m \in M}$ defines a
  monoid biaction.

  Finally, the fact that the multiplication on
  $\eta[X_{\cD_1}^*]{**}M$ is given by~\eqref{eq:20} is a
  straightforward consequence of the definition of semidirect products
  (cf. Section~\ref{ss:semidirect}). To conclude that $\eta \times id$
  is a monoid quotient, it suffices to observe that~\eqref{eq:20} may
  be rewritten as
  \[ (\underline x, m)(\underline x', m') = (\eta\times id)
    (\rho_{m'}^*(\underline x)\lambda^*_m(\underline x'), mm') =
    (\eta\times id) ((\underline x, m)(\underline x', m')). \popQED\]
\end{proof}

We just finished proving the following:
\begin{proposition}\label{p:6}
  Let $\cV$ be an lp-variety of languages and $\cD \subseteq \cP((A
  \times 2)^*)$ be a Boolean subalgebra closed under quotients that is
  generated by the union of its retracts $\cD_0 \subseteq \cP(A^*)$
  and $\cD_1 \subseteq \cP(A^* \otimes \N)$. Then, we have a quotient
  of \bim{}s
  \begin{center}
    \begin{tikzpicture}[node distance = 40mm] \node (A) at (0,0)
      {$\beta(X_{\cD_1}^*) \times X_{\cD_0}$}; \node [left of = A] (B)
      {$T^*{**}M$}; \node [below of = A, yshift = 20mm] (C)
      {$X_{\cV(X_{\cD_1})} \times X_{\cD_0}$}; \node [left of = C] (D)
      {$\eta[T^*]{**}M$};
      \draw[>->] (B) to (A); \draw[>->] (D) to node[above, ArrowNode]
      {$\iota$} (C); \draw[->>] (B) to (D); \draw[->>] (A) to
      node[right, ArrowNode] {$\eta \times id$} (C);
    \end{tikzpicture}
  \end{center}
\end{proposition}
We remark that, by~\eqref{eq:22} and by Corollary~\ref{l:14}, we have that
the biaction of $\eta[T^*]{**}M$ on $X_{\cV(X_{\cD_1})} \times
X_{\cD_0}$ is defined by
\begin{equation}
  \label{eq:23}
  \lambda_{(\eta(\underline t), m)}(\eta(\gamma), x_0) =
  (\eta\circ \widetilde \ell_{\lambda_{\underline t}(x_0)} \circ
  \beta(\lambda_m^*)(\gamma), \lambda_m(x_0))
\end{equation}
and
\begin{equation}
  \label{eq:30}
  \rho_{(\eta(\underline t), m)}(\eta(\gamma), x_0) = (\eta \circ \tilde
  r_{\rho_{\underline t}(x_0)} \circ \beta(\rho_m^*) (\gamma),\,
  \rho_{m}(x_0))
\end{equation}
for every $(\underline t, m) \in T^* {**}M$ and $(\gamma, x_0) \in
\beta(X_{\cD_1}^*) \times X_{\cD_0}$. (Recall that, if $\underline t =
t_0 \dots t_{k-1}$ then $\lambda_{\underline t}(x_0) =
\lambda_{t_0}(x_0) \dots \lambda_{t_{k-1}}(x_0) \in X_{\cD_1}^*$, the
maps $\tilde \ell_w, \tilde r_w: \beta(X_{\cD_1}^*) \to
\beta(X_{\cD_1}^*)$ are the unique continuous extensions of the left
and right multiplication in the free monoid $X_{\cD_1}^*$ by the word
$w \in X_{\cD_1}^*$, and $\beta(\lambda_m^*), \beta(\rho_m^*):
\beta(X_{\cD_1}^*) \to \beta(X_{\cD_1}^*)$ are the unique continuous
extensions of the maps $\lambda_m^*, \rho_m^*: X_{\cD_1}^* \to
X_{\cD_1}^*$.)

We are ready to prove the main result of this section. The reader
should recall the notation listed in Table~\ref{tab:my-table}.

\begin{theorem}\label{t:2}
  Let $\cV$ be an lp-variety of languages and $\cD \subseteq \cP((A
  \times 2)^*)$ be a Boolean subalgebra closed under quotients that is
  generated by the union of its retracts $\cD_0 \subseteq \cP(A^*)$
  and $\cD_1 \subseteq \cP(A^* \otimes \N)$.  Let $h: A^* \to
  \eta[T^*] {**} M$ be the homomorphism defined by
  \[h(a) = (\eta\circ \pi(a,1), \pi(a)).\]
  Then, a language is recognized by the \bim{} $\rvd =
  (\eta[T^*]{**}M,\, \iota, \,X_{\cV(X_{\cD_1})} \times X_{\cD_0})$
  via~$h$ if and only if it belongs to $\cV \circ \cD$.
\end{theorem}
\begin{proof}
  We first show, by induction on the length of words, that $h(w) =
  (\eta \circ \tau_{\cD_1}(w), \pi(w))$, for every $w \in A^*$. This
  is trivially the case if $w$ is the empty word. Let $w \in A^*$ and
  $a \in A$ be a letter. Then, by induction hypothesis, we have
  \[h(wa) = h(w)h(a) = (\eta \circ \tau_{\cD_1}(w), \pi(w))(\eta\circ
    \tau_{\cD_1}(a), \pi(a)),\]
  and by definition of the multiplication on $(\eta[T^*]{**}M)$
  (cf.~\eqref{eq:20}), it follows that
  \[h(wa) = (\eta(\rho_{\pi(a)}^*(\tau_{\cD_1}(w))
    \lambda_{\pi(w)}^*(\tau_{\cD_1}(a))), \pi(wa)).\]
  Now, since $\tau_{\cD_1}(w) = \pi(w, 0) \dots \pi(w,\card w - 1)$ and
  $\tau_{\cD_1}(a) = \pi(a,1)$, we have
  \[\rho_{\pi(a)}^*(\tau_{\cD_1}(w)) = \rho_{\pi(a)}\circ\pi(w, 0)
    \dots \rho_{\pi(a)}\circ \pi(w,\card w - 1) = \pi(wa, 0) \dots
    \pi(wa,\card w - 1)\]
  and
  \[\lambda_{\pi(w)}^*(\tau_{\cD_1}(a)) = \lambda_{\pi(w)}(\pi(a,1))
    = \pi(w(a,1)) = \pi(wa, \card w ).\]
  Therefore,
  \[h(wa) = (\eta (\pi(wa, 0) \dots \pi(wa,\card w - 1) \pi(wa, \card
    w)), \pi(wa)) = (\eta \circ \tau_{\cD_1}(wa), \pi(wa)).\]

  Now, the languages recognized by $(\eta[T^*]{**}M\rightarrowtail
  X_{\cV(X_{\cD_1})} \times X_{\cD_0})$ via $h$ are precisely the
  finite unions of languages of the form
  \[h^{-1}(\widehat K \times \widehat L) = \tau_{\cD_1}^{-1}\circ
    \eta^{-1}(\widehat K) \cap \pi^{-1}(\widehat L)\]
  for some $K\in \cV(X_{\cD_1})$ and $L \in \cD_0$. Since $\eta$ and
  $\pi$ are, respectively, dual to the embeddings $\cV(X_{\cD_1})
  \rightarrowtail \cl(X_{\cD_1}^*)$ and $\cD \rightarrowtail \cP((A
  \times 2)^*)$, it follows that
  \[h^{-1}(\widehat K \times \widehat L) = \tau_{\cD_1}^{-1}(K) \cap
    L.\]
  Thus, by Definition~\ref{sec:recogn-form-obta-1} of $\cV(X_{\cD_1})
  \vcirc \cD_1$, the languages recognized by~$h$ are precisely the
  lattice combinations of languages of $\cV(X_{\cD_1}) \vcirc \cD_1$
  and of~$\cD_0$.
\end{proof}

In particular, since the set of languages recognized by a \bim{} via a
fixed morphism is a Boolean algebra closed under quotients
(cf. Section~\ref{sec:1}), it follows that the lattice $\cV \circ \cD$
is, in fact, a Boolean algebra closed under quotients.

\begin{corollary}\label{c:10}
  Let $\cV$ be an lp-variety of languages and $\cD \subseteq \cP((A
  \times 2)^*)$ be a Boolean subalgebra closed under quotients that is
  generated by the union of its retracts $\cD_0 \subseteq \cP(A^*)$
  and $\cD_1 \subseteq \cP(A^* \otimes \N)$. Then, the lattice $\cV
  \circ \cD$ is a Boolean algebra closed under quotients.
\end{corollary}

Finally, for the reader who is familiar with Almeida and Weil's
work~\cite{AlmeidaWeil1995}, we sketch here the relationship between
their Decomposition Theorem for semidirect
products~\cite[Theorem~5.1]{AlmeidaWeil1995} and our
Theorem~\ref{t:2}, thereby exhibiting in which sense our result may be
seen as a generalization of their. We start by
stating~\cite[Theorem~5.1]{AlmeidaWeil1995} in a language that is
closer to that used in this paper. Given pseudovarieties $\bf V$ and
$\bf W$ of (finite) monoids, we will denote by $\cV$ and $\cW$,
respectively, the varieties of (regular) languages determined by~$\bf
V$ and $\bf W$ under \emph{Eilenberg's correspondence}
(see~\cite[Chapter~VII, Section~3]{Eilenberg76}). It follows
from~\cite[Theorem~3.6.1]{Almeida94} that, for any pseudovariety~$\bf
V$, the underlying topological space of the free $A$-generated
pro-$\V$ monoid $ \overline{\Omega}_A{\bf V}$ is homeomorphic to the
Stone dual $X_{\cV(A)}$ of $\cV(A)$, and the \emph{canonical
  embedding} $\iota: A \rightarrowtail \overline{\Omega}_A{\bf V}$ is
a restriction of the continuous quotient $\pi:\beta(A^*)
\twoheadrightarrow X_{\cV(A)}$ dual to the embedding $\cV(A)
\rightarrowtail \cP(A^*)$. Moreover, the co-restriction $A^*
\twoheadrightarrow \pi[A^*]$ of $\pi$ is the syntactic morphism
of~$\cV(A)$, and $A^* \twoheadrightarrow \pi[A^*] \rightarrowtail
X_{\cV(A)}$ its syntactic \bim{}-stamp
(cf. Section~\ref{sec:1}). Finally, we will say that a homomorphism
$\theta:A^* \to Z$ into a profinite monoid $Z$ recognizes the language
$L \subseteq A^*$ provided $L = \theta^{-1}(V)$ for some clopen subset
$V \subseteq Z$. We may then state Almeida and Weil's Decomposition
Theorem for semidirect products as follows:

\begin{theorem}\label{t:3}
  Let $\bf V$ and $\bf W$ be two pseudovarieties of finite monoids,
  and $A$ be a finite alphabet.\footnote{We remark that their result
    is valid for $A$ profinite, but for the sake of simplicity we will
    consider the case where $A$ is finite.} We write $Y = X_{\cW(A)}
  \times A \times X_{\cW(A)}$ and denote by $\eta': Y^* \to
  X_{\cV(Y)}$ the restriction of the dual of the embedding $\cV(Y)
  \rightarrowtail \cl(Y^*)$. Then, the right and left actions of the
  profinite monoid $X_{\cW(A)}$ on the set $Y$ given by $x \cdot (x_1,
  a, x_2) = (xx_1, a, x_2)$ and $(x_1, a, x_2) \cdot x = (x_1, a,
  x_2x)$ can be extended to a continuous biaction of $X_{\cW(A)}$ on
  $X_{\cV(Y)}$, which satisfies
  \[x \cdot \eta'(y_1 \dots y_n)\cdot x' = \eta'((xy_1x') \dots
    (xy_nx')),\]
  for every $x, x' \in X_{\cD_1}$ and $y_1 \dots y_n \in Y^*$. The
  resulting two-sided semidirect product $X_{\cV(Y)} {**} X_{\cW(A)}$
  is a profinite monoid (for the product topology of its underlying
  space $X_{\cV(Y)} \times X_{\cW(A)}$).

  Moreover, if $\pi': A^* \to X_{\cW(A)}$ is the restriction of the
  dual of the embedding $\cW(A) \rightarrowtail \cP(A^*)$, then
  letting $\theta(a) = (\eta'(\pi'(1),a,\pi'(1)), \pi'(a))$ defines a
  unique monoid homomorphism $\theta: A^* \to X_{\cV(Y)} {**}
  X_{\cW(A)}$ that recognizes exactly the lattice combinations of
  languages recognized by a two-sided semidirect product $M{**}N$ with
  $M \in \bf V$ and $N \in \bf W$.
\end{theorem}

Now, given pseudovarieties of monoids $\bf V$ and $\bf W$, we take for
$\cD_0 \subseteq \cP(A^*)$ the Boolean algebra $\cW(A)$ and we let
$\pi_0: A^* \to X_{\cD_0}$ be the restriction of the dual of $\cD_0
\rightarrowtail \cP(A^*)$. In order to suitably define~$\cD_1$, we
will bear in mind the isomorphism $\xi: A^* \otimes \N \to A^* \times
A \times A^*$, which identifies a marked word $(w, i)$ with the triple
$(a_0 \dots a_{i-1}, a_i, a_{i+1}\dots a_{n-1})$, for $w = a_0 \dots
a_{n-1}$. Accordingly, we let $\widetilde \pi_0: A^* \times A \times
A^* \to X_{\cD_0} \times A \times X_{\cD_0}$ be defined by $\widetilde
\pi_0(u, a, v) = (\pi_0(u), a, \pi_0(v))$ and $\widetilde\pi_1: A^*
\otimes \N \to X_{\cD_0} \times A \times X_{\cD_0}$ be defined by
$\widetilde\pi_1(w, i) = \widetilde \pi_0(a_0 \dots a_{i-1}, a_i,
a_{i+1}\dots a_{n-1})$, for $w = a_0 \dots a_{n-1}$. Finally, $\cD_1
\rightarrowtail \cP(A^* \otimes \N)$ is the embedding whose dual is
the co-restriction $\pi_1$ of $\widetilde\pi_1$ to the closure of its
image, so that $\widetilde \pi_1 = e \circ \pi_1$, where $e: X_{\cD_1}
\rightarrowtail X_{\cD_0} \times A \times X_{\cD_0}$ is the subspace
embedding. Equivalently, $\cD_1$ consists of the lattice combinations
of languages of the form $\xi^{-1}(L_1 \times \{a\} \times L_2)$, with
$L_1, L_2 \in \cD_0$ and $a \in A$. Using Corollary~\ref{c:11}, one
may then verify that the sublattice~$\cD$ of $\cP((A \times 2)^*)$
generated by $\cD_0 \cup \cD_1 \cup \{A_z\}$ is a Boolean algebra
closed under quotients and thus, $\cD_0$, $\cD_1$, and $\cD$ are as in
the statement of Theorem~\ref{t:2}, and the map $\pi: (A\times 2)^*
\to X_{\cD}$ restricts and co-restricts to $\pi_0$ and to
$\pi_1$. Also note that, by definition of~$Y$ in Theorem~\ref{t:3} and
of $\cD_0$ above, we have $Y = X_{\cD_0} \times A \times
X_{\cD_0}$. If $\eta$ and $\eta'$ are as in Theorems~\ref{t:2}
and~\ref{t:3}, respectively, then by Lemma~\ref{l:5}, we have a
commutative diagram
\begin{center}
  \begin{tikzpicture}[node distance = 20mm]
    \node (B) at (0,0) {$X_{\cD_1}^*$}; \node[right of = B] (A) {$X_{\cV(X_{\cD_1})}$};
    \node[below of = B, yshift = 5mm] (VB) {$Y^*$};
    \node[below of = A, yshift = 5mm] (VA) {$X_{\cV(Y)}$};
    \draw[->] (B) to node[ArrowNode] {$\eta$} (A); \draw[>->] (B) to
    node[left, ArrowNode] {$e^*$} (VB); \draw[>->] (A) to node[right,
    ArrowNode] {$\widehat e$} (VA); \draw[->] (VB) to
    node[below,ArrowNode] {$\eta'$} (VA);
  \end{tikzpicture}
\end{center}
Moreover, it is easily seen that the biaction of $X_{\cW(X)} =
X_{\cD_0}$ on $X_{\cV(Y)}$ given by Theorem~\ref{t:3} is an extension
of the biaction of $M \subseteq X_{\cD_0}$ on $\eta[X_{\cD_1}^*]
\subseteq X_{\cV(Y)}$ given by Lemma~\ref{l:11}. In particular, the
map $(\widehat e \times id) \circ \iota:\eta[T^*]{**} M
\rightarrowtail X_{\cV(Y)}{**}X_{\cW(A)}$ is a submonoid embedding,
and the following diagram commutes:
\begin{center}
  \begin{tikzpicture}[node distance = 40mm]
    \node (X) at (0,0) {$A^*$}; \node[right of = X] (bX)
    {$\eta[T^*]{**}M$}; \node[below of = bX, yshift = 20mm] (Z)
    {$X_{\cV(Y)} {**} X_{\cW(A)}$}; \node[right of = bX] (A)
    {$X_{\cV(X_{\cD_1})} \times X_{\cD_0}$};
    \draw[->] (X) to node[above, ArrowNode] {$h$}(bX);
    \draw[->] (X) to node[left, yshift = -2mm, ArrowNode] {$\theta$}
    (Z); \draw[>->] (A) to node[right, ArrowNode] {$\;\;\widehat e \times
      id$} (Z); \draw[>->] (bX) to node[above, ArrowNode] {$\iota$}
    (A);
  \end{tikzpicture}
\end{center}
Thus, in order to conclude that Theorem~\ref{t:3} of Almeida and Weil
is, indeed, a consequence of our Theorem~\ref{t:2}, it suffices to
argue that $\cV \circ \cD$ consists of the lattice combinations of
languages recognized by a monoid of the form $M{**}N$, with $M \in
{\bf V}$ and $N \in {\bf W}$. In case $\bf V$ and $\bf W$ are
definable by suitable fragments of first order logic, that is the
content of~\cite[Lemma~3.8]{TessonTherien07}. The general case follows
standard arguments underlying the so-called \emph{block-product
  principle} for finite monoids (see e.g.~\cite[Section~6.3]{Pin1997}
or~\cite[Section~4.2]{StraubingTherien2002}). For the reader's
convenience, we sketch its proof below.

\begin{proposition}
  A language $L \subseteq A^*$ belongs to $\cV \circ \cD$ if and only
  if it is a lattice combination of languages over~$A$ recognized by a
  monoid of the form $M{**}N$, with $M \in {\bf V}$ and $N \in \bf W$.
\end{proposition}
\begin{proof}
  Let us denote by $(\cV{**}\cW)(A)$ the lattice generated by the
  languages over~$A$ recognized by a monoid of the form $M{**}N$, with
  $M \in {\bf V}$ and $N \in \bf W$. Also recall that $\cV\circ \cD =
  (\cV(X_{\cD_1})\vcirc X_{\cD_1}) \cup \cD_0$. Since $\cD_0 = \cW(A)$
  and $N$ is a submonoid of every semidirect product of the form
  $M{**}N$, it is clear that $\cD_0 \subseteq (\cV{**}\cW)(A)$. Let $L
  \in \cV(X_{\cD_1})\vcirc X_{\cD_1}$, say $L = \tau_{\cD_1}^{-1}(K)$,
  for some $K \in \cV(X_{\cD_1})$. Since~$X_{\cD_1}$ embeds in~$Y$ via
  $e$, by definition of variety, there exists a language $K_0 \in
  \cV(Y)$ such that $K = (e^*)^{-1}(K_0)$, and by Lemma~\ref{l:4},
  there is a finite continuous quotient $q: Y \twoheadrightarrow B$
  and a monoid homomorphism $g: B^* \to M$ into a monoid $M \in \bf V$
  such that $K_0 = (g \circ q^*)^{-1}(P)$, for some $P \subseteq
  M$. On the other hand, for each $b \in B$, $q^{-1}(b)$ is a clopen
  subset of $Y = X_{\cD_0} \times A \times X_{\cD_0}$ and, as so, it
  may be written as a finite union $\bigcup_{i = 1}^{k_b} (\widehat
  {L_{i,b}} \times \{a_{i, b}\} \times \widehat {L_{i,b}'})$, for some
  $L_{i,b}, L_{i, b}' \in \cD_0 = \cW(A)$ and $a_{i, b} \in A$. We let
  $h: A^* \to N$ be a monoid homomorphism into a monoid $N \in \bf W$
  that recognizes every language $L_{i,b}$ and $L_{i, b}'$, and we
  define $\widetilde h: A^* \times A \times A^* \to N \times A \times
  N$ by $\widetilde h (u, a, v) = (h(u), a, h(v))$. Note that, in
  particular, every clopen subset in the image of $q^{-1}$ is of the
  form $\widehat { \widetilde h^{-1}(Q)}$, for some $Q \subseteq N
  \times A \times N$. Therefore, the map $q \circ \widetilde \pi_0$
  factors through $\widetilde h$, say via $r: N \times A \times N \to
  B$ and, since $\widetilde \pi_0$ is dense, $q$ is surjective, and
  $B$ is finite, we have that $r$ is a quotient. Finally, we let
  $\tau_h: A^* \to (N \times A \times N)^*$ be given by $\tau_h(a_0
  \dots a_{n-1}) = (1, a_0, h(a_1 \dots a_{n-1}))(h(a_0), a_1, h(a_2
  \dots a_{n-1})) \dots (h(a_0 \dots a_{n-2}), a_{n-1}, 1)$. Routine
  computations show that the following diagram commutes:
  \begin{center}
    \begin{tikzpicture}[node distance = 30mm]
      \node (X) at (0,0) {$A^*$}; \node[right of = X, xshift = -10mm]
      (bX) {$X_{\cD_1}^*$}; \node[right of = X, xshift = 10mm] (Y)
      {$Y^*$}; \node[below of = bX, yshift = 15mm, xshift = 10mm] (Z)
      {$B^*$}; \node[right of = bX, xshift = 10mm] (A)
      {$B^*$};\node[below of = X, yshift =15mm] (N) {$(N \times A
        \times N)^*$}; \node[below of = A, yshift = 15mm] (M) {$M$};
      \draw[->] (X) to node[above, ArrowNode] {$\tau_{\cD_1}$}(bX);
      \draw[->] (X) to node[left, ArrowNode] {$\tau_h$} (N); \draw[->]
      (A) to node[right, ArrowNode] {$f$} (M); \draw[->>] (Y) to
      node[above, ArrowNode] {$q^*$} (A); \draw[>->] (bX) to
      node[above, ArrowNode] {$e^*$} (Y);\draw[->>] (N) to
      node[below, ArrowNode] {$r^*$} (Z); \draw[->] (Z) to node[below,
      ArrowNode] {$f$} (M);
    \end{tikzpicture}
  \end{center}
  In particular, we have $L = \tau_h^{-1}(f \circ r^*)^{-1}(P)$, which
  is a language recognized by the semi-direct product $M^{N \times
    N}{**}N$ (usually called \emph{block product} and denoted $M \Box
  N$), with $N$ biacting on $M^{N \times N}$ by
  \[n \cdot \alpha \cdot n': N \times N \to M, \quad (x, y) \mapsto
    \alpha(xn, n'y),\]
  for every $n,n' \in N$ and $\alpha \in M^{N\times N}$. Indeed, one
  may verify that
  \[L = g^{-1}(\{\alpha \in M^{N \times N} \mid \alpha(1,1) \in P \}
    \times N),\]
  where $g$ is the unique homomorphism $g: A^* \to M^{N \times
    N}{**}N$ satisfying $g(a) = (\alpha_a, h(a))$, with $\alpha_a: N
  \times N \to M$ given by $\alpha_a(x,y) = \alpha \circ r (x, a, y)$.

  Conversely, let $h: A^* \to M{**}N$ be a monoid homomorphism, with
  $M \in \bf V$ and $N \in \bf W$, and write $h(w) = (h_1(w), h_2(w))$
  (note that $h_2$ is a monoid homomorphism, but $h_1$ may not
  be). Since $\cV \circ \cD$ is a lattice, and the languages
  recognized by $h$ are all lattice combinations of languages of the
  form $h^{-1}(P \times N)$ and $h^{-1}(M \times Q)$, with $P
  \subseteq M$ and $Q \subseteq N$, it suffices to show that the
  latter belong to $\cV \circ \cD$. It is clear that every language of
  the form $h^{-1}(M \times Q)$ belongs to $\cD_0 = \cW(A)$, as it is
  recognized by~$N$. We will now argue that the language $L = h^{-1}(P
  \times N)$ belongs to $\cV(X_{\cD_1}) \vcirc \cD_1$. Since $N \in
  \bf W$, we have that $h_2$ factors through $\pi_0$ via a continuous
  map, say $h_2 = r \circ \pi_0$, with $r: X_{\cD_0} \to N$
  continuous. We further let $\widetilde r: X_{\cD_1} \to N \times A
  \times N$ be the continuous function defined by $\widetilde r (x) =
  (r \times id \times r) \circ e(x)$, for every $x \in X_{\cD_1}$, and
  $g: (N \times A \times N)^* \to M$ be the unique monoid homomorphism
  satisfying $g(n, a, n') = n \cdot h_1(a) \cdot n'$, where $(\_)
  \cdot n$ and $n \cdot (\_)$ denote, respectively, the right and left
  actions of $n$ on $M$. Finally, if $q: M{**} N \twoheadrightarrow M$
  is the projection onto~$M$, one may check that the following diagram
  commutes:
  \begin{center}
    \begin{tikzpicture}[node distance = 30mm]
      \node (X) at (0,0) {$A^*$}; \node[right of = X] (bX) {$M {**}
        N$}; \node[right of = bX] (A) {$M$};\node[below of = X, yshift
      =15mm] (N) {$X_{\cD_1}^*$}; \node[below of = A, yshift = 15mm]
      (M) {$(N \times A \times N)^*$};
      \draw[->] (X) to node[above, ArrowNode] {$h$}(bX); \draw[->] (X)
      to node[left, ArrowNode] {$\tau_{\cD_1}$} (N); \draw[<-] (A) to
      node[right, ArrowNode] {$g$} (M); \draw[->>] (bX) to node[above,
      ArrowNode] {$q$} (A);\draw[->>] (N) to node[below, ArrowNode]
      {$\widetilde r^{\,*}$} (M);
    \end{tikzpicture}
  \end{center}
  In particular, we have $L = \tau_{\cD_1}^{-1}(K)$, where $K = (g
  \circ \widetilde r^{\,*})^{-1}(P)$ is a language of
  $\cV(X_{\cD_1})$. This shows that $L$ belongs to $\cV(X_{\cD_1}
  \vcirc \cD_1)$, as required.
\end{proof}

\section{Examples}\label{sec:examples}  We
finish the paper with the computation of the \bim{} $\rvd =
(\eta[T^*]{**}M, \, \iota, \, X_{\cV(X_{\cD_1})} \times X_{\cD_0})$
for an arbitrary Boolean subalgebra $\cD \subseteq \cP((A \times
2)^*)$ closed under quotients that is generated by the union of its
retracts $\cD_0 \subseteq \cP(A^*)$ and $\cD_1 \subseteq \cP(A^*
\otimes \N)$, and three specific lp-varieties~$\cV$ (recall the
notation from Table~\ref{tab:my-table}). The lp-varieties considered
are naturally defined by classes of sentences of the
form~$\Gamma_\cQ$, as in Lemma~\ref{l:2}, for some set of
quantifiers~$\cQ$. In particular, by Theorem~\ref{t:2}, it follows
that if $\varphi(x)$ is a formula defining a language in~$\cD$ then,
for every quantifier $Q \in \cQ$, the language $Qx \ \varphi(x)$ is
recognized by $\rvd$.  In Section~\ref{sec:2}, we consider the case
where $\cQ = \{\exists\}$ consists of the \emph{existential
  quantifier}, in Section~\ref{sec:4} we consider the set of
\emph{modular quantifiers} $\cQ = \{\exists^r_q \mid r \in \{0, \dots,
q-1\}\}$ for some positive integer $r$, and in Section~\ref{sec:5} we
consider the set of \emph{majority quantifiers} $\cQ = \{{\sf Maj}_k
\mid k \in \mathbb Z_{\geq 0}\}$. It is easy to verify that the
lp-strain of languages obtained from each of these choices of $\cQ$ is
indeed an lp-variety, so that Theorem~\ref{t:2} does apply.

In what follows, for every set $X$ and every monoid $(N, +)$, we will
use $[X \to N]_{fin}$ to denote the set of functions $f: X \to N$ such
that $f(x) = 0$ for all but finitely many $x \in X$. Note that, when
equipped with pointwise addition, $[X \to N]_{fin}$ is also a monoid
whose identity is the constant function equal to the identity
of~$N$. If $w$ is a word over some alphabet (finite or profinite),
then we will use $c(w)$ to denote the set of letters that occur
in~$w$, that is, $c(w)$ denotes the \emph{content of~$w$}. Finally,
for a subset $S$ of the alphabet at hand, we use $\card w_S$ to denote
the number of occurrences of a letter of $S$ in $w$, and in the case
where $S = \{a\}$ is a singleton, we simply write $\card w_a$.

\subsection{Existential quantifier}\label{sec:2}
Given a finite alphabet $A$ and a subset $S \subseteq A$, we denote by
$L_S$ the language over~$A$ defined by the formula $\exists x \
\bigvee_{a \in S} P_a(x)$, that is, $L_S = A^* S A^* = \{w \in A^*
\mid c(w) \cap S \neq 0\}$. We let $\cV_\exists$ be defined by
\[A \mapsto \cV_\exists(A)= \langle{\{L_S \mid S \subseteq
    A\}}\rangle_{\sf BA},\]
for every finite alphabet~$A$. Then, $\cV_\exists$ is an lp-variety of
\emph{regular} languages. It is easy to see that, for each finite
alphabet $A$, the syntactic monoid of $\cV_\exists(A)$ is the free
join-semilattice generated by~$A$, that is, $(\cP(A), \cup)$ -- the
powerset of~$A$ equipped with union. Moreover, the syntactic morphism
of~$\cV_\exists(A)$ is given by
\[\eta_A: A^* \twoheadrightarrow \cP(A), \quad w \mapsto c(w) = \{a \in A
  \mid a \text{ occurs in }w\}.\]
Let $Y$ be a profinite alphabet. Since $\cP(A)$ is the dual space of
$\cV_\exists(A)$, it follows from the definition of $\cV_\exists(Y)$
that
\[X_{\cV_\exists(Y)} = \lim_{\longleftarrow} \ \{\cP(A) \mid Y
  \twoheadrightarrow A \text{ is a finite continuous quotient}\}.\]
In other words, $X_{\cV_\exists(Y)}$ is the free profinite
join-semilattice on~$Y$. It is a well-known fact that this is the
\emph{Vietoris space} of $Y$, that is, the topological space
$\viet(Y)$ consisting of the closed subsets of $Y$ and equipped with
the topology generated by the subsets of the form
\[\Diamond V = \{C \in \viet(Y) \mid C \cap V \neq \emptyset\}\qquad
  \text{and}\qquad \Box V= \{C \in \viet(Y) \mid C \subseteq V\},\]
where $V$ is a clopen subset of $Y$. Note that, since $(\Diamond V)^c
= \Box (V^c)$, the clopen subsets of $\viet(Y)$ are the Boolean
combinations of the subsets of the form $\Box V$, for $V \in
\cl(Y)$. We refer to \cite[Chapter~I, Section~17]{Kuratowski66} and
to~\cite{Michael51} for further reading on the Vietoris construction.

Our next goal is to compute the map $\eta: \beta(Y^*)
\twoheadrightarrow X_{\cV_\exists(Y)}$ dual to the embedding
$\cV_\exists(Y)\rightarrowtail \cl(Y^*)$. The following is an easy
consequence of Lemma~\ref{l:4}:
\begin{lemma}\label{l:15}
  Let $Y$ be a profinite alphabet. Then, the Boolean algebra
  $\cV_\exists(Y)$ is generated, as a lattice, by the languages of the
  form $V^*$ and $Y^*VY^*$, where $V \subseteq Y$ is a clopen subset.
\end{lemma}
\begin{proof}
  By Lemma~\ref{l:4}, we know that $\cV_\exists(Y)$ consists of the
  languages of the form $(\pi^*)^{-1}(K)$ for some finite continuous
  quotient $\pi: Y \twoheadrightarrow A$ and some language $K \in
  \cV_\exists(A)$. Since $\cV_\exists(A)$ is the Boolean algebra
  generated by the languages of the form $A^* SA^*$, with $S \subseteq
  A$, it follows that $\cV_\exists(Y)$ is the Boolean algebra
  generated by the languages of the form $(\pi^*)^{-1}(A^* S A^*) =
  Y^* \pi^{-1}(S) Y^*$. Finally, we observe that the clopen subsets
  of~$Y$ are precisely the subsets of the form $\pi^{-1}(S)$ for some
  finite continuous quotient $\pi: Y \twoheadrightarrow A$ and some
  subset $S \subseteq A$, and that $(Y^* V Y^*)^c = (V^c)^*$ for every
  $V \in \cl(Y)$. Thus, $\cV_\exists(Y)$ is generated, as a lattice,
  by the languages of the form $V^*$ and $Y^* V Y^*$, where $V \in
  \cl(Y)$.
\end{proof}

We may now describe $\eta: \beta(Y^*) \twoheadrightarrow \viet(Y)$.
\begin{proposition}\label{p:7}
  Let $Y$ be a profinite alphabet. Then, the dual of the embedding
  $\cV_\exists(Y) \rightarrowtail \cl(Y^*)$ is the map
  \[\eta: \beta(Y^*) \twoheadrightarrow \viet(Y), \quad \gamma \mapsto \bigcap \{V
    \in \cl(Y) \mid V^* \in \gamma\}.\]
\end{proposition}
\begin{proof}
  Since arbitrary intersections of clopen sets are closed, the map
  $\eta$ is well-defined. Using the duality between Boolean algebras
  and Boolean spaces, it suffices to show that $\eta$ is a continuous
  surjective map whose dual~$\eta^{-1}: \cl(\viet(Y)) \to \cl(Y^*)$
  has image $\cV_\exists(Y)$. Let us first prove that $\eta$ is
  surjective. Given $C \in \viet(Y)$, we let
  \[\cF_C = \{V^* \mid V \in \cl(Y), \, C \subseteq V\} \cup \{Y^*
    V Y^* \mid V \in \cl(Y), \, C \cap V \neq \emptyset\},\]
  and we argue that $\cF_C$ has the finite intersection property. Let
  $V_1, \dots, V_m, W_1, \dots, W_n \in \cl(Y)$ be such that $C
  \subseteq V_i$ ($i = 1, \dots, m$) and $C \cap W_j \neq \emptyset$
  ($j = 1, \dots, n$). If $y_j \in C \cap W_j$ for every $j \in \{1,
  \dots, n\}$, then the word $w = y_1 \dots y_n$ belongs to
  $\bigcap_{i = 1}^m V_i^* \cap \bigcap_{j = 1}^n Y^*W_jY^*$, and
  thus, this is a nonempty subset of~$Y^*$ as required. Therefore,
  there is an ultrafilter $\gamma_C \in \beta(Y^*)$ extending
  $\cF_C$. Since such an ultrafilter must satisfy
  \[C \subseteq V \iff V^* \in \gamma_C,\]
  for every $V \in \cl(Y)$, it follows that $\eta(\gamma_C) = C$ and
  this shows surjectivity of~$\eta$. Now, if $W \in \cl(Y)$ and
  $\gamma \in \beta(Y^*)$, we have
  \begin{align*}
    \gamma \in \eta^{-1}(\Box W)
    & \iff \bigcap\{V \in \cl(Y) \mid V^* \in \gamma\} \subseteq
      W\qquad \text{(by definition of~$\eta$)}
    \\ & \iff W^c \subseteq \bigcup \{V^c \mid V \in \cl(Y), \, V^*
         \in \gamma\}
    \\ & \iff \exists V_1, \dots, V_n \in \cl(Y) \colon V_i^* \in
         \gamma \text{ and } W^c \subseteq \bigcup_{i = 1}^n V_i^c
    \\ & \hspace{7cm} \text{(because $W^c$ is a compact subset
         of~$Y$)}
    \\ & \iff \exists V \in \cl(Y) \colon V^* \in
         \gamma \text{ and } V \subseteq W
    \\ & \hspace{3cm}   \text{(because
         $\gamma$ is closed under finite intersections and $\bigcap_{i
         = 1}^n V_i^* = (\bigcap_{i = 1}^n V_i)^*$)}
    \\ & \iff W^* \in \gamma \qquad \text{(because $\gamma$ is upward
         closed)}
    \\ & \iff \gamma \in \widehat {W^*}.
  \end{align*}
  That is, we have $\eta^{-1}(\Box W) = \widehat {W^*}$
  for every $W \in \cl(Y)$. This proves both continuity of~$\eta$ and
  the equality $\eta^{-1}(\cl(\viet(Y))) = \cV_\exists(Y)$ as intended.
\end{proof}
We now make explicit each item in Table~\ref{tab:my-table} for this
case. Note that, the restriction of $\eta$ to $Y^*$ is the map
\[\eta_Y: Y^* \to \viet(Y), \quad w \mapsto c(w) = \{y \in Y \mid y \text{
    occurs in } w\},\]
whose image is $\cP_{fin}(Y)$ -- the family of finite subsets of~$Y$.
Indeed, that is because, for every word $w \in Y^*$, if $\gamma_w =
\{K \in \cl(Y^*) \mid w \in K\}$ is the associated principal
ultrafilter and $V$ a clopen subset of~$Y$, then we have
\[V^* \in \gamma_w \iff w \in V^* \iff c(w) \subseteq V.\]

Now, if $\cD \subseteq \cP((A \times 2)^*)$ is a Boolean subalgebra
closed under quotients, $\pi: (A \times 2)^* \to X_\cD$ is the dual of
the embedding $\cD \rightarrowtail \cP((A \times 2)^*)$, and $T=
\pi[A^* \otimes \N]$ and $M= \pi[A^*]$ are as before, we have
$\eta[T^*] = \cP_{fin}(T) \subseteq \cP_{fin}(X_{\cD_1})$. By a
straightforward application of Lemma~\ref{l:11}, the biaction of~$M$
on $\cP_{fin}(T)$ is given by
\[\lambda_m(P) = \{m t \mid t \in P\}\quad \text{and}\quad \rho_m(P) =
  \{tm\mid t \in P\},\]
for every $P \in \cP_{fin}(T)$ and $m \in M$. In particular, the
monoid multiplication of $\eta[T^*]{**}M = \cP_{fin}(T){**}M$ is given
by
\[(P_1, m_1) (P_2, m_2) = (\{t_1m_2 \mid t_1 \in P_1\} \cup \{m_1t_2
  \mid t_2 \in P_2\}, m_1m_2),\]
for every $(P_1, m_1), (P_2, m_2) \in \cP_{fin}(T){**}M$.  Although
this suffices to characterize the \bim{}
\[{\bf R}_{\cV_{\exists} \circ \cD}= (\eta[T^*]{**}M, \iota,
  X_{\cV_\exists(X_{\cD_1})} \times X_{\cD_0}) = (\cP_{fin}(T){**}M ,
  \iota, \viet(X_{\cD_1}) \times X_{\cD_0}),\]
we may use~\eqref{eq:23} and~\eqref{eq:30} to explicitly compute the
biaction of $\viet(X_{\cD_1}) \times X_{\cD_0}$ on
$\cP_{fin}(T){**}M$. We start by computing the first coordinate
of~\eqref{eq:23} and~\eqref{eq:30}. Given a closed subset $C \subseteq
Y$, we let $\gamma_C \in \beta(Y^*)$ satisfy
\begin{equation}
  C \subseteq V \iff V^* \in \gamma_C,\label{eq:31}
\end{equation}
for every $V \in \cl(Y)$. Note that, by the proof of
Proposition~\ref{p:7}, such an ultrafilter always exists and is such
that $\eta(\gamma_C) = C$.
\begin{lemma}\label{l:7}
  For every $C \in \viet(X_{\cD_1})$, $w \in X_{\cD_1}^*$, and $m \in
  M$, the following equalities hold:
  \[\eta \circ \widetilde\ell_w\circ \beta(\lambda_m^*) (\gamma_C) =
    \lambda_m[C] \cup c(w)\qquad \text{and}\qquad \eta \circ
    \widetilde r_w\circ \beta(\rho_m^*) (\gamma_C) =
    \rho_m[C]\cup c(w).\]
\end{lemma}
\begin{proof}
  By definition of~$\eta$, the closed set $\eta \circ
  \widetilde\ell_w\circ \beta(\lambda_m^*) (\gamma_C)$ consists of
  those points $x \in X_{\cD_1}$ that belong to every clopen subset $V
  \subseteq X_{\cD_1}$ satisfying $V^* \in \widetilde\ell_w\circ
  \beta(\lambda_m^*) (\gamma_C)$. Since $\widetilde\ell_w$ and
  $\beta(\lambda_m^*)$ are, respectively, the duals of the
  homomorphisms $\ell_w$ (recall~\eqref{eq:12}) and
  $(\lambda_m^*)^{-1}$, we have that $V^* \in \widetilde\ell_w\circ
  \beta(\lambda_m^*) (\gamma_C)$ if and only if $(\lambda_m^*)^{-1}
  \circ \ell_w(V^*) \in \gamma_C$. Now, it is easy to see that
  \[(\lambda_m^*)^{-1} \circ \ell_w(V^*) =
    \begin{cases}
      (\lambda_m^{-1}(V))^*, &\text{if $w \in V^*$}; \\\emptyset,
      &\text{else}.
    \end{cases}
  \]
  Therefore, using~\eqref{eq:31}, we may conclude that $\eta \circ
  \widetilde\ell_w\circ \beta(\lambda_m^*) (\gamma_C)$ consists of
  those points $x \in X_{\cD_1}$ that belong to every clopen subset $V
  \subseteq X_{\cD_1}$ satisfying $w \in V^*$ (that is, $c(w)\subseteq
  V$) and $C \subseteq \lambda_m^{-1}(V)$ (that is, $\lambda_m[C]
  \subseteq V$). Since $c(w)$ and $\lambda_m[C]$ are both closed
  subsets of $Y$, it follows that
  \[\eta \circ \widetilde\ell_w\circ \beta(\lambda_m^*) (\gamma_C) =
    \lambda_m[C] \cup c(w)\]
  as required. The second equality of the statement is proved in a
  similar manner.
\end{proof}
Now, given $(C, x_0) \in \viet(X_{\cD_1}) \times X_{\cD_0}$ and $(P,
m) \in \cP_{fin}(T) {**} M$, we let $\underline t \in T^*$ be such
that $\eta(\underline t) = P$. Using~\eqref{eq:23}, \eqref{eq:30}, and
Lemma~\ref{l:7}, we have that the biaction of $\cP_{fin}(T){**}M$ on
$\viet(X_{\cD_1}) \times X_{\cD_0}$ is given by
\[\lambda_{(P, m)}(C, x_0) = (\lambda_m[C] \cup \{\lambda_t(x_0) \mid
  t \in P\}, \lambda_m(x_0))\]
and
\[\rho_{(P, m)} (C, x_0) = (\rho_m[C] \cup \{\rho_t(x_0) \mid t \in
  P\}, \rho_m(x_0)).\]
Finally, comparing our construction with that
of~\cite{GehrkePetrisanReggio16}, we have that the \bim{} ${\bf
  R}_{\cV_{\exists} \circ \cD}$ just described properly embeds in the
\emph{unary Sch\"utzenberger product} of $(M_\cD \rightarrowtail
X_\cD)$ introduced
in~\cite[Definition~11]{GehrkePetrisanReggio16}. Indeed, since we are
assuming that $\cD \subseteq \cP((A \times 2)^*)$ is a Boolean
subalgebra closed under quotients that is generated by the union of
its retracts $\cD_0$ and $\cD_1$, we have that $\cD$ is isomorphic to
the product $\cD_0 \times \cD_1 \times \cD_z$
(cf. Section~\ref{sec:3}). Dually, it follows that $X_\cD$ may be seen
as the disjoint union of its closed subspaces $X_{\cD_0}$,
$X_{\cD_1}$, and $X_{\cD_z}$. In particular, $X_{\cD_0}$ and
$X_{\cD_1}$ are properly contained in $X_\cD$. On the other hand, we
recall that the space component of the unary Sch\"utzenberger product
of $(M_\cD \rightarrowtail X_\cD)$ is $\viet(X_\cD) \times X_\cD $,
while that of the ${\bf R}_{\cV_{\exists} \circ \cD}$ is
$\viet(X_{\cD_1}) \times X_{\cD_0}$. Thus, it is clear that the first
one properly contains the latter. A similar analysis may be done for
the monoid component of each of the \bim{}s involved.

\subsection{Modular quantifiers}\label{sec:4}
Before describing the elements of Table~\ref{tab:my-table} for the
lp-variety defined by the modular quantifiers, we remark that the
existential quantifier considered in the previous section is defined
by the two-element Boolean algebra ${\bf 2}$ in the following sense:
for every word $w$ and every formula $\varphi(x)$, we have
\[w \models \exists x \ \varphi(x)\iff \bigvee_{i = 0}^{\card w -1}
  \chi_{ \varphi(x)}(w, i) = 1,\]
where the right-hand side is an equality in $\bf 2$ and
$\chi_{\varphi(x)}(w, i) = 1$ if and only if $(w, i) \models
\varphi(x)$.
In~\cite{GehrkePetrisanReggio17,GehrkePetrisanReggio21}, the authors
extended the results of~\cite{GehrkePetrisanReggio16} to a broader
class of quantifiers defined by finite semirings in a similar way: if
$(S, +, \cdot\, , 0, 1)$ is a semiring and $r \in S$, then $Q_S^r$ is
defined by
\[w \models Q_S^r x \ \varphi(x) \iff \sum_{i = 0}^{\card w - 1}
  \chi_{ \varphi(x)}(w, i) = r,\]
where $\chi_{\varphi(x)}(w, i)$ is $1$ if $(w, i) \models
\varphi(x)$ and is~$0$ otherwise.  When $S = {\mathbb Z}_q$ is the
semiring of integers modulo $q$, the quantifier $Q_{q}^r$ is nothing
but the modular quantifier $\exists_q^r$.
Unlike in~\cite{GehrkePetrisanReggio17,GehrkePetrisanReggio21}, where
the full semiring structure of $\mathbb Z_q$ is needed, the reader
should notice that we will only make use of the underlying additive
structure of $\mathbb Z_q$.

We fix a positive integer $q$, and we let $\cV_{{\sf mod} \, q}$
denote the lp-variety given by
\[A \mapsto \cV_{{\sf mod} \, q}(A) = \langle \{L_{S, r} \mid S
  \subseteq A, \ r \in \{0, \dots, q-1\}\}\rangle_{\sf BA},\]
where $L_{S,r}$ denotes the language defined by the formula
$\exists_q^r\, x \ \bigvee_{a \in S} P_a(x)$, that is, $L_{S, r} = \{w
\in A^* \mid \card w_S \equiv r\ {\rm mod} \ q\}$. Once again, we have
that $\cV_{{\sf mod}\, q}$ is an lp-variety of \emph{regular}
languages.  Note that, since for every $S \subseteq A$ and $r \in \{0,
\dots, q-1\}$, the equality $L_{S,r}^c = \bigcup \{L_{S, r'} \mid r'
\in \{0, \dots, q-1\}, \, r' \neq r\}$ holds, the Boolean algebra
$\cV_{{\sf mod} \, q}(A)$ is generated, as a lattice, by the languages
of the form $L_{S, r}$.  It is easy to see that, for each finite
alphabet $A$, $\mathbb Z_q^A$ equipped with pointwise addition is the
syntactic monoid of $\cV_{{\sf mod} \, q}(A)$ and
\[\eta_A: A^* \twoheadrightarrow \mathbb \mathbb Z_q^A, \quad w \mapsto
  (\card w_a {\rm mod }\ q)_{a \in A}\]
is its syntactic morphism.  In~\cite{Reggio2020} a
computation of the space
\[X_{\cV_{{\sf mod} \, q}(Y)} =\lim_{\longleftarrow} \ \{ \mathbb
  Z_q^A\mid A \text{ is a finite continuous quotient of }Y\},\]
is provided and the Boolean algebra~$\cV_{{\sf mod} \, q}(Y)$ is
described (cf. Lemmas~3.5 and~4.1, respectively). We will now proceed
with the description of ${\bf R}_{\cV_{{\sf mod} \, q} \circ \cD}$
without using these results.

\begin{lemma}\label{l:16}
  Let $Y$ be a profinite alphabet. Then, the Boolean algebra
  $\cV_{{\sf mod} \, q}(Y)$ is generated, as a lattice, by the
  languages of the form $L_{V, r} = \{w \in Y^* \mid \card w_V \equiv
  r \ {\rm mod}\ q\}$, where $V \subseteq Y$ is a clopen subset and $r
  \in \{0, \dots, q-1\}$.
\end{lemma}
\begin{proof}
  Since each $\cV_{{\sf mod} \, q}(A)$ is generated, as a lattice, by
  the languages of the form $L_{S, r}$, with $S \subseteq A$ and $r
  \in \{0, \dots, q-1\}$, by Lemma~\ref{l:4}, the languages of
  $\cV_{{\sf mod} \, q}(Y)$ are the lattice combinations of languages
  of the form $(\pi^*)^{-1}(L_{S, r})$ for some finite continuous
  quotient $\pi: Y \twoheadrightarrow A$, and some $S \subseteq A$ and
  $r \in \{0, \dots, q-1\}$.  The claim follows from observing that
  \[(\pi^*)^{-1}(L_{S, r}) = \{w \in Y^* \mid \card w_{\pi^{-1}(S)}
    \equiv r \ {\rm mod}\ q\},\]
  and that the clopen subsets of~$Y$ are exactly the subsets of the
  form $\pi^{-1}(S)$.
\end{proof}

\begin{proposition}\label{p:9}
  Let $Y$ be a profinite alphabet. Then, $([Y\to \mathbb Z_q]_{fin},
  +)$ is the syntactic monoid of $\cV_{{\sf mod} \, q}(Y)$ and
  \[\eta_Y: Y^* \twoheadrightarrow [Y\to \mathbb Z_q]_{fin}, \quad
    w \mapsto (f_w: y \mapsto \card w_y {\rm mod }\ q)\]
  is its syntactic morphism. Moreover, $\cV_{{\sf mod} \, q}(Y)$ is
  isomorphic to the Boolean subalgebra $\cB_Y$ of $\cP([Y\to \mathbb
  Z_q]_{fin})$ generated by the subsets of the form
  \[\langle V, r\rangle = \{f \in [Y\to \mathbb Z_q]_{fin} \mid \sum
    \{f(y) \mid y \in V\} \equiv r \ {\rm mod } \ q \},\]
  where $V \in \cl(Y)$ and $r \in \mathbb Z_q$.
\end{proposition}
\begin{proof}  %
  First note that $\eta_Y$ is a well-defined surjective
  function. Then, proving that $\eta_Y$ is the syntactic morphism of
  $\cV_{{\sf mod} \, q}(Y)$ amounts to proving that $\ker(\eta_Y)$ is
  the syntactic congruence~$\sim$ of $\cV_{{\sf mod} \, q}(Y)$. By
  Lemma~\ref{l:16}, two words $u, v \in Y^*$ are $\sim$-equivalent if
  and only if $\card u_V \equiv \card v_V \ {\rm mod} \ q$, for every
  $V \in \cl(Y)$. Since, for every $y \in Y$, there is a clopen subset
  $V \subseteq Y$ such that $y \in V$ and $(V \setminus \{y\}) \cap
  (c(u) \cup c(v)) = \emptyset$, it follows that $u \sim v \iff \eta_Y
  (u) = \eta_Y(v)$ as required.

  The last claim follows from having $\eta_Y^{-1}(\langle V, r\rangle)
  = L_{V, r}$, for every $V \in \cl(Y)$ and $r \in \{0, \dots, q-1\}$,
  and from Lemma~\ref{l:16}.
\end{proof}

It remains to describe the monoid $\eta[T^*]{**}M$ (recall
Table~\ref{tab:my-table}). By Proposition~\ref{p:9}, we have that
$\eta[T^*] = [T\to \mathbb Z_q]_{fin}$ is the submonoid of $[Y\to
\mathbb Z_q]_{fin}$ whose functions vanish in $Y \setminus T$. Using
Lemma~\ref{l:11}, we may then conclude that the biaction of $M$ on
$\eta[T^*] = [T\to \mathbb Z_q]_{fin}$ is defined by
\[\lambda_m(f): t \mapsto \sum \{f(t') \mid m t' = t\} \quad \text{and}
  \quad \rho_m(f): t \mapsto \sum \{f(t') \mid t'm = t\} ,\]
for every $f \in [T \to \mathbb Z_q]_{fin}$ and $m \in M$.  As
expected, our construction is tightly related to that
of~\cite{GehrkePetrisanReggio17,GehrkePetrisanReggio21}: the \bim{}
${\bf R}_{\cV_{{\sf mod} \, q} \circ \cD}$ embeds in the \bim{}
$(\Diamond_{\mathbb Z_q}(X_{\cD}), \Diamond_{\mathbb Z_q}(M_{\cD}))$
of~\cite{GehrkePetrisanReggio17,GehrkePetrisanReggio21}.
\subsection{Majority quantifiers}\label{sec:5} 
Our last example concerns the majority quantifiers. Although an
algebraic approach was undertaken in~\cite{KrebsLangeReifferscheid05},
it relies on an iterative application of a block product construction
and a simple description of a topo-algebraic recognizer for the
languages of $\cV_{\sf Maj}\circ \cD$ built on a recognizer for the
languages of $\cD$ is not known (here, $\cV_{\sf Maj}$ denotes the
lp-variety defined by the majority quantifiers). Therefore, this
example is a truly new application of the work of this paper.

Let $A$ be a finite alphabet. For each subset $S \subseteq A$ and for
each integer $k \in \mathbb Z$, we let $L_{S, k}$ denote the language
of all words $w \in A^*$ satisfying $\card w_S - \card w_{S^c} >
k$. Note that, if $k \geq 0$, then the language $L_{S, k}$ is defined
by the formula ${\sf Maj}_k x\ \bigvee_{a \in S}\tP_a(x)$. Then,
$\cV_{\sf Maj}$ is the lp-variety defined by
\[A \mapsto \cV_{\sf Maj}(A) = \langle \{L_{S, k} \mid S \subseteq A,
  \ k \in \mathbb N\}\rangle_{\sf BA}.\]
Since, for every $S \subseteq A$ and for every $k \in \mathbb Z$, we
have $L_{S, k}^c = L_{S^c, -(k+1)}$, it follows that $\cV_{\sf
  Maj}(A)$ is the lattice generated by $\{L_{S, k} \mid S \subseteq A,
\, k \in \mathbb Z\}$.

The proof of the following is omitted as it is very similar to the
proof of Lemma~\ref{l:16}.
\begin{lemma}\label{l:13}
  For every profinite alphabet $Y$, the Boolean algebra $\cV_{\sf
    Maj}(Y)$ is generated, as a lattice, by the languages of the form
  $L_{V, k}= \{w \in Y^* \mid \card w_V - \card w_{V^c} > k\}$, where
  $V \in\cl(Y)$ and $k \in \mathbb Z$.
\end{lemma}
Using this characterization we may then describe the syntactic
morphism of $\cV_{\sf Maj}(Y)$.
\begin{proposition}\label{p:8}
  Let $Y$ be a profinite alphabet. Then, $([Y\to \mathbb \N]_{fin},
  +)$ is the syntactic monoid of $\cV_{\sf Maj}(Y)$ and
  \[\eta_Y: Y^* \twoheadrightarrow [Y\to \mathbb N]_{fin}, \quad
    w \mapsto (f_w: y \mapsto \card w_y)\]
  is the syntactic morphism. Moreover, $\cV_{\sf Maj}(Y)$ is
  isomorphic to the Boolean subalgebra $\cB_Y$ of $\cP([Y\to \mathbb
  N]_{fin})$ generated by the subsets of the form
  \[\langle V, k\rangle = \{f \in [Y\to \mathbb N]_{fin} \mid \sum
    \{f(y) \mid y \in V\} - \sum \{f(y) \mid y \notin V\} > k \},\]
  where $V \in \cl(Y)$ and $k \in \mathbb N$.
\end{proposition}
\begin{proof}
  It is clear that $\eta_Y$ is a surjective well-defined function. To
  show that $\eta_Y$ is the syntactic morphism of $\cV_{\sf Maj}(Y)$
  we need to show that the syntactic congruence~$\sim$ of $\cV_{\sf
    Maj}(Y)$ coincides with $\ker(\eta_Y)$. By Lemma~\ref{l:13}, two
  words $u, v \in Y^*$ are $\sim$-equivalent if and only if, for every
  $V \in \cl(Y)$ and $k \in \mathbb Z$, the equivalence
  \begin{equation}
    \label{eq:32}
    \card u_V -  \card u_{V^c} > k \iff \card v_V - \card v_{V^c} > k
  \end{equation}
  holds.  In particular, we have the inclusion $\ker(\eta_Y) \subseteq
  {\sim}$. Conversely, suppose that $u, v \in Y^*$ are
  $\sim$-equivalent. We first note that $\card u = \card v$. Indeed,
  that is because, by taking $V = Y$ in~\eqref{eq:32}, we have
  \[\forall k \in \mathbb Z, \ \card u > k \iff \card v > k.\]
  We denote $n = \card u = \card v$. Now, given a letter $y \in Y$,
  we let $V \in \cl(Y)$ be such that $y \in V$ and $(V \setminus
  \{y\}) \cap (c(u) \cup c(v)) = \emptyset$. That is, $V$ is such that
  $\card u_{V} = \card u_y$ and $\card v_{V} = \card v_y$. Then, using
  the equalities $\card u_{V^c} = n - \card u_V$ and $\card v_{V^c} =
  n - \card v_V$, it follows from~\eqref{eq:32} that
  \[\forall k \in \mathbb Z, \ 2 \card u_Y - n > k \iff 2 \card v_y -
    n> k,\]
  which clearly implies that $\card u_y = \card v_y$ as required.

  The last assertion follows from Lemma~\ref{l:13} and from the
  equality $(\eta_Y)^{-1}(\langle V, k\rangle) = L_{V, k}$, which
  holds for every $V \in \cl(Y)$ and $k \in \mathbb Z$.
\end{proof}

Finally, we let $\cD \subseteq \cP((A \times 2)^*)$ be a Boolean
subalgebra closed under quotients that is generated by its retracts
$\cD_0$ and $\cD_1$. As before, we let $\pi: (A \times 2)^* \to X_\cD$
be dual to the embedding $\cD \rightarrowtail \cP((A \times 2)^*)$,
and we denote $T = \pi[A^* \otimes \N]$ and $M: = \pi[A^*]$. Then, by
Proposition~\ref{p:8}, we have that the underlying set of the monoid
component of ${\bf R}_{\cV_{\sf Maj} \circ \cD}$ is $[T \to \N]_{fin} \times M$, and by
Lemma~\ref{l:11}, the biaction of $M$ on $[T \to \mathbb \N]_{fin}$ is
defined by
\[\lambda_m(f): t \mapsto \sum \{f(t') \mid m t' = t\} \quad \text{and}
  \quad \rho_m(f): t \mapsto \sum \{f(t') \mid t'm = t\} ,\]
for every $f \in [T \to \mathbb N]_{fin}$ and $m \in M$. Therefore,
for every $(f_1, m_1), (f_2, m_2) \in [T \to \N]_{fin} {**} M$, we
have $(f_1, m_1)(f_2, m_2) = (f, m_1m_2)$, where $f(t) = \sum \{f_1(t')
\mid t'm_2 = t\} + \sum \{f_2(t') \mid mt' = t\}$ for every $t \in T$.

\paragraph{Competing interests.} The authors declare none.

\footnotesize

\end{document}